\documentclass[
 aps,
 pra,
 showpacs,
 showkeys,
 superscriptaddress,
 reprint,%
twocolumn,
]{revtex4-2}

\usepackage{physics}
\usepackage{mathtools}
\usepackage{amssymb, amsfonts, amsthm}
\usepackage{siunitx}
\usepackage[toc,page]{appendix}
\usepackage{bm}
\usepackage{soul,xcolor}
\usepackage{bbm}
\usepackage{graphicx}
\usepackage{subfigure}
\usepackage{url}
\usepackage{comment}
\usepackage{dutchcal}
\usepackage{ytableau}
\usepackage{fixmath}
\usepackage{tikz}
\usepackage{braket}
\usetikzlibrary{quantikz2}
\usepackage{hyperref}
\hypersetup{colorlinks=true, linkcolor=blue, citecolor=blue, urlcolor=blue, unicode=true}
\usepackage{cleveref}

\newtheorem{theorem}{Theorem}[section]
\newtheorem{lemma}[theorem]{Lemma}
\newtheorem{proposition}[theorem]{Proposition}
\newtheorem{corollary}[theorem]{Corollary}

\theoremstyle{definition}
\newtheorem{definition}[theorem]{Definition}

\theoremstyle{remark}

\newcommand{\SWAP}{SW\hspace{-1.2mm}AP}
\newcommand{\SoS}{SoS_1(\mathcal{Proj})}

\newcommand{\knote}[1]{\footnote{{\color{red} {\bf Kevin}: {#1}}}}
\newcommand{\lunote}[1]{\footnote{{\color{red} {\bf Lu}: {#1}}}}
\newcommand{\cnote}[1]{\footnote{{\color{red} {\bf Chaithanya}: {#1}}}}
\newcommand{\jnote}[1]{\footnote{{\color{red} {\bf Jun}: {#1}}}}
\newcommand{\onote}[1]{\footnote{{\color{blue} {\bf Ojas}: {#1}}}}

\renewcommand{\knote}[1]{}
\renewcommand{\lunote}[1]{}
\renewcommand{\cnote}[1]{}
\renewcommand{\jnote}[1]{}
\renewcommand{\onote}[1]{}
\renewcommand\footnotemark{}

\begin{document}

\title{An SU(2)-symmetric Semidefinite Programming Hierarchy for Quantum Max Cut}




\newcommand{\UNM}{Department of Physics and Astronomy and Center for Quantum Information and Control, University of New Mexico, Albuquerque, New Mexico 87131, USA}
\newcommand{\Caltech}{Department of Computing and Mathematical Sciences, Caltech, Pasadena, CA, USA}
\newcommand{\Sandia}{Sandia National Laboratories, Albuquerque, NM, USA}
\newcommand{\ISSP}{The Institute for Solid State Physics, The University of Tokyo, Kashiwa, Chiba, Japan}
\newcommand{\CSIQUDS}{Department of Computer Science \& Institut quantique, Universit\'e de Sherbrooke, QC, Canada}

\author{Jun Takahashi}
\affiliation{\UNM}
\affiliation{\ISSP}

\author{Chaithanya Rayudu}
\affiliation{\UNM}
\author{Cunlu Zhou}
\affiliation{\UNM}
\affiliation{\CSIQUDS}

\author{Robbie King}
\affiliation{\Caltech}

\author{Kevin Thompson}
\author{Ojas Parekh}
\affiliation{\Sandia}

\date{\today}

\begin{abstract}
Understanding and approximating extremal energy states of local Hamiltonians is a central problem in quantum physics and complexity theory.
Recent work has focused on developing approximation algorithms for local Hamiltonians, and in particular the ``Quantum Max Cut'' ({\scshape QMaxCut}) problem, which is closely related to the antiferromagnetic Heisenberg model. 
In this work, we introduce a family of semidefinite programming (SDP) relaxations based on the Navascu{\'e}s-Pironio-Ac{\'i}n (NPA) hierarchy which is tailored for {\scshape QMaxCut} by taking into account its SU(2) symmetry. 
We show that the hierarchy converges to the optimal {\scshape QMaxCut} value at a finite level, which is based on a characterization of the algebra of SWAP operators.
We give several analytic proofs and computational results showing exactness/inexactness of our hierarchy at the lowest level on several important families of graphs.

We also discuss relationships between SDP approaches for {\scshape QMaxCut} and frustration-freeness in condensed matter physics and numerically demonstrate that the SDP-solvability practically becomes an efficiently-computable generalization of frustration-freeness.
Furthermore, by numerical demonstration we show the potential of SDP algorithms to perform as an approximate method to compute physical quantities and capture physical features of some Heisenberg-type statistical mechanics models even away from the frustration-free regions.

\end{abstract}

\maketitle

\section{Introduction}

The study of spin models plays a fundamental role both in physics and computer science. While 
most models are 
generally
too difficult to solve exactly \cite{bet31zur,lie62ord,pid17}, they provide insights into 
physical phenomena by 
serving as an effective description of condensed matter systems \cite{san10com,sac23qua}. 
The antiferromagnetic Heisenberg model has been well-studied in physics and forms the focus of a recent flurry of work in optimization \cite{gha19alm, par22opt, ans20bey, par21app}, with the goal of extending the rich field of approximation algorithms to quantum problems. 
Already as a classical spin model, the Ising model plays a central role in the intersection between statistical physics and combinatorial optimization. The problem of computing the ground state energy for an antiferromagnetic Ising model on an arbitrary graph is known to be equivalent to the classical ``Max Cut" ({\scshape MaxCut}) problem, one of the {\sf NP}-complete problems originally listed by Karp \cite{karp1972red}. While computing the ground state energy exactly is therefore hopeless in general, the celebrated Goemans-Williamson (GW) algorithm \cite{goe95imp} obtains an approximate solution 
which is optimal under the assumption of the {\it Unique Games Conjecture} \cite{khot07} and {\sf P}$\neq${\sf NP}. 

The {\it quantum} Max Cut ({\scshape QMaxCut}) problem is closely related to the antiferromagnetic quantum {\it Heisenberg} model 
and plays a crucial role in understanding the 
hardness of approximation of Local Hamiltonian Problems.  
Finding the ground state of the antiferromagnetic  Heisenberg model corresponds to finding the {\it maximum} energy state of {\scshape QMaxCut}, yet the complexity of approximating these problems likely differs.  Polynomial-time approximation algorithms with constant-factor guarantees are known for {\scshape QMaxCut} on arbitrary interaction graphs, while these are not expected to exist for the antiferromagnetic Heisenberg model.  The {\scshape QMaxCut} Hamiltonian was designed to bear similarity to {\scshape MaxCut}, and this has enabled new types of approximation algorithms for quantum local Hamiltonians that draw inspiration from classical approximations for {\scshape MaxCut} and other constraint satisfaction problems.
{\scshape QMaxCut} parallels {\scshape MaxCut} in that the decision version of the problem is known to be {\sf QMA}-complete \cite{pid17}, but unlike {\scshape MaxCut} the precise approximability of {\scshape QMaxCut} remains largely enigmatic. Recent works have been steadily improving the achievable approximation factor \cite{gha19alm, par22opt, ans20bey, par21app, kin22, lee22}, as well as conjecturing limitations on the achievable approximation factor \cite{hwa22uni}, but these upper and lower bounds have a sizable gap. In contrast, many combinatorial optimization problems are conjectured to be optimally approximated by techniques similar to the GW algorithm \cite{khot07}. 
In general 
techniques of this form can be regarded as the first order of a family of approximation algorithms derived from the Lasserre hierarchy. 

The Lasserrre hierarchy (or its dual, the Sum of Squares Hierarchy) is the tool of choice for many combinatorial optimization problems, with a well developed theory and practice (see e.g., \cite{lau09sum}). For a given problem this hierarchy corresponds to a set of semidefinite programs of increasing size and complexity (with increasing level).  At high level these hierarchies converge to the optimal solution of combinatorial optimization problems under fairly general assumptions \cite{las01glo, par03sem}, and at low level they relax the optimization problem.
The hierarchy has the benefit of providing an explicit proof that the objective achieved by the SDP bounds the optimal objective value (a ``sum-of-squares'' proof).  On the other hand, there are known limitations on these hierarchies, with very simple objective functions provably non-convergent until the SDP reaches exponential size \cite{gri01lin}.

Navascu{\'e}s, Pironio, and Ac{\'i}n (NPA) \cite{pir10con, nav08con} generalized Lasserre's construction \cite{las01glo} to the quantum setting, producing a powerful tool for quantum information problems. Working directly with quantum states is infeasible on classical computers since they require exponential resources in space and time in general, so in many cases NPA and similar hierarchies provide new avenues for understanding quantum systems. Such hierarchies are also referred to as noncommutative or quantum Lasserre hierarchies.  Many authors have used hierarchies to characterize quantum correlations \cite{nav07bou, nav08con}, design entanglement witnesses \cite{bac17}, and probe questions in entangled games \cite{kem07ent,kem09uni,Bam15sum,joh16ext,wat18alg,cui20gen,ji22mip,nav07bou}. In quantum Chemistry \cite{maz07red,nak01var} it is generally referred to as the {\it variational $2$-RDM method} and is used to provide computational bounds on the electronic structure problem when the dimension is too large for direct computation. More generally, SDP relaxations have been used for studying quantum many-body problems in various settings \cite{kho21sca,hai20var,bau12low,bar12sol}. Our primary application of interest is using NPA for the local Hamiltonian problem, along the lines of a recent thrust of work in quantum optimization \cite{bra16pro, gha19alm, par21app, par22opt, hwa22uni, has22opt}.

The main difference between NPA-like hierarchies and the Lasserre hierarchy is that NPA relaxes optimization over non-commuting rather than commuting variables. One might expect that quantum optimization landscape would parallel the classical one and that largely the same techniques would be useful for a breadth of problems, however, it appears that quantum optimization is richer in many ways. There are not known techniques which apply to many different problems, and, contrary to the classical case, it is known that the simplest ``first order'' algorithm is {\it not} optimal \cnote{By saying not optimal, we might be underselling the nuances in the quantum setting a bit more than we might want to {\bf Jun}: I think this is about the Pauli NPA, so it's OK. But maybe people will be confused exactly in that way is your point?}\knote{I beefed it up a little} for {\scshape QMaxCut} \cite{par22opt}, which sharply contrasts with the case for {\scshape MaxCut} \cite{goe95imp}. 
Interestingly, it is unclear at this point what form the optimal algorithm should take or even if there is an optimal classical algorithm. Since QMA-hard problems have witnesses which are highly entangled, it is likely difficult to describe them and to determine what kind of quantum state/algorithm is best for the problem.  
Consequently, it is unclear what the best form of NPA is for {\scshape QMaxCut}, since NPA is defined using abstract non-commutative operators, and it could be that the optimal approximation algorithm takes advantage of a clever choice of the operators. 
The generic $2$-Local Hamiltonian problem generalizes many classical problems \cite{woc03}, including those which are inapproximable (with constant approximation factor) under ${\sf P}\neq {\sf NP}$ \cite{zuc06}, so it is reasonable to expect that the optimal approximation algorithm takes advantage of the specific family or Hamiltonians it is designed for. There is precedent in this direction in that symmetry has already been used to drastically reduce the size of SDP relaxations on both the quantum \cite{ioa21} and classical \cite{gat04} side.
One immediate inconvenience of the Pauli-based NPA hierarchy used in the past \cite{gha19alm, par21, bra19} is that the first level of the hierarchy fails to solve {\scshape QMaxCut} for the simplest types of nontrivial instances one can think of (star graphs \cite{par21app}). This again is in sharp contrast with the classical {\scshape MaxCut} case, since the GW algorithm solves {\it all} of the bipartite graph instances exactly, which includes the star graphs as the simplest subclass. 
So far, some works have focused on an NPA hierarchy based on using the Pauli operators as variables, as well as hinted at another kind of hierarchy using the anti-ferromagnetic local terms of the Hamiltonian as non-commuting variables in the optimization \cite{par22opt}.


\subsection{Our Contributions}

\onote{[Explain QMaxCut]}
\onote{[Summarize contributions]}
{\scshape QMaxCut} is the problem of solving for the largest eigenvalue of a class of instances of the $2$-Local Hamiltonian problem.  {\scshape QMaxCut} instances are parameterized by weighted graphs.  Given a vertex set $V$ and a function $w$ from pairs of vertices to $\mathbb{R}_{\geq 0}$ such a Hamiltonian is written as 
$$
H=\sum_{\substack{i, j \in V\\i<j}} w_{ij} \frac{\mathbb{I}-X_i X_j -Y_i Y_j -Z_i Z_j}{4},
$$
where $X_i$, $Y_i$, and $Z_i$ stand for Pauli matrices $X$, $Y$, and $Z$ on qubit $i$.  {\scshape QMaxCut} Hamiltonians are naturally invariant under conjugation by any local unitary transformation on all qubits, so Schur-Weyl duality implies that the optimal eigenstate lies in an irreducible representation of the symmetric group.\onote{[sentence relating SWAP and symmetric group]} Hence in defining NPA it is sensible to use permutation operators or, equivalently, polynomials in the $2$-local {\scshape QMaxCut} (Heisenberg) terms.  We demonstrate that an abstractly defined {\it operator program} has objective matching the extremal eigenvalue and that the objective of NPA defined using this operator program converges at some finite level to the optimal solution.  We learned upon completing the present work that this observation is already implicit in some results in representation theory \cite{pro76, pro21, lit34}.  However, a unique contribution of our work is an explicit and self-contained description of the SWAP operator program that is accessible to the broader communities such as quantum information and computer science, as well as additional context for its role in the local Hamiltonian problem.  We show that the (weaker) real valued version of the NPA hierarchy agrees with the standard one at level-$1$, while giving an explicit example which we (numerically) demonstrate separates the real and complex versions in general.  The real version is studied in many works \cite{par21app, ans20bey, par22opt} so we motivate these works while also providing evidence that they could be improved.  
The lowest level of our SDP family roughly corresponds to the {\it second} lowest level of the Pauli hierarchy used in those works, 
but achieves the same approximation factor with the lowest level. 
This simplification not only leads to a runtime speedup by an order of magnitude, but also can lead to deeper implications for fully understanding the complexity landscape of 
{\scshape QMaxCut}.

In the direction of improving existing algorithms, we give several new families of graphs where we demonstrate exactness/inexactness of our family of SDPs.  In existing approximation algorithms \cite{par21app, par22opt, kin22, lee22} a deep understanding of instances which the low level SDP gets correct is an integral part of the analysis (the so called ``star bound''), so it is possible that results established here could lead to approximation algorithms with better performance.  One particularly prominent example where we demonstrate SDP exactness is for weighted star graphs.  We are aware of an unpublished proof of this preceding our results \cite{per_comm}, but here we provide a different proof of this fact which gives a pleasing ``geometric'' interpretation of monogamy of entanglement inequalities in the context of NPA hierarchies.  
The weighted star bound seems likely to have many applications; here we demonstrate that it implies exactness for another family of graphs, the ``double star'' graphs.  
We complement this with many other classes of graphs where we can show exactness, some of which correspond to condensed matter physics models including the Majumdar Ghosh-model and the Shastry-Sutherland model.  Additionally, we provide two families (complete graphs with odd number of vertices, and ``crown graphs" with certain weights) of graphs where we can analytically prove looseness of NPA at the first level. In fact we are able to provide an analytic characterization of when low levels of the hierarchy are exact on crown graphs.


\knote{Deleted paragraph here}

Equipped with the new SDP family, we then provide extensive numerical results studying the exactness/inexactness of NPA at low levels. 
We first provide results for an exhaustive search among all possible unweighted graphs up to $8$ vertices, and then proceed to physically interesting cases with up to $60$ vertices. 
With the exhaustive search, we find no unifying features among examples where NPA is exact, and examples which are seemingly ``simple'' where the optimal SDP objective at low level is far from the extremal eigenvalue as well. For cycles we find that neither the first level of our SDP family or second level the hierarchy previously considered in \cite{par21app, par22opt} is exact at low levels in sharp contrast to MaxCut where the lowest level is exact on all even cycles and the second level is exact on all cycles\cite{bar14}\jnote{We don't actually have numerical evidence about Lv2 Proj, if I remember correctly. We couldn't go to large enough systems to test if it's actually failing for large cycles.}. 
It is impossible to rigorously certify that NPA achieves the optimal eigenvalue using purely numerics, since we have many cases where the optimal SDP objective is only different from the extremal eigenvalue in the 4th or 5th decimal place. 
We classify graphs according to how the error of the SDP optimal solution behaves as a function of the tolerance parameter for the SDP. This lets us confidently conclude from numerics, whether the NPA is giving the exact extremal eigenvalue or not. 
 In doing so, we are able to explicitly show separation of different NPA hierarchies, which is otherwise subtle. 

Moreover, we run numerical simulations on some condensed matter physics models, demonstrating that the the lowest level of our NPA hierarchy obtains exact ground states of ``frustration-free" quantum spin systems such as the Majumdar-Ghosh and Shastry-Sutherland models. We point out that this is a natural consequence from the connection between frustration-freeness and sum of squares proof, showing that the NPA hierarchy as a whole is essentially a generalization of the frustration-free notion. 

The salient feature of our numerical results is that the SDP seems to predict many important physical properties even on instances where it is not achieving the optimal eigenvalue.  For instance, in models with a phase transition, the SDP also appears to reflect that, by having a discontinuous optimal SDP objective as a function of the parameters. 
Additionally, the SDP obtains the correct decaying exponent for the correlation function as on the Heisenberg spin chain, even though there is strong evidence that it does not correctly predict the optimal energy.  This suggests the capability of SDP solutions to exhibit nontrivial long-range entangled features of a critical ground state to some extent.  Using ``pseudo entanglement'' to model quantum systems and predict their physical properties seems to be a relatively open and exciting research direction with only a few results known \cite{has22per}. Since simulating large quantum systems is intractable on classical computers, the NPA hierarchy provides the possibility of probing features of quantum systems using (non-physical) pseudo states on a classical computer which would be unobtainable otherwise. 
This type of numerical analysis is only possible with our projector-based NPA hierarchy, since with the Pauli-based NPA hierarchy, the matrix size for SDP grows faster. Although the scaling difference is theoretically only a constant factor, the largest computable system size being $\sim 60$ qubits rather than $\sim 20$ makes a practical difference in terms of how deeply we can actually probe their performance. Moreover, in the projector SDP formulation, most of the variables in the moment matrix are free variables, an important feature that can significantly improve the numerical efficiency of solving SDPs when implemented in SDP solvers like MOSEK \cite{mosek}. This difference has enabled us to conduct both the exhaustive search and probing statistical physics model of sizes beyond what is reachable with exact diagonalization.

\textbf{Note Added:}  \onote{[Our papers have some shared results but are larelgy complementary; our focus is more on understanding SDPs for appproximating or understanding ground states of local Hamiltonians.  More complicated or powerful formulations are possible [other paper], but it's important to understand ``simple'' formulations, because they're easier to analyze, SDPs are slow to solve, and it's good to understand their power.  Thus there is strong motivation to understand the power of the simplest SDPs that offer strong bounds.]}

After preparing this draft we became aware of an independent group of researchers with complementary results to ours \cite{other_guys}.  The two papers have different themes in that our paper is largely focused on understanding the performance of low level SDP relaxations for Quantum Max Cut problems, while \cite{other_guys} establishes a more sophisticated hierarchy and uses representation theory to analyze the extremal energies for certain Quantum Max Cut instances.  There is clearly value in understanding powerful SDP relaxations, but we argue that it is also important to understand ``simple'' formulations because they are easier to analyze and SDPs are generally practically slow to solve.  Thus there is a strong motivation to understand the simplest SDPs which offer strong bounds.  We establish numerically and analytically that low levels of the hierarchy are exact on certain families of graphs with an eye toward solid state physics and approximation algorithms, while \cite{other_guys} is able to calculate the exact extremal eigenvalue for Hamiltonians which have a {\it signed clique decomposition}.  Here the two papers are very different in that we focus on the SDP solution rather than the exact solution for the Hamiltonian problem. \cite{other_guys} also investigates non-commutative Groebner bases for the $\SWAP$ algebra which we do not touch on and establishes finite convergence of their hierarchy at a lower level that we were able to show (\Cref{prop:fin_conv} in this work versus Theorem 4.8 in \cite{other_guys}).  We expect that both papers have much to offer one another, but we leave the full set of implications from the combined results for future work.  

\paragraph{Subsequent Work.  }  After the preprint of this work many authors improved on the best known approximation factors for Quantum Max Cut \cite{lee2024improved, apte2025improved, huber2024second, gribling2025improved, lee2022optimizing}, with the current state of the art for general graphs being $0.611$ and the best conjectured achievable approximation factor being $0.625$\cite{apte2025conjectured}.  Additional new directions include the study of the complexity of Quantum Max Cut and more general Hamiltonians \cite{kallaugher_et_al:LIPIcs.ITCS.2025.63, piddock2025quantum}, the study of approximation algorithms for generalizations of Quantum Max Cut to higher dimensions \cite{jorquera2024monogamy, carlson2023approximation}, the study of approximation algorithms for the EPR Hamiltonian \cite{ju2025improved, tao2025refined, apte20250} as well as approximating constrained quantum problems \cite{parekh2024constrained, culf2024approximation}.  Generalizations of the weighted star inequality given in this paper have been used as crucial components in the improved rounding algorithms (i.e. \cite{lee2024improved, apte2025improved}).  Code used for the numerics of this paper can be found at \cite{czqubit_projector_sdp_qmaxcut}.

\section{Notation}\label{sec:notation}
The Pauli matrices are defined as:
\begin{align*}
\label{eq:paulis}
X=\begin{bmatrix}
0 & 1 \\
1 & 0
\end{bmatrix},
\,\,\,\,\,\,
&Y=\begin{bmatrix}
0 & -i \\
i & 0
\end{bmatrix}, 
Z=\begin{bmatrix}
1 & 0 \\
0 & -1
\end{bmatrix},\\
\text{and}\,\,\,
&\mathbb{I}=\begin{bmatrix}
1 & 0 \\
0 & 1
\end{bmatrix}.
\end{align*}
\noindent Subscripts indicate quantum subsystems among $n$ qubits.  For instance, the notation $\sigma_i$ is used to denote a Pauli matrix $\sigma \in \{X,Y,Z\}$ acting on qubit $i$, i.e., $\sigma_i := \mathbb{I} \otimes \mathbb{I} \otimes \ldots \otimes \sigma \otimes \ldots \otimes \mathbb{I} \in \mathbb{C}^{2^n \times 2^n}$, where the $\sigma$ occurs at position $i$.  $\mathbb{I}$ will also be used to denote identity matrices of arbitrary context dependent size.  

We will be considering weighted graphs, $(V, \{w_{ij}\}_{ij\in V\times V})$, where each weight is non-negative.  Without loss of generality we can assume the graph is complete by possibly setting some weights to zero, so we need not include an edge set in the description of the graph.  The complex conjugate transpose of a given matrix will take the standard notation, $A^*$, and we will denote the $\max/\min$ eigenvalue of a given operator $A$ as $\mu_{max}(A)/\mu_{min}(A)$ respectively.  We will be considering Hermitian operators and operators which differ from a Hermitian matrix by a similarity transform, i.e., $T A T^{-1}$ is Hermitian, so this can be well-defined by
\begin{align*}
\mu_{min} (A) :=\min \lambda\in \mathbb{R}: \det(\lambda \mathbb{I} - A)=0;\\
\mu_{max} (A) :=\max \lambda\in \mathbb{R} : \det(\lambda \mathbb{I} - A)=0. 
\end{align*}
The associated eigenvectors for the max / min eigenvalues will be denoted as $|\mathrm{GS}\rangle$ regardless of max or min, since it will be clear from the context.

We will need to discuss matrix/scalar variables and will generally denote these with lower case letters, while upper case letters will be generally used to denote assignments to those variables.  Polynomials in matrix variables will generally be denoted with Greek letters.  

In later sections, we will be considering special families of graphs to give exactness and inexactness results. The notation $K_n$ is used to denote the complete graph with uniform weights on $n$ vertices. In other words, $K_n = (V, w)$ with $|V|=n$ and $w_{ij}=1$ for all $(i,j)\in V\times V$. Similarly, $K_{n,m}$ denotes the complete bipartite graph with $n$ and $m$ vertices on the two sides, i.e., $K_n = (V, w)$ with $V= V_A \cup V_B$ and $w_{ij}=1$ for all $(i,j)\in V_A\times V_B$ but $0$ otherwise. Here, $|V_A|=n,~|V_B|=m$ and $V_A \cap V_B = \varnothing$. Complete tripartite graphs are denoted similarly as $K_{n,m,l}$ with the vertex set separating to three nonoverlapping subsets. 

In this paper, we occasionally use the term ``Gram vectors''. 
For any positive semidefinite matrix $M\in\mathbb{K}^{k\times k}$, there always exist a matrix $V\in\mathbb{K}^{k\times k}$ such that $V^*V=M$. 
It follows that there must exist a collection of vectors $|\alpha\rangle \in \mathbb{K}^k$ such that $M_{\alpha\beta} = \langle \alpha |\beta \rangle \text{ for all } \alpha,\beta \in\{1,2,\ldots,k\}$; 
namely, these vectors correspond to the columns of $V$.  
We refer to these vectors $|\alpha\rangle$ as the \emph{Gram vectors} of $M$. 
The field $\mathbb K$ can be either $\mathbb{C}$ or $\mathbb{R}$ depending on the context. 
There are infinitely many possible choices of the Gram vectors for a given matrix $M\succeq0$, so we use this term only when referring to properties that are independent of this choice (i.e., are well-defined) or when discussing a specific choice suffices. 
\section{NPA Hierarchy}
\subsection{Phrasing the Local Hamiltonian problem as an Operator Program}

Operator programs are a powerful and flexible way of stating difficult problems.  These are generally stated as the problem of optimization over non-commuting (nc) polynomials over sets of non-commuting variables ($\{a_i\}$).  In this context the variables $\{a_i\}$ are unspecified complex matrices of some finite, fixed, unspecified size ($a_i$ has the same size as $a_j$ so that multiplication is well defined).  It could be the case that the objective goes to infinity as the matrices get larger or that the objective converges to some fixed value in the limit of large matrices, but for the cases we will consider here an optimal feasible solution to the problem will consist of matrices of finite size, so the programs discussed here are all well-defined and explicitly obtain their maxima/minima.  
Depending on the convention \cite{pir10con}, one often also includes variables $\{a_i^*\}$ for denoting the complex conjugate of the matrix variables, however, in this paper these are redundant since we will always optimize over Hermitian matrices.  Polynomials in these variables will consist of linear combinations of monomials in the nc variables.  The set of monomials of degree $\leq \ell$ is denoted $\Gamma_\ell=\{a_{i_1} a_{i_2} ...a_{i_q}: q \leq \ell\}$ so an arbitrary degree-$\ell$ nc polynomial can be denoted $\theta(\{a_i\})=\sum_{\phi \in \Gamma_\ell} \theta_\phi \phi$ where $\theta_\phi \in \mathbb{C}$ for all $\phi$.  $\Gamma_\ell$ will always contain a term of degree $0$, $\mathbb{I}$.  $\mathbb{I}$ varies inside the program since it will have size matching the $\{a_i\}$ but will always denote an identity of the appropriate size.  
\begin{definition}\label{def:op_prog} Given nc polynomials $\theta$ and $\{\eta_i\}$ with $\theta^*=\theta$, an operator program $\mathcal{O}$ is an optimization problem of the following form:
\begin{align}
\min/\max \quad & \braket{g| \theta(\{a_i\})|g}\\
\mathrm{s.t.} \quad & \eta_j(\{a_i\}) =0 \,\, \text{  for all $j$},\\
& \label{eq:herm_ness}a_i =a_i^*,\, \forall i, \\
& \braket{g| g} =1.
\end{align}
\end{definition}

To build intuition we will first consider an ncp optimization problem where the constraints force the variables to be commuting, and hence the problem reduces to a combinatorial optimization problem.  Given a graph $(V, w)$, the {\scshape MaxCut} problem is equivalent to the following optimization problem:
\begin{align}
MC(V, w) :=&\max \,\, \sum_{ij} w_{ij} \frac{1-z_i z_j}{2} \nonumber\\
s.t. &\,\, z_i \in \{\pm 1\} \,\, \forall \,\, i \in [n].
\end{align}
We could also have phrased {\scshape MaxCut} as a local Hamiltonian problem where the local terms of the Hamiltonian are diagonal in the $Z$ basis \cite{woc03}.  In this case the largest eigenvalue would be $MC(V, w)$.  Since the matrix is diagonal, the extremal eigenvector can be assumed to be a computational basis state WLOG and this basis state provides the optimal assignment for Max Cut: 
\begin{align}\label{eq:ham_max_cut_def}
    MC(V, w)=&\max \,\, \bra{g} \sum_{ij} w_{ij} \frac{\mathbb{I}-Z_i Z_j}{2} \ket{g}\nonumber \\
    s.t. &\,\, \braket{g|g}=1.
\end{align}
Note that the operators above are {\it not} variables, they are the explicitly defined Pauli matrices from \cref{sec:notation}.  Additionally the vector $\ket{g}$ is a vector variable of fixed size, unlike \Cref{def:op_prog}.  A natural direction for stating \cref{eq:ham_max_cut_def} as a ncp optimization problem is ``promoting'' the actual Pauli matrices $Z_i$ to matrix variables $z_i$.  This would lead to a {\it relaxation} where the optimal solution to the operator program would be at least the solution to the relevant {\scshape MaxCut} instance. 
 To get the objectives to match we will need to explicitly enforce constraints on $z_i$ which are satisfied by $Z_i$.  We must demand that the $z_i$ commute, as well as that they square to the identity.  The resulting operator program is
\begin{align}\label{eq:op_max_cut_def}
    \max \quad & \bra{g} \sum_{ij} w_{ij} \frac{\mathbb{I}-z_i z_j}{2} \ket{g} \\ 
    \label{eq:mc1}\mathrm{s.t.}  \quad &z_i^2=\mathbb{I} \,\, \forall \,\, i \in[n],\\
     \label{eq:mc2}&z_i z_j-z_j z_i=0 \,\, \forall \,\,i, j \in [n],\\
     \label{eq:mc_end}&z_i^*=z_i\,\, \forall \,\, i \in [n],\\
     &\braket{g|g}=1.
\end{align}

 \begin{proposition}
     The program defined in \cref{eq:op_max_cut_def} has optimal objective $MC(V, w)$.
 \end{proposition}
\begin{proof}
    Let $z_i=Z_i'$ and $\ket{g}=\ket{\psi}$ be the optimal solution to \cref{eq:op_max_cut_def}.  $Z_i'$ all square to the identity and are Hermitian so they have at most two eigenvalues, $\{\pm 1\}$.  Since the $Z_i'$ all commute we can construct a basis which simultaneously diagonalizes all the $Z_i'$.  The objective is diagonal in this basis so we may assume WLOG that $\ket{\psi}$ is one of these basis elements and that $\bra{\psi} Z_i' \ket{\psi} \in \{\pm 1\}$.  Let us define $z_i'=\bra{\psi} Z_i' \ket{\psi} \in \{\pm 1\}$, so $(z_i')^2=1$.  By the eigenvector property,
    \begin{align*}
    \bra{\psi} \sum_{ij} w_{ij} \frac{\mathbb{I}-Z_i' Z_j'}{2} \ket{\psi}
    &=\sum_{ij} w_{ij} \frac{1-\bra{\psi}Z_i'\ket{\psi} \bra{\psi} Z_j'\ket{\psi}}{2} \\
    &=\sum_{ij} w_{ij} \frac{1-z_i' z_j'}{2},
    \end{align*}
    so the optimal objective of \cref{eq:op_max_cut_def} is less than or equal to $MC(V, w)$.  We already know that optimal objective of \cref{eq:op_max_cut_def} is greater than or equal to $MC(V, w)$ since it is a relaxation.
    
\end{proof}

Naturally we may consider a generic $2$-Local Hamiltonian problem and ask similar questions.  Arbitrary $2$-Local Hamiltonians may be written as $H=\sum_{ij} H_{ij}$ where $H_{ij}$ acts only on qubits $i$ and $j$.  We can express each $H_{ij}$ in the Pauli basis as 
\begin{equation}
    H_{ij}=\sum_{\sigma, \gamma \in \{\mathbb{I}, X, Y, Z\}} c_{\sigma, \gamma}^{ij} \,\, \sigma_i \gamma_j,
\end{equation}
for $c_{\sigma, \gamma}^{ij} \in \mathbb{R}$.  This lets us express the overall Hamiltonian as
\begin{equation}
    H=\sum_{ij} \sum_{\sigma, \gamma \in \{\mathbb{I}, X, Y, Z\}} c_{\sigma, \gamma}^{ij} \,\, \sigma_i \gamma_j.
\end{equation}
The maximum eigenvalue problem is then
\begin{align}
   \mu_{max}(H)= &\max \,\, \bra{g} \sum_{ij} \sum_{\sigma, \gamma \in \{\mathbb{I}, X, Y, Z\}} c_{\sigma, \gamma}^{ij} \,\, \sigma_i \gamma_j \ket{g} \nonumber \\
   s.t. &\,\, \braket{g|g}=1.
\end{align}
We may promote the Pauli matrices above to operator variables, $X_i \rightarrow x_i, Y_i \rightarrow y_i, Z_i \rightarrow z_i$, to get a relaxation, but we will need to know what constraints to enforce to ensure that the operator problem has the same objective as the explicit local Hamiltonian problem we have in mind, just as for {\scshape MaxCut}.  Enforcing constraints of the form \cref{eq:mc1,eq:mc2,eq:mc_end} plus additional anti-commutation constraints is sufficient: 
\begin{definition}
Given a $2$-Local Hamiltonian $H$ on $n$ qubits,  
\begin{align}
\mathcal{Pauli}(H) &:= \max\quad \bra{g}  \left(\sum_{ij} \sum_{\sigma, \gamma \in \{\mathbb{I}, x, y, z\}} c_{\sigma, \gamma}^{ij}\,\, \sigma_i \gamma_i\right)\ket{g}\\
\nonumber \mathrm{s.t.} \quad &\text{for all distinct $j, k\in [n]$}:\\
&\mathbb{I}=x_j^2=y_j^2=z_j^2, \\
\label{pauli_anti_commute_const} &\{x_j, y_j\}=0, \,\,\{x_j, z_j\}=0, \,\,\{y_j, z_j\}=0,\\
&a_j b_k-b_k a_j=0 \,\, \forall \,\,a, b \in \{x, y, z\},\\
&x_j^*=x_j, \,\,y_j^*=y_j, \,\,z_j^*=z_j,\\
&\braket{g|g}=1.
\end{align}
\end{definition}
\begin{proposition}[Theorem 2.3 in \cite{cha17}]
    $\mu_{max}(H)=\mathcal{Pauli}(H)$.
\end{proposition}
The proof of this statement proceeds by showing that {\it any} operators which satisfy the relations above must be equal to the Pauli matrices up to overall unitary and tensoring with identity matrices.  In a sense the smallest feasible solution to $\mathcal{Pauli}(H)$ are the Pauli matrices themselves and larger solutions must have the same objective.  In the language of representation theory, the Pauli group has only two irreducible representations: the trivial representation and the defining representation.\knote{Edited here}

\subsection{{\scshape QMaxCut} as an Operator Program}\label{subsec:qmcop}

While the $\mathcal{Pauli}$ program is very nice because of its generality, Hamiltonians are often best studied with the natural symmetry present taken into account.  Our interest is in a specific family of Local Hamiltonians known as ``Quantum Max Cut'' ({\scshape QMaxCut}) in many works \cite{gha19alm, par22opt, ans20bey, par21app}, so our aim is to produce the ``natural'' operator programs for these Hamiltonians.  Given a weighted graph $(V, w)$ with non-negative weights $w_{ij}\geq0$, the corresponding {\scshape QMaxCut} instance is defined on $n=|V|$ qubits \footnote{In later sections, $n$ is not always $|V|$ depending on the graph we focus on, which should be clear from the context.} by 
\begin{equation}\label{eq:QMCHamDef}
QMC(V, w):= \mu_{max} \left( \sum_{ij} w_{ij} H_{ij}\right),
\end{equation}
where $H_{ij} := \frac{1}{4} \left(\mathbb{I} - X_i X_j-Y_i Y_j -Z_i Z_j \right)$. 
The term $H_{ij}$ is a projector to the singlet state $|\psi^-_{ij}\rangle := (|0_i1_j\rangle - |1_i0_j\rangle)/\sqrt{2}$. Note that the singlet state is order sensitive ($|\psi^-_{ij}\rangle = - |\psi^-_{ji}\rangle$), but the Hamiltonian is not ($H_{ij}=H_{ji}$). 
This Hamiltonian has been well-studied in physics for decades, serving as central model for quantum magnetism. It has the nice property that it is rotation-invariant; that is, for any single-qubit unitary $U$, we have $(U^\dag)^{\otimes n} H_{ij} U^{\otimes n} = H_{ij}$. 
$H \succeq 0$ since we only consider non-negative weights $w$ ($\succeq 0$ denotes that a matrix is positive semidefinite.).  

It will be convenient for us to have a definition of another Hamiltonian which is simply an affine shift of the {\scshape QMaxCut} Hamiltonian.  If we define the usual quantum SWAP operators as 
\begin{equation}\label{eq:swap_def}
P_{ij}=\begin{bmatrix} 
1 & 0 & 0 & 0\\
0 & 0 & 1 & 0\\
0 & 1 & 0 & 0\\
0 & 0 & 0 & 1
\end{bmatrix}_{ij}=\frac{\mathbb{I} +X_i X_j +Y_i Y_j +Z_i Z_j}{2},
\end{equation}
we can then define
\begin{equation}
\SWAP(V, w)=\mu_{min} \left(\sum_{ij \in E} w_{ij} P_{ij} \right).
\end{equation}
The extremal eigenvalues are related as:
\begin{equation}
\SWAP (V, w)= \sum_{jk } w_{jk} -2 QMC (V, w).
\end{equation}

Our approach is to promote the operators $H_{ij}$ and $P_{ij}$ to variables, but we are left with the same question of deciding what constraints to include to accurately capture the local Hamiltonian problem.  Our work naturally extends that of ~\cite{par22opt}, who used such operators to obtain an optimal approximation for {\scshape QMaxCut} using product states.  The following sets of constraints are sufficient for {\scshape QMaxCut} and SWAP Hamiltonians respectively:

\begin{definition}[$\mathcal{Proj}(V, w)$]  Given Hamiltonian $H=\sum_{ij} w_{ij} H_{ij}$ corresponding to graph $(V, w)$, define 
\begin{align}
\mathcal{Proj}(H)=&\mathcal{Proj}(V, w):= \max \ \bra{g} \left(\sum_{jk} w_{jk} h_{jk}  \right) \ket{g}&\\
\nonumber \mathrm{s.t.}\quad &\forall \text{ distinct }\,\, i, j, k, l\in [n]:\\
  &h_{ij}^2 = h_{ij}, \label{eq:singproj1}\\
  &h_{ij}h_{kl} = h_{kl} h_{ij}, \label{eq:singprojcomm}\\
  &h_{ij} h_{jk} +h_{jk} h_{ij}= \frac{1}{2}(h_{ij}+h_{jk} -h_{ik}),\label{eq:anticommproj}\\
  &h_{ij}^* =h_{ij}=h_{ji},\label{eq:singproj2}\\
  &\braket{g|g}=1.\label{eq:singprojnormalize}
\end{align}
\end{definition}

\begin{definition}[$\mathcal{Perm}(V, w)$]  Given Hamiltonian $H=\sum_{ij} w_{ij} P_{ij}$ corresponding to graph $(V, w)$, define 
\begin{align}
\mathcal{Perm}(H)=&\mathcal{Perm}(V, w):=\min \ \bra{g} \left(  \sum_{jk} w_{jk} p_{jk}  \right)\ket{g}&\\
\nonumber \mathrm{s.t.} \quad & \forall \text{ distinct }\,\, i, j, k, l\in [n]:\\
\label{eq:sym_const_1}&p_{ij}^2=\mathbb{I},\\
\label{eq:sym_const_2}&p_{ij} p_{kl}=p_{kl}p_{ij},\\
\label{eq:anti_comm}&p_{ij} p_{jk} +p_{jk} p_{ij} =p_{ij} +p_{jk} +p_{ik}-\mathbb{I},\\
\label{eq:herm_const} &p_{ij}^*=p_{ij}=p_{ji},\\
&\braket{g|g}=1.
\end{align}
\end{definition}

It is easy to verify that $\mathcal{Proj}$ and $\mathcal{Perm}$ are equivalent in the sense that optimal objectives of $\mathcal{Perm}$ and $\mathcal{Proj}$ are affine shifts of one another: 
\begin{equation}
\mathcal{Perm} (V, w)= \sum_{jk \in E} w_{jk} -2 \mathcal{Proj} (V, w).
\end{equation}  This can be verified by observing that if $\{P_{jk}'\}$ is a feasible solution for $\mathcal{Perm}$ then $\{(\mathbb{I}- P_{jk}')/2\}$ is a feasible solution for $\mathcal{Proj}$ and that if $\{H_{jk}'\}$ is feasible for $\mathcal{Proj}$ then $\{\mathbb{I}-2 H_{jk}'\}$ is feasible for $\mathcal{Perm}$.  

The intuition behind the $\mathcal{Perm}$ constraints can most easily be understood in the context of the representation theory of the symmetric group.  It is well known that the symmetric group has a ``finite presentation''.  Loosely speaking this means that there is a finite set of generators such that any element of the symmetric group can be written as the product of generators, and any product of elements from the symmetric group can be inferred from some finite set of multiplication rules on those generators.  Using standard notation, the symmetric group $S_n$ is generated by transpositions $(i, j)$ subject to the following rules:
\begin{align}
    \nonumber \forall \text{ distinct }\,\, i, j, k, l\in [n]:\\
\label{eq:fin_pres_1}(i, j)^2&=1,\\
\label{eq:fin_pres_2}(i, j) (k, l)&=(k, l)(i, j),\\
\label{eq:fin_pres_3}(i, j) (j, k) (i, j)&=(j, k)(i, j)(j, k).
\end{align}
Since all multiplicative identities can be derived from these rules, if we have operators $p_{ij}$ which satisfy analogous relations then multiplication of products of monomials in $\{p_{ij}\}$ must behave exactly like products of transpositions.  Hence feasible solutions $\{p_{ij}'\}$ must correspond to a representation of the symmetric group (see Appendix \ref{sec:rep_theory}).
Note that there is a correspondence between \cref{eq:sym_const_1} and \cref{eq:fin_pres_1} as well as between \cref{eq:sym_const_2} and \cref{eq:fin_pres_2}, but apparently none for \cref{eq:fin_pres_3}.
Instead, the operator program $\mathcal{Perm}$ contains an additional ``anti-commuting constraint'' \cref{eq:anti_comm}.  $\mathcal{Perm}$ actually has an implicit constraint corresponding to \cref{eq:fin_pres_3} (\Cref{prop:add_const}), so the constraints present enforce that the operators $p_{ij}$ ``look like'' a finite presentation of the symmetric group, plus an additional anti-commuting constraint. 
This fact is crucial for understanding why the operator programs listed are accurately capturing the relevant local Hamiltonian problems (\cref{thm:perm_is_opt}).  \Cref{eq:sym_const_1,eq:sym_const_2,eq:anti_comm,eq:herm_const}
force the operators to correspond to a representation of the Symmetric group, and \cref{eq:anti_comm} further forces the operators to correspond to the correct representation for the Hamiltonian.  

\begin{proposition}\label{prop:add_const}
    For all distinct $i, j, k\in [n]$, $\mathcal{Perm}$ satisfies the additional constraint 
    \begin{equation}\label{eq:totalswap}
    p_{ij} p_{jk}p_{ij}
    =p_{ik},
    \end{equation}
    and $\mathcal{Proj}$ satisfies the additional constraint
    \begin{equation}\label{eq:quarterformula}
    4 h_{ij}h_{jk}h_{ij}=h_{ij}.
    \end{equation}
\end{proposition}
\begin{proof}
    From the anticommutation relation \cref{eq:anticommproj}, we can expand $h_{jk}=h_{ij}+h_{ik}-2(h_{ij}h_{ik}+h_{ik}h_{ij})$ to obtain
    \begin{eqnarray}
    h_{ij}h_{jk}h_{ij}&=&
    h_{ij}\bigl(h_{ij}+h_{ik}-2(h_{ij}h_{ik}+h_{ik}h_{ij})\bigr)h_{ij}\,\,\,\,\,\,\,\,\\
        &=&h_{ij}-3h_{ij}h_{ik}h_{ij},
    \end{eqnarray}
    where we used \cref{eq:singprojnormalize} as well. 
    Repeating the same substitution for $h_{ik}$ in the second term gives 
$h_{ij}h_{jk}h_{ij}=h_{ij}-3\bigl(h_{ij}-3h_{ij}h_{jk}h_{ij}\bigr)$,  
    which results in \cref{eq:quarterformula} after solving the linear equation. To obtain \cref{eq:totalswap}, apply the same proof with $h_{ij}=(\mathbb{I}-p_{ij})/2$. 
\end{proof}

This additional constraint \cref{eq:quarterformula} is actually one of the basic relations in the Temperley-Lieb algebra \cite{tem71rel} describing the 1-dimensional Heisenberg chain and variants. The factor of 4 could be understood in relation with other algebraic structures used for analyzing the Heisenberg model \cite{san05gro,bea06som}. 
Furthermore, an immediate corollary of \cref{prop:add_const} is that 
\begin{equation}\label{eq:sym_const_3}
p_{ij}p_{jk}p_{ij}=p_{ik}=p_{jk}p_{ij}p_{jk},
\end{equation}giving the constraint corresponding to \cref{eq:fin_pres_3}.  Reasoning as above one can obtain:

\begin{theorem}[\cite{pro76}]\label{thm:perm_is_opt}
\knote{edited here}$\mathcal{Perm}(V, w)=\SWAP(V, w)$.
\end{theorem}
\begin{proof}
See \ref{sec:perm_conv}.
\end{proof}

\begin{corollary}
$\mathcal{Perm}$ is a {\sf QMA}-complete operator program.  
\end{corollary}
\begin{proof}
    Determining $\SWAP(G, w)$ is {\sf QMA}-complete \cite{pid17} and by \Cref{thm:perm_is_opt} finding $\SWAP(G, w)$ is equivalent to finding $\mathcal{Perm}(G, w)$.
\end{proof}

\Cref{thm:perm_is_opt} establishes that the optimal operator variables in $\mathcal{Perm}$ and actual quantum operators in $\SWAP$ are the same, or that $\mathcal{Perm}$ can be optimized using the fixed assignments $p_{ij}=P_{ij}$ WLOG. 
Similarly, the abstract operator variables in $\mathcal{Proj}$ denoted by $h_{ij}$ become indistinguishable from the actual quantum operators of {\scshape QMaxCut} denoted by $H_{ij}$ when they are optimized according to the program. 
This is somewhat surprising, since in general, there are infinite amount of nontrivial relations between singlet projectors $H_{ij}$ that include higher order terms. The theorem implies that all of them must be derivable purely from relations in the program, namely \crefrange{eq:singproj1}{eq:singproj2} for $h_{ij}$, and \crefrange{eq:sym_const_1}{eq:herm_const} for $p_{ij}$. 
Proposition \ref{prop:add_const} could be seen as one such derivation, and we show other examples in Appendix \ref{app:relationderivation}. 
Later in section \ref{sec:exact}, we will take advantage of this equivalence, and identify $h_{ij}$ and $H_{ij}$ for practical purposes. If we explicitly add the constraints \cref{eq:sym_const_3} to the program then the above conclusions hold even if only a single anti-commuting constraint is included.  So if we only enforced $p_{12}p_{23}+p_{23}p_{12}=p_{12}+p_{23}+p_{13}-\mathbb{I}$ and did not include any other anti-commuting constraints then we would still have $\SWAP(V, w)=\mathcal{Perm}(V, w)$, essentially because a single constraint rules out all the incorrect representations.

The rotation invariance of the cost function $QMC(V,w)$ is known as the SU(2) symmetry of the isotropic Heisenberg model in condensed matter physics. The ground state of the model can always be made to be SU(2) symmetric for a finite-size system. Therefore, having independent $x$, $y$, and $z$ variables for the operator program like $\mathcal{Pauli}(H)$ could be considered somewhat wasteful. An intuitive way to view our $\mathcal{Proj}$ operator program is that it only deals with the fundamental SU(2) symmetric objects, namely the expectation value of the projectors to the singlet state. As we show in the next section, we can construct a provably converging NPA hierarchy with this approach, which turns out to be practically very efficient compared to the $\mathcal{Pauli}$ hierarchy. However, interestingly, there are also instances where explicitly breaking the SU(2) symmetry allows the SDP-relaxed value to approximate the true solution more closely, as discussed in section \ref{subsubsec:smallstat}.

\subsection{The Hierarchy}\label{sec:hierarchy}
The NPA hierarchy \cite{pir10con} gives a general procedure for constructing a family of semidefinite optimization programs parameterized by ``level'' which provide increasingly accurate relaxations on operator programs $\mathcal{O}$ (\Cref{def:op_prog}).  The definition given in \cite{pir10con} is more general than the one presented here which we have simplified for our particular application.  We motivate the definition of the hierarchy by considering a ``moment matrix''.  Let $\{A_i\}$, $\ket{G}$ be some feasible solution to $\mathcal{O}$.  $M_\ell$ is defined to be a complex Hermitian matrix with rows/columns indexed by elements of $\Gamma_\ell :=\{a_{i_1} a_{i_2} ...a_{i_m}: m \leq \ell\}$, the set of monomials of degree $\leq \ell$.
Entries of $M_\ell$ are defined so that 
\begin{align}
M_\ell(a_{i_1} ... a_{i_m} , a_{j_1}...a_{j_r}&):=\braket{G | (A_{i_1}...A_{i_m})^* A_{j_1} ... A_{j_r} |G }\nonumber \\
= &\braket{G| A_{i_m} ... A_{i_1} A_{j_1} ... A_{j_r}| G}.
\end{align}
By definition $M_\ell$ gives consistent values for all monomials in $\Gamma_{2\ell}$, and we can unambiguously define 
\begin{align}\label{eq:val_def}
 M_\ell( a_{i_1} ... a_{i_r} ):=\begin{cases}
     M_{\ell} (\mathbb{I}, a_{i_1} ... a_{i_r}) \text{ if $r\leq \ell$},\\
     M_\ell(a_{i_\ell} ... a_{i_1}, a_{i_{\ell+1}} ... a_{i_r}) \text{ otherwise}.
 \end{cases}   
\end{align}
Similarly for any degree-$\ell$ nc polynomials $\beta(\{A_i\})=\sum_{\phi \in \Gamma_\ell} \beta_\phi \phi(\{A_i\})$ and $\gamma(\{A_i\})=\sum_{\phi \in \Gamma_\ell}\gamma_\phi \phi(\{A_i\})$,  with $ \beta_\phi, \gamma_\phi\in \mathbb{C}$ for all $\phi$, we define 
\begin{align}
    M_\ell(\beta, \gamma)&:=\sum_{\phi, \phi' \in \Gamma_\ell} {\beta_\phi}^* \gamma_{\phi'} \braket{G | \phi^*(\{A_i\}) \phi'(\{A_i\})|G} \nonumber\\
    &= \sum_{\phi, \phi' \in \Gamma_\ell} {\beta_\phi}^* \gamma_{\phi'} M_\ell(\phi, \phi') \nonumber\\
    &= \sum_{\phi, \phi' \in \Gamma_\ell} (\beta_\phi)^* \gamma_{\phi'} M_\ell(\phi^*\phi') .
\end{align}
Furthermore, since the polynomials $\{\eta_k\}$ corresponding to the constraints must evaluate to zero for the variables $\{A_i\}$, we must also have that $0=\eta_k(\{A_i\})=\bra{G} \eta_k(\{A_i\}) \ket{G}=M_\ell(\eta_k)$.  More generally, feasibility for $\mathcal{O}$ implies that for any nc polynomials $\beta, \gamma$ with $\deg(\beta\eta_k\gamma)\leq 2 \ell$, 
\begin{equation}\label{eq:npa_const_0}
    M_\ell(\beta \eta_k \gamma)=0, 
\end{equation}
since $\bra{G}\beta \eta_k \gamma \ket{G}=\bra{G}\beta 0 \gamma \ket{G}$. The matrix $M_{\ell}$ also naturally satisfies 
\begin{align}\label{eq:npa_const_M}
    M_\ell=M_{\ell}^*  \text{ ~~~and~~~ }
    M_\ell \succeq 0,
\end{align}
where the latter constraint can be seen from the fact that for any complex column vector $v$, $v^* M_\ell v=M_{\ell} (\beta, \beta)=\bra{G} \beta^* \beta \ket{G}\geq 0$ for some $\beta$ that is a degree-$\ell$ nc polynomial.  

Given an operator program $\mathcal{O}$, the $\ell$-th level of the NPA hierarchy, denoted by $NPA_\ell (\mathcal{O})$, is the optimization of the matrix $M_\ell$ as a variable, subject to the above properties viewed as constraints, not requiring a corresponding legitimate vector $\ket{G}$ or operators $\{A_i\}$. $NPA_\ell (\mathcal{O})$ relaxes $\mathcal{O}$ since we could have used the optimal $\ket{G}, \{A_i\}$ to define $M_\ell$, so if $\mathcal{O}$ is a maximization problem, $NPA_\ell(\mathcal{O})\geq \mathcal{O}$. 
For a given monomial $a_{i_1}...a_{i_r}$, $M_\ell(a_{i_1}...a_{i_r})$ is defined as a specific entry of $M_\ell$ according to \cref{eq:val_def}. 
However, there are many other distinct entries of $M_\ell$ which will match. For example, $M_2(a_1a_2):= M_2(\mathbb{I}, a_1 a_2)$, but $ M_2(\mathbb{I}, a_1 a_2)=M_2(a_1, a_2)=M_2(a_2 a_1, \mathbb{I})$.  In $NPA_\ell$ we will force the value of $M_\ell(a_{i_1}...a_{i_r})$ to be consistent with the other distinct entries by enforcing that these other entries of $M_\ell$ are equal to the ``cannonical'' value $M_2(a_1a_2):= M_2(\mathbb{I}, a_1 a_2)$.  Note that $M_\ell$ will be {\it indexed} by variables $\phi \in \Gamma_\ell$ of the original operator program, but these $\phi$ do not vary inside $NPA_\ell$\footnote{We note that ``good perturbations'' as in \cite{pro21} can also be used to construct an SDP hierarchy for the Permutation program without requiring explicitly enforcing the constraints \Cref{npa_const_p_i}.  See \cite{other_guys} remark 5.5 for details.}. 
\begin{definition}[$NPA_\ell (\mathcal{O})$]
Given an operator program $\mathcal{O}$ with objective to max / min the extremal eigenvalue of a nc polynomial $\theta\in \text{span}(\Gamma_{2 \ell})$ satisfying $\theta^*=\theta$ and subject to constraints $\{\eta_k\}$, define
\begin{align}
NPA_\ell&(\mathcal{O}):= \max /\min \quad M_\ell(\theta)\\
\label{eq:id_const}\mathrm{s.t.}\quad & M_\ell(\mathbb{I})=1,\\
\label{concat_const}&M_\ell(\phi, \phi')=M_\ell(\phi^* \phi') ~~~~ \forall  \phi, \phi'\in \Gamma_\ell,\\
\label{npa_const_p_i} &M_\ell (\beta \eta_k \gamma)=0 ~~\forall \beta, \gamma \in \Gamma_{2 \ell}, \eta_k: \deg(\beta\eta_k\gamma)\leq 2 \ell,\\
&M_\ell\succeq 0, ~\mathrm{Hermitian}.
\end{align}
\end{definition}
In several works \cite{par21app, ans20bey, par22opt} a ``real version'' of $NPA_\ell$ is studied, $NPA_\ell^{\mathbb{R}}$.  The key difference is that $NPA_\ell^{\mathbb{R}}$ takes only the real portion of the moment matrix, effectively ``zeroing out'' skew-Hermitian polynomials.
\begin{definition}[$NPA_\ell^{\mathbb{R}} (\mathcal{O})$]
Given an operator program $\mathcal{O}$ with objective to max / min the extremal eigenvalue of a nc polynomial $\theta\in \text{span}(\Gamma_{2 \ell})$ satsifying $\theta^*=\theta$ and subject to constraints $\{\eta_k\}$, define\begin{align}
NPA_\ell^{\mathbb{R}} &(\mathcal{O}):= \max /\min \quad M_\ell^{\mathbb{R}}(\theta)\\
\mathrm{s.t.} \quad & M_\ell^{\mathbb{R}} (\mathbb{I})=1,\\
&M_\ell^{\mathbb{R}} (\phi, \phi')=M_\ell^{\mathbb{R}} (\phi^* \phi') ~~~~ \forall  \phi, \phi'\in \Gamma_\ell,\\
&\label{npa_const_p_i_real}M_\ell^{\mathbb{R}} ((\beta \eta_k \gamma + \gamma^* \eta_k^* \beta^*)/2)=0\nonumber \\
&\hspace{1.5cm} \forall \beta, \gamma \in \Gamma_{2 \ell}, \eta_k: \deg(\beta\eta_k\gamma )\leq 2\ell,\\
&M_\ell^\mathbb{R} \succeq 0, ~\mathrm{Symmetric}.
\end{align}
\end{definition}
As noted in other works \cite{bra19, gha19alm}, $ NPA_\ell^\mathbb{R}(\mathcal{O})$ is a relaxation on $NPA_\ell(\mathcal{O})$ ($NPA_\ell(\mathcal{O}) \leq NPA_\ell^\mathbb{R}(\mathcal{O})$ if $\mathcal{O}$ is a maximization problem) since given a feasible $M_\ell$ for $NPA_\ell(\mathcal{O})$, we can take $M_\ell^\mathbb{R}:=(M_\ell+(M_\ell^*)^T)/2$ to obtain a feasible solution to $NPA_\ell^\mathbb{R}$ with the same objective, but not necessarily the other way around. 
We show explicit examples in \cref{subsubsec:smallstat} where the separation is strict for the $\mathcal{Pauli}$ case. 
However, for the specific SDP we focus on most in this paper, $NPA_1(\mathcal{Proj})$, this is a distinction without a difference:
\begin{proposition}
    For any weighted graph $(V, w)$, $NPA_1(\mathcal{Proj}(V, w))=NPA_1^\mathbb{R}(\mathcal{Proj}(V, w))$.
\end{proposition}
\begin{proof}
Let $M_1^\mathbb{R}$ be a feasible solution to $NPA_1^\mathbb{R}(\mathcal{Proj})$.  We claim that $M_1^\mathbb{R}$ is feasible for $NPA_1(\mathcal{Proj})$ and hence $NPA_1^\mathbb{R}(\mathcal{Proj}) \leq NPA_1(\mathcal{Proj})$.  This is demonstrated by showing that all the constraints $\eta$ enforced on $M_1^\mathbb{R}$ (constraints of the form \Cref{npa_const_p_i_real}) are Hermitian and hence $M_1^\mathbb{R}(\eta+\eta^*)=0$ implies $M_1^\mathbb{R}(\eta)=0$.  It is easy to see that the nc polynomials $\eta$ of the form $h_{ij}^2-h_{ij}$ and $h_{ij}h_{jk}+h_{jk}h_{ij}-(h_{ij}+h_{jk}-h_{ik})/2$ are Hermitian, hence $M_1^\mathbb{R}(\eta)=0$ for these polynomials.  Since we are looking at moment matrices with a maximum degree $2$ if $\beta \eta \gamma$ has max degree $2$ for $\eta$ one of the polynomials above $\beta=\mathbb{I}=\gamma$ so there are no other constraints to check and $M_1^\mathbb{R}$ must be feasible for $NPA_1(\mathcal{Proj})$
     
\end{proof}

Each $NPA_\ell$ can be solved via semidefinite programming, and the dual set of programs is called the {\it Sum of Squares} (SoS) hierarchy due to the interpretation of the hierarchy as an optimization over {\it sum of squares proofs}.  To motivate this imagine we are trying to maximize an nc polynomial $\theta$ and obtain an expression of the form
\begin{equation}\label{eq:sos_mot}
    \lambda \mathbb{I}-\theta=\sum_i \psi_i^* \psi_i+\sum_j \beta_j \eta_j \gamma_j,
\end{equation}
where $\lambda\in \mathbb{R}$, $\psi_i\in \mathrm{span}_\mathbb{C}(\Gamma_\ell)$ for all $i$ and $\deg(\beta_j\eta_j \gamma_j)\leq 2 \ell$ for all $j$.  Since $\eta_k$ are constraints of $\mathcal{O}$ we must have that 
\begin{equation}
    \lambda \mathbb{I}-\theta=\sum_i \psi_i^* \psi_i 
\end{equation}
for {\it any} feasible solution to $\mathcal{O}$.  $\sum_i \psi^* \psi_i$ is manifestly $\succeq 0$ so for all feasible solutions $\bra{G} \lambda \mathbb{I}-\theta\ket{G} \geq 0\Rightarrow \lambda \geq \bra{G} \theta \ket{G}$ which implies $\lambda \geq NPA_\ell(\mathcal{O})$.  An expression of the form \cref{eq:sos_mot} is generally referred to as a {\it sum of squares proof}.  The SoS optimization problem at level $\ell$ is to find the smallest $\lambda$ such that $\lambda \mathbb{I}-\theta$ can be deformed via the constraints to an expression of the form $\sum_i \psi^* \psi$: $\lambda \mathbb{I}-\theta-\sum_j \beta_j \eta_j \gamma_j =\sum_i \psi_i^* \psi_i$.  Crucially, the constraints that the SoS proof uses are of degree at most $\ell$ at the corresponding level ($\deg(\beta_j\eta_j\gamma_j)\leq 2\ell$).  For ease of reference, define the set of constraint deformation polynomials as
\begin{equation}\label{eq:U_ell_def}
    U^\ell(\{\eta_k\}):= \mathrm{span}_\mathbb{C}\{\beta \eta_k \gamma: \beta, \gamma\in \Gamma_\ell, \deg(\beta\eta_k\gamma)\leq 2\ell\}
\end{equation}

\begin{definition}[$SoS_\ell(\mathcal{O})$]\label{def:SoS}
Given an operator program $\mathcal{O}$ with objective to max/min the extremal eigenvalue of $\theta$ subject to the constraints $\{\eta_k\}$,
\begin{align}
    &SoS_\ell(\mathcal{O}):= \min/\max \quad \lambda\\
    \mathrm{s.t.}\quad & \sum_i \psi_i^* \psi_i +\beta \nonumber\\
    &=\begin{cases} \lambda \mathbb{I}-\theta \text{ if $\mathcal{O}$ is a maximization problem}\\ \theta-\lambda \mathbb{I} \text{ if $\mathcal{O}$ is a minimization problem} \end{cases},\\
    &\psi_i\in \mathrm{span}_\mathbb{C}(\Gamma_\ell) \,\, \forall i,\\
    &\beta \in U^\ell(\{\eta_k\}),\\
    &\lambda \in \mathbb{R}.
\end{align}
\end{definition}

\textbf{Equivalence}.  SDPs for $\mathcal{Perm}$ and $\mathcal{Proj}$ are equivalent since a feasible solution for $NPA_\ell(\mathcal{Perm})$ can be used to construct a feasible solution to $NPA_\ell(\mathcal{Proj})$ and vice versa.  The reasoning for this is analogous to the reasoning behind the operator programs $\mathcal{Perm}$ and $\mathcal{Proj}$ being equivalent: Given $M_\ell$ feasible for $NPA_\ell(\mathcal{Perm})$ one can construct $M_\ell'$ feasible for $NPA_\ell(\mathcal{Proj})$ according to
\begin{align}
M_\ell'(h_{i j}  &... h_{k l}, h_{m n} ... h_{o p})\nonumber \\
&:= M_\ell\left(\frac{\mathbb{I}-p_{i j}}{2}...\frac{\mathbb{I}-p_{k l}}{2}, \frac{\mathbb{I}-p_{m n}}{2} ... \frac{\mathbb{I}-p_{o p}}{2} \right).
\end{align}

Note that the R.H.S is evaluated using the expression for $M_\ell(\beta, \gamma)$ for nc polynomials $\beta$ and $\gamma$ as in \cref{sec:hierarchy}.  The primal/dual pair $NPA_\ell/SoS_\ell$ satisfies strong duality \footnote{In order to show this, one only needs to verify that the Slater's condition \cite{ram97} is satisfied. This is not too hard to see, by considering the moment matrix corresponding to the completely mixed state being strictly positive definite.}, so it holds that $NPA_\ell(\mathcal{Perm})=SoS_\ell(\mathcal{Perm})$, $NPA_\ell(\mathcal{Proj})=SoS_\ell(\mathcal{Proj})$, and that the objectives of the two sets of programs differ by a known affine shift.  

\textbf{Convergence.  }It is simple to check that any $\mathcal{Perm}$ satisfies a boundedness condition (it is {\it Archimedean} \cite{pir10con}).  For this we need to show the existence of a constant $C$ such that $C^2 \mathbb{I}-\sum_{i<j} P_{ij}\succeq 0$ for any feasible assignment to the variables $p_{ij}=P_{ij}$.  Since $P_{ij}^2=\mathbb{I}$, $P_{ij} \preceq \mathbb{I}$ so $C^2 \mathbb{I}-\sum_{i<j} P_{ij}\succeq \left(C^2-\binom{n}{2}\right) \mathbb{I}$.  Hence we may choose $C=\sqrt{\binom{n}{2}}$.  The results of \cite{pir10con} then imply that $NPA_{\ell}(\mathcal{Perm})$ converges to the optimal objective of $\mathcal{Perm}$ in the limit of large $\ell$.  In fact, the constraints present are strong enough to guarantee finite convergence of $NPA_\ell$.  Essentially the proof of this statement involves showing that moment matrices will satisfy the ``rank condition'' of \cite{pir10con} hence an optimal operator solution can be constructed from the optimal solution to $NPA_\ell(\mathcal{Perm})$ at some finite $\ell$.

The proof we give borrows heavily from \cite{other_guys}, but it does not directly follow from their work 
\footnote{In the first version of this paper we proved a weaker result.  We were able to prove this stronger result only after becoming familiar with \cite{other_guys}}.  Our hierarchy does not reduce via a Grobner basis so we must algebraically prove constraints to use them rather than simply observing them via calculation (see Lemma $3.12$ of \cite{other_guys}).

\begin{proposition}\label{prop:fin_conv}
    Let $\ell^*=\lceil\frac{n}{2} \rceil+4$, then $NPA_{\ell^*}(\mathcal{Proj}(V, w)) = NPA_{\ell^*+1}(\mathcal{Proj}(V, w))$ 
    
    \noindent $ = NPA_{\ell^* + k}(\mathcal{Proj}(V, w))$ for any $k\in \mathbb{Z}_{\geq 0}$.
\end{proposition}
\begin{proof}
Let $M=M_{\ell^*}$ be an optimal solution to $ NPA_{\ell^*}(\mathcal{Proj}(V, w))$.  Let Gram vectors $\{\ket{\phi}\}_{\phi\in \Gamma_{\ell^*}}$ be defined so that $M(\phi', \phi)=\braket{\phi'|\phi}$.  For $\theta =\sum_{\phi \in \Gamma_{\ell^*}} c_\phi \phi $ with $c_\phi \in \mathbb{C}$ we will use the notation $\ket{\theta}=\sum_\phi c_\phi \ket{\phi}$.  Our goal is to demonstrate that for all $\zeta \in \Gamma_{\ell^*-3}$ of degree $\ell^*-3$ there exists $\theta_\zeta \in \text{span}_\mathbb{C} \Gamma_{\ell^*-4}$ such that 
\begin{equation}\label{eq:flat_trunc}
\braket{\phi | \zeta}=\braket{\phi|\theta_\zeta}
\end{equation}
for all $\phi \in \Gamma_{\ell^*-3}$.  If this holds then we can apply Theorem $2$ of \cite{pir10con} to infer the proposition by the following reasoning.  Let $Q_1=span_\mathbb{C} \{\ket{\phi}: \phi \in \Gamma_{\ell^*-4}\}$ and $Q_2=span_\mathbb{C} \{\ket{\phi}: \phi \in \Gamma_{\ell^*-3}\}$.  If there is a vector in $Q_2$ which is linearly independent from the vectors in $Q_1$ then there exists nonzero $\ket{w}=\sum_{\phi \in \Gamma_{\ell^*-3}} c_\phi \ket{\phi}$ such that $\braket{w|\phi}=0$ for all $\phi\in \Gamma_{\ell^*-4}$.  In this case $0 < \braket{w|w}=\sum_{\phi \in \Gamma_{\ell^*-3}} c_\phi \braket{w|\phi}=\sum_\phi c_\phi \braket{w|\theta_\phi}=0$.  Hence the submatrices of $M$ corresponding to $\Gamma_{\ell^*-3}$ and $\Gamma_{\ell^*-4}$ must have the same rank and we can apply an argument similar to Theorem $2$ of \cite{pir10con}\footnote{Demonstrating that at some level $\ell$ some submatrix corresponding to terms of degree $<\ell$ satisfies the rank condition and using this fact to construct an optimal solution is a generalization of flat extension known as ``flat truncation''\cite{nie2013certifying}}.

Let $\{\eta_j\}$ be the constraints described in \Cref{eq:singproj1}-\Cref{eq:singproj2}.  By definition, any $ \upsilon \in U^{\ell^*}(\{\eta_k\})$ of degree $\ell^*$ must satisfy $\braket{\phi| \upsilon}=0$ for all $\phi\in \Gamma_{\ell^*-3}$.  We will establish \Cref{eq:flat_trunc} by expressing $ \braket{\phi|\zeta}=\braket{\phi|\upsilon +\zeta} =\braket{\phi|\theta_\zeta}$ for $\theta_\zeta, \upsilon$ as described above.  In general $\upsilon$ will be of the form $\upsilon=\sum_j \beta_j \eta_j \gamma_j$.  After the individual components of $\upsilon$ are multiplied out but before like terms are canceled the degree of $\upsilon$ will be less than or equal to $\ell^*$.  After like terms are canceled the degree of $\upsilon$ will be at most $\ell^*-3$ and the degree of $\zeta+\upsilon=: \theta_\zeta$ will be at most $\ell^*-4$ once like terms are canceled.  We must solve NPA to level $\ell^*$ so that the SDP has high enough degree constraints to ``observe'' the constraint $\upsilon$.

Let $\zeta$ be some monomial of degree $\ell^*-3=\lceil n/2 \rceil +1 $.  Consider the graph $G=(V, E)$ with vertex set $[n]$ and an edge between $i$ and $j$ exactly when $h_{ij}$ is one of the variables included in the monomial $\zeta$.  By Lemma $3.14$ of \cite{other_guys} $E$ must contain a subgraph which is a triangle, a three edge line, a three edge star or two connected components containing a two edge line.  Formally, it must be that for distinct $i, j, k, l, a, b, c$ a set of the form $\{ij, jk, ik\}$, $\{ij, jk, kl\}$, $\{ij, ik, il\}$ or $\{ij, jk, ab, bc\}$ is a subset of $E$. 

Let us suppose that $\{ij, jk, ab, bc\}\subset E$.  By adding constraints $\beta \eta \gamma$ where $\eta$ is of the form \Cref{eq:anti_comm} we can permute the variables appearing in $\zeta$ to any order that we choose ($\braket{\phi|\zeta}=\braket{\phi|\beta \eta \gamma+\zeta}= \braket{\phi|\zeta'}+\braket{\phi|\theta}$ for $\zeta'$ a monomial with two adjacent variables swapped and $\theta$ is some polynomial of degree $\ell^*-4$).  Hence we can derive $\braket{\phi|\zeta}=\braket{\phi| A h_{ij} h_{jk} h_{ab} h_{bc}} +\braket{\phi|\theta}$ where again $\theta$ is some polynomial of degree $\ell^*-4$ and $A$ is a product of exactly $\ell^*-7$ variables.  \Cref{cor:alg1} \Cref{item:alg5} implies the existence of a degree $7$ polynomial, $\tau \in U^\ell(\{\eta_j\})$, such that $h_{ij} h_{jk} h_{ab} h_{bc}+\tau=\chi$ for $\chi$ a degree $3$ polynomial.  So $\braket{\phi|A h_{ij} h_{jk} h_{ab} h_{bc}}=\braket{\phi|A\tau + A h_{ij} h_{jk} h_{ab} h_{bc}}=\braket{\phi|A \chi}$.  Hence $\braket{\phi|\zeta}=\braket{\phi| A\chi+\theta}$ for $deg(A\chi +\theta)=\max\{\ell^*-7+3, \ell^*-4\}=\ell^*-4$. 

The other cases are handled similarly.  Permute the variables so that the triangle, three edge line, or three edge star occurs at the end of $\zeta$ then use \Cref{cor:alg1} \Cref{item:alg3}, \Cref{item:alg1} and \Cref{item:alg4} respectively to reduce the overall degree.

\end{proof}

We note that the overall degree of the polynomial $A(h_{ij} h_{jk} h_{ab} h_{bc}+\tau)$ before cancellation in the proof is $\ell^*$.  This is the reason why we needed to go to level $\lceil n/2 \rceil+4$ rather than $\lceil n/2\rceil+1$ to obtain convergence.  The hierarchy of \cite{other_guys} is able to obtain convergence at level $\lceil n/2\rceil+1$ because of Grobner basis techniques.  We in fact have a computer assisted proof which implies that the projector hierarchy converges at $\lceil n/2 \rceil+3$ which is essentially the same proof except using the slightly stronger \Cref{cor:alg2}.  The proof is computer assisted in the sense that the expressions are high enough degree that we were not able to simplify them by hand. It is notable that the hierarchy described here has similar convergence guarantees to the Grobner basis hierarchy of \cite{other_guys} but is simpler we do not need to impose constraints arising from reducing polynomials via the Grobner basis. 

\begin{proposition}[Computer Assisted Proof]\label{prop:fin_conv2}
    Let $\ell^*=\lceil\frac{n}{2} \rceil+3$, then $NPA_{\ell^*}(\mathcal{Proj}(V, w)) = NPA_{\ell^*+1}(\mathcal{Proj}(V, w))$ 
    $ = NPA_{\ell^* + k}(\mathcal{Proj}(V, w))$ for any $k\in \mathbb{Z}_{\geq 0}$.
\end{proposition}

For many Hamiltonians of interest $NPA_\ell(\mathcal{Proj})$ converges with smaller SDP sizes than $NPA_\ell(\mathcal{Pauli})$.  
For example, most of the graphs we address in the following section (star graphs, complete (bipartite) graphs, crown graphs, etc.) will be exactly solved by $NPA_1(\mathcal{Perm})$ but not by $NPA_1(\mathcal{Pauli})$. 
The matrix size required for convergence then reads
$\binom{n}{2}3^2+3n+1$ and $\binom{n}{2}+1$ for $\mathcal{Pauli}$ and $\mathcal{Proj}$ respectively, where the latter is more efficient by an approximate factor of 9.

\section{Exact results on some families of graphs}\label{sec:exact}

Here we detail our results concerning the exactness/inexactness of $NPA_1(\mathcal{Proj})/SoS_1(\mathcal{Proj})$ on many interesting classes of {\scshape QMaxCut} Hamiltonians.  
First of these classes is the positive weighted star graphs. The proof technique for this class involves reconstructing a quantum state with the exact same energy as the output of the $NPA_1(\mathcal{Proj})$ program.  A crucial component of this proof is a reinterpretation of ``monogamy of entanglement'' inequalities in terms of the possible angles between Gram vectors from $NPA_1(\mathcal{Proj})$. 
We show the constraints on these angles from the polynomial inequalities derived in \cite{par22opt} are actually saturated for the case of star graphs. This provides an interesting geometric perspective for monogamy of entanglement in the context of $NPA_1(\mathcal{Proj})$.

Theoretically, sufficient conditions are known to verify whether the output value of the SDP at a particular level in our hierarchy is the true ground state energy or not. One of the conditions is a quantum generalization of the so-called classical flat extension condition~\cite{pir10con}. 
This test involves running a higher level SDP in the hierarchy to check that the moment matrix obtained at the higher level is a linear extension of the moment matrix from the lower level. 
If it is, then we can guarantee that the SDP output is equal to the true ground state energy.

The other proofs for exactness relies on SoS proofs, which we analytically construct. Since the SDP hierarchies defined in \Cref{sec:hierarchy} are relaxations of the Local Hamiltonian problem, it is sufficient to construct a feasible solution to $SoS_\ell$ which achieves the optimal eigenvalue as the objective. 
To state concretely, we will be utilizing the following theorem: 
\begin{theorem}\label{thm:duality}
The upper bound obtained by the $NPA_\ell(\mathcal{Proj})$ matches exactly with the maximum eigenvalue if and only if there exists a $SoS_\ell(\mathcal{Proj})$ that upper bounds the maximum eigenvalue tightly.
\end{theorem}
Some results in the exactness proofs of other graphs will have some overlap with the first case of weighted star graph, but the explicitly constructive nature of SoS proof gives a complementary understanding of how the SDP algorithm obtains the exact solution. Finally, we discuss the sharp contrast in the SDP performance for complete graphs with even and odd number of vertices, which could be seen as a quantum version of the parity problem addressed in \cite{gri01lin}. 
Here, we prove cases where $NPA_1(\mathcal{Proj})$ is always {\it insufficient} to obtain the maximum-eigenvalue state. 



\subsection{Positive weighted star graph}
In this section we generalize the result of \cite{par21app} and prove that $NPA_1(\mathcal{Proj}(H))$ has optimal objective matching the extremal energy if the Hamiltonian is a positively weighted star.  Since $NPA_1^{\mathbb{R}}(\mathcal{Proj}(H))=NPA_1(\mathcal{Proj}(H))$ this implies also that $NPA_1(\mathcal{Proj})=\mu_{max}(H)$.  To our knowledge the first known proof of this statement is from unpublished personal correspondence \cite{per_comm}, however the proof we present here is simpler and has an intuitive geometric interpretation. 
The following theorem, proved in \cite{par22opt} about monogamy of entanglement on a triangle (three qubits $i, j$ and $k$), will be the starting point for our proof. 

\begin{theorem}[\protect{\cite[Lemma~7]{par22opt}}] \label{thm:MonogomyOfEntanglementOnTraingle}
    For any feasible moment matrix $M_1$ of $NPA_1(\mathcal{Proj})$, the following inequalities are true:
    \begin{align}
        0 \leq M_1(\mathbb{I}, h_{ij})+ M_1(\mathbb{I}, h_{jk}) + M_1(\mathbb{I}, h_{ki}) &\leq 3/2 \\
        M_1(\mathbb{I},h_{ij})^2 + M_1(\mathbb{I},h_{jk})^2  +M_1(\mathbb{I},h_{ki})^2
        &\leq \nonumber\\
        2\Big[ M_1(\mathbb{I},h_{ij})M_1(\mathbb{I},h_{jk})~~~~~~~~~~~~~~&\nonumber\\
        +M_1(\mathbb{I},h_{jk}) M_1(\mathbb{I},h_{ki})+ M_1(\mathbb{I},h_{ki})M_1(\mathbb{I},&h_{ij})\Big].
        \label{eq:MonogomyOfEntanglementOnTraingle}
    \end{align}
\end{theorem}
Note that in \cite{par22opt}, the variables are defined by swap operators (\cref{eq:swap_def} in this work), and the above form could be derived by simply using the relation between swap operators and singlet projectors $h_{ij} = (1-p_{ij})/2$.

\begin{lemma} \label{thm:60DegreeAngle}
    For any feasible moment matrix $M_1$ of $NPA_1(\mathcal{Proj})$, indexed by $\{I, h_{ij} \text{ where }i,j \in \{0,1,\ldots,n\} \text{ and } i < j\}$, the angle between any two normalized Gram vectors of indices sharing one vertex, i.e. $h_{ij}$ and $h_{jk}$ where $i,j,k$ are all distinct, is no less than $60^\circ$ and no greater than $120^\circ$.
\end{lemma}

\begin{proof}
    Let $\ket{\mathbb{I}}$, $\ket{h_{ij}}$ be the Gram vectors of $M_1$ corresponding to indices $\mathbb{I}$, $h_{ij}$ for all $i<j$, i.e.,  
    \begin{align}
        M_1(\mathbb{I},h_{ij}) = \braket{\mathbb{I}|h_{ij}}, M_1(h_{ij},h_{kl}) = \braket{h_{ij}|h_{kl}}.
    \end{align} 
    Now recall that we have the following constraints on $M_1$: whenever $AB = CD$ where $A,B,C,D$ are all degree-1 polynomials in singlet projectors, $M_1(A,B) = M_1(C,D)$. From this, $h_{ij}^2 = h_{ij}$ implies that 
    \begin{align}
        M_1(h_{ij},h_{ij}) = M_1(\mathbb{I},h_{ij}).
    \end{align}
    Similarly, the anti-commutation relation for singlet projectors \cref{eq:anticommproj} implies that 
    \begin{align}
        4M_1(h_{ij}, h_{jk}) = M_1(\mathbb{I}, h_{ij}) + M_1(\mathbb{I}, h_{jk}) - M_1(\mathbb{I}, h_{ki}).
    \end{align}
    Starting from \cref{eq:MonogomyOfEntanglementOnTraingle}, we can derive the following.
    \begin{widetext}
    \begin{align}
        \big[M_1(\mathbb{I},h_{ij}) + M_1(\mathbb{I},h_{jk})  -M_1(\mathbb{I},h_{ki})\big]^2 &\leq 4M_1(\mathbb{I},h_{ij})M_1(\mathbb{I},h_{jk}) \\
        \iff \hspace{34mm} 16 M_1(h_{ij}, h_{jk})^2 &\leq 4M_1(\mathbb{I},h_{ij})M_1(\mathbb{I},h_{jk}) \\
        \iff \hspace{43mm} \braket{h_{ij}|h_{jk}}^2 &\leq 4 \braket{\mathbb{I}|h_{ij}} \braket{\mathbb{I}|h_{jk}} = \braket{h_{ij}|h_{ij}} \braket{h_{jk}|h_{jk}} \\
        \iff \hspace{9mm} | \braket{h_{ij}|h_{jk}}|/ \sqrt{\braket{h_{ij}|h_{ij}} \braket{h_{jk}|h_{jk}}} &\leq 1/2.
        \label{eq:derivation60}
    \end{align}
    \end{widetext}
    \Cref{eq:derivation60} implies that the angle between the vectors $\ket{h_{ij}}$ and $\ket{h_{jk}}$ must be between $60^\circ$ and $120^\circ$.
\end{proof}

\begin{theorem}\label{thm:weightedstar}
    The first level of the NPA hierarchy with $\mathcal{Proj}$ solves {\scshape QMaxCut} exactly for any positively weighted star graphs, i.e.,  
    $NPA_1(\mathcal{Proj}(H))=QMC(H)=\mu_{max}(H)$ for

    \begin{align}\label{eq:starham}
    H = \sum_{i=1}^{n} w_i h_{0i},\quad w_i > 0\ ~~\forall i .
    \end{align}
\end{theorem}

\begin{proof}
    Recall that the moment matrix $M_1$ of $NPA_1(\mathcal{Proj})$ program is indexed by $\{\mathbb{I}, h_{ij} \text{ where }i,j \in \{0,1,\ldots,n\} \text{ and } i < j\}$. Let $\ket{\widetilde{\mathbb{I}}}$, $\ket{\widetilde{h_{ij}}}$ be the normalized vectors corresponding to $\ket{\mathbb{I}}$ and $\ket{h_{ij}}$. Restating \Cref{thm:60DegreeAngle}, for any feasible moment matrix $M_1$, the angle between any two normalized vectors of singlet projector indices sharing one vertex ($\ket{\widetilde{h_{ij}}}$ and $\ket{\widetilde{h_{jk}}}$ where $i,j,k$ are all distinct) is no less than $60^\circ$ and no greater than $120^\circ$, i.e.,
    \begin{align}
        \left| \braket{\widetilde{h_{ij}}|\widetilde{h_{jk}}}\right| \leq \frac{1}{2}.
    \end{align}
    The rest of the proof involves showing that when the objective function is a Hamiltonian of the form \cref{eq:starham},
    for the optimal solution, the above inequality saturates for any two normalized vectors with distinct indices from $\left\{h_{0i}\right\}_{i=1}^{n}$
    \begin{align}
        \braket{\widetilde{h_{0i}}|\widetilde{h_{0j}}} = \frac{1}{2} \quad \forall\ 0< i < j \leq n.
        \label{eq: sixty_angle_property}
    \end{align}
    When this equality holds, we can construct an actual quantum state that has the same energy as the objective value of $NPA_1(\mathcal{Proj})$ program. 
    We do this by mapping each of the normalized vectors $\ket{\widetilde{h_{0i}}}$ to a state that has a singlet between $0^\text{th}$ and $i^\text{th}$ qubits, i.e., $(\ket{0}_0\otimes\ket{1}_i - \ket{1}_0\otimes \ket{0}_i)/\sqrt{2}$, and the rest of the qubits forming a maximal total spin state, e.g., all qubits in the spin-up state. 
    This mapping preserves the property \cref{eq: sixty_angle_property} and also determines the normalized vectors corresponding to other indices $h_{ij}$ to be $\ket{\widetilde{h_{ij}}} = \pm (\ket{\widetilde{h_{0i}}} - \ket{\widetilde{h_{0j}}})$ for  $0< i < j$ where the positive or negative sign in the front depending on whether $M_1(\mathbb{I}, h_{0i})$ is greater or less than $M_1(\mathbb{I}, h_{0j})$ respectively.
    Furthermore, when \cref{eq: sixty_angle_property} is satisfied, the $\ket{\widetilde{\mathbb{I}}}$ that maximizes the objective function will also be in the span of $\{\ket{\widetilde{h_{0i}}}\}_{i=1}^{n}$. 
    Then, since the output of the $NPA_1(\mathcal{Proj})$ program is an upper bound on the maximum eigenvalue of the Hamiltonian, this implies that the output of $NPA_1(\mathcal{Proj})$ is exact.

    Let $A$ be a matrix where the $i^\text{th}$ row of the matrix is the weighted vector $\sqrt{w_i}\bra{\widetilde{h_{0i}}}$. 
    Without loss of generality, we can assume that $A$ is of size $n \times n$.
    Let $A = PU$ be its polar decomposition where $P$ is a positive semi-definite matrix and $U$ is an orthogonal matrix.
    We can rewrite the objective function as the following:
    \begin{align}
        M_1(\mathbb{I}, H) = &\sum_{i=1}^{n} w_i M_1(\mathbb{I}, h_{0i}) \nonumber \\
        = &\sum_{i=1}^{n} w_i \braket{\mathbb{I}|h_{0i}} = \sum_{i=1}^{n} w_i \braket{\mathbb{I}|h_{0i}}^2/\braket{\mathbb{I}|h_{0i}}\nonumber \\
        = &\sum_{i=1}^{n} w_i \braket{\mathbb{I}|h_{0i}}^2/\braket{h_{0i}|h_{0i}} = \sum_{i=1}^{n} w_i \braket{\widetilde{\mathbb{I}}|\widetilde{h_{0i}}}^2\nonumber\\
        = &\braket{\widetilde{\mathbb{I}}|A^{T}A|\widetilde{\mathbb{I}}} = \braket{\widetilde{\mathbb{I}}|U^T P^2 U|\widetilde{\mathbb{I}}}.
    \end{align}
    Since we want to maximize the objective function, its value cannot be greater than the maximum eigenvalue of $P^2$, and it is equal to the maximum eigenvalue when $U\ket{\widetilde{\mathbb{I}}}$ is the maximum eigenvector of $P^2$.
    
    The set of constraints that $\left|\braket{\widetilde{h_{0i}}|\widetilde{h_{0j}}} \right| \leq \frac{1}{2}$, where $i\neq j$, can be written as $\left|(AA^{T})_{ij} \right| = \left|(P^2)_{ij} \right| \leq \frac{1}{2}\sqrt{w_i w_j}$, where $i\neq j$, and the $ij$ subscript indicates that it's the $i^{\text{th}}$ row  $j^{\text{th}}$ column element of the particular matrix. 
    The $\ket{\widetilde{h_{ij}}}$ vectors being normalised implies $(AA^T)_{ii} = (P^2)_{ii} = w_i$ for $i \in \{1,2,...,n\}$. Given these constraints on $P^2$, the maximum eigenvalue of $P^2$ is maximized when $(P^2)_{ij} = \frac{1}{2}\sqrt{w_i w_j}$ . To see this, consider $P^2$ with its maximum eigenvector $\ket{v}$ where some of the matrix elements of $P^2$ are negative. Let ${abs}(P^2)$ and $\ket{{abs}(v)}$ be the matrix where we take the element-wise absolute value of the matrix and the vectors. It is easy see that $\braket{{abs}(v)|{abs}(P^2)|{abs}(v)} \geq \braket{v|P^2|v}$, which implies that the maximum eigenvalue of $abs(P^2) \geq $ maximum eigenvalue of $P^2$. When all the elements of a matrix are non-negative, Perron-Frobenius theorem implies that the maximum eigenvalue is a non-decreasing function of each of the individual matrix elements and strictly increasing in the case of irreducible matrix thus implying that the optimal $P^2$ has $(P^2)_{ij} = \frac{1}{2}\sqrt{w_i w_j}$. The Gram vectors of this optimal $P^2$ are precisely $\{\sqrt{w_i} \ket{\widetilde{h_{0i}}}\}_{i=1}^{n}$ which satisfy the property $\braket{\widetilde{h_{0i}}|\widetilde{h_{0j}}} = \frac{1}{2}\,\, \forall i\neq j$. For the optimal $P^2$, the maximum eigenvector $U\ket{\widetilde{\mathbb{I}}}$ is also in the span of the vectors $\{\ket{\widetilde{h_{0i}}}\}_{i=1}^{n}$ and so is $\ket{\widetilde{\mathbb{I}}}$ since the dimension of subspace formed by the span of $\{\ket{\widetilde{h_{0i}}}\}_{i=1}^{n}$ is $n$.
\end{proof}

\subsection{Complete bipartite graphs and some extensions}\label{subsec:compbip}

\begin{figure*}[t]
    \centering
    \includegraphics[width=17cm]{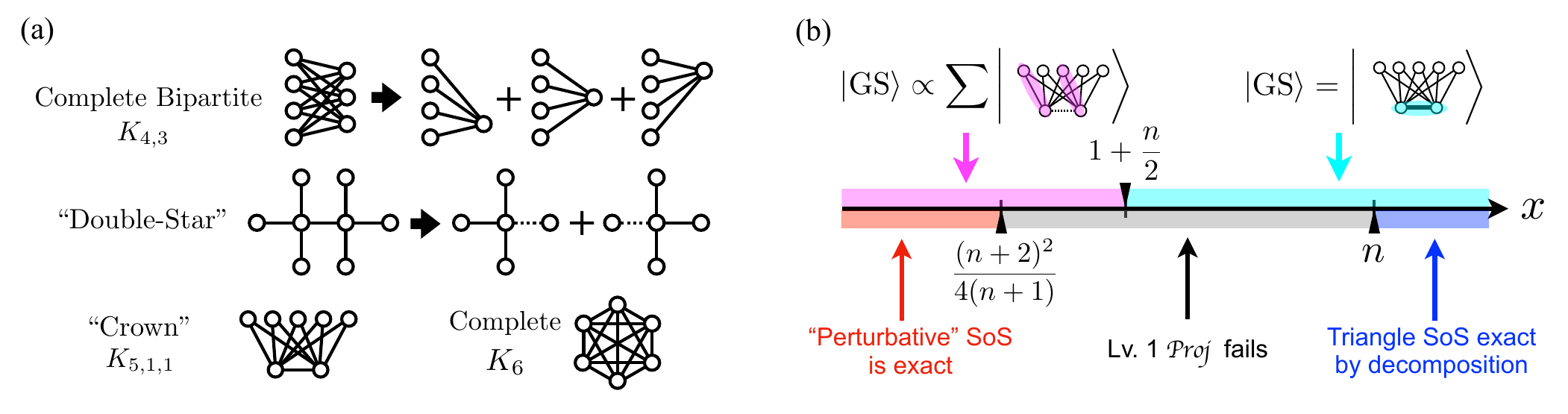}
    \caption{(a) Illustration of graphs where we provide exact $SoS$ in this section and how some of them could be decomposed. (b) The ``phase diagram" of the crown graph with $n\geq 3$, where the true ground state is illustrated on the top, and regions where $SoS_1(\mathcal{Proj})$ gives exact solutions are shown on the bottom (colored in dark red and blue). }
    \label{fig:SmallGraphs}
\end{figure*}

In this section, we show explicit $SoS_1(\mathcal{Proj})$ proofs for several family of graphs (shown in \cref{fig:SmallGraphs} (a)). 
Such exact SoSs provide a rigorous proof that $NPA_1(\mathcal{Proj})$ obtains the exact solution of the corresponding {\scshape QMaxCut} by \cref{thm:duality}. Since the exact SoS we provide here will eventually relate to the frustration-free concept in condensed matter physics (\cref{subsec:cmp}), we will be using the term ``ground state'' even though we are considering the maximum eigenstate. This of course corresponds to the lowest energy state for the Hamiltonian with the opposite sign (i.e., the antiferromagnetic Heisenberg model).

An important technique for demonstrating exact SoS proofs is the decomposition of graphs into smaller graphs leading to a decomposition of the SoS proof into smaller SoS proofs (schematically shown in the figure).  The simplest example of such decomposition arises when considering the SoS for the complete bipartite graph, which decomposes into several uniformly-weighted star graphs $H=\sum_{i=1}^n h_{0i}$. 
The weighted star graph can be solved exactly by $NPA_1(\mathcal{Proj})$ as shown in the previous section, however,  the explicit $\SoS$ cannot be analytically written down in general. The unweighted case however, gives us the simplest case of an exact $\SoS$ : 
\begin{align}\label{eq:StarSOS}
   \left( \sqrt{\frac{n+1}{2}}\mathbb{I} - \sqrt{\frac{2}{n+1}}\sum_{i=1}^{n} h_{0i} \right)^{\bf 2}
   + \frac{1}{n+1}\sum_{1\leq j<k\leq n} h_{jk}^{\bf 2} \nonumber\\
   = \frac{n+1}{2}\mathbb{I} -\sum_{i=1}^{n} h_{0i} 
   = \lambda\mathbb{I} -H.
\end{align}
The simplest way to see that this SoS is indeed exact (i.e., holds with the true ground state energy $\lambda=(n+1)/2$), is to simply compare with the true ground state energy calculation. 
For this Hamiltonian, the ground state is the uniform superpositon of all one-singlet states on top of the edges satisfying $h_{0i}$. All other spins should be pointing in one fixed direction (and thus has $n$-fold degeneracy).

The above $\SoS$ equation could be interpreted as an upper bound for the maximum eigenvalue of $H$. The reasoning is as the following. Since the left hand side is a {\it sum of squares}, it implies that the right hand side is positive semidefinite, i.e., $0\preceq\lambda\mathbbm{I}-H$. 
By reordering, we get $\lambda\mathbbm{I}\succeq H$, which upper bounds the eigenvalue of $H$. 
For this particular case, 
the bound we obtain matches exactly to the actual maximum eigenvalue for the uniform star graph with $n$ edges ($n+1$ qubits in total), thus providing a proper proof of the fact that the above described state is indeed the ground state. 
Note that \cref{eq:StarSOS} could be confirmed straightforwardly by using the anticommutation relation \cref{eq:anticommproj}, and does not require computation of exponentially large matrices. 
Furthermore, by applying the maximum eigenvalue state $|\mathrm{GS}\rangle$ from the left and right to \cref{eq:StarSOS}, 
we obtain 
\begin{align}\label{eq:NullSpace}
   &\langle \mathrm{GS} | \left( \sqrt{\frac{n+1}{2}}\mathbb{I} - \sqrt{\frac{2}{n+1}}\sum_{i=1}^{n} h_{0i} \right)^2 | \mathrm{GS}\rangle \nonumber \\
   &~~~~~~~~~~~~~~~~~~~~~~+ 
   \frac{1}{n+1}\sum_{1\leq j<k\leq n}  \langle \mathrm{GS} | h_{jk}^2 | \mathrm{GS} \rangle \nonumber\\
   &= \frac{n+1}{2} -  \langle \mathrm{GS} | \sum_{i=1}^{n} h_{0i}  | \mathrm{GS}\rangle ~~\mathbf{= 0}, 
\end{align}
since the value $(n+1)/2$ is exactly the maximum eigenvalue of $H=\sum h_{0i}$.
We can see that all the terms appearing as the square on the left hand side must have expectation value 0 for the ground state $|\mathrm{GS}\rangle$, i.e., 
\begin{align}
    \langle \mathrm{GS} | \left( \sqrt{\frac{n+1}{2}}\mathbb{I} - \sqrt{\frac{2}{n+1}}\sum_{i=1}^{n} h_{0i} \right) | \mathrm{GS}\rangle =& 0 \\
    ~\forall~ 1\leq j<k \leq n,~~~~~\langle \mathrm{GS} | h_{jk} | \mathrm{GS} \rangle=&0,
\end{align}
because they are all positive-semidefinite from being squared, and that is the only way to achieve the value 0 after the summation.
In general, any {\it exact} SoS decomposition will always correspond to a list of null spaces of the ground state $|\mathrm{GS}\rangle$ for this reason \footnote{This could be seen as the SoS way of showing that all spins on the same sublattice in a perfect N{\'e}el state should be forming a maximal spin state that has 0 singlet density, a well known fact in condensed matter physics as the Marshall-Lieb-Mattis theorem.}.

Now let us consider the complete bipartite graph with $n+m$ vertices ($n\geq m$). 
The Hamiltonian could be written as
\begin{equation}\label{eq:uniformstar}
    H=\sum_{i\in A, j\in B}h_{ij}
\end{equation}
where we assume that the vertices are divided into two groups $A$ and $B$, with the edge set being $E=\{(i,j)|i\in A, j\in B\}$ and $|A|=n, |B|=m$. 
To our advantage, we can reuse the uniform star SoS because of the decomposition property as follows: 
The maximum eigenvalue of $H$ on $K_{n,m}$ is exactly the same as that of
$K_{n,1}$ (i.e., a star graph with $n$ leaves) multiplied by $m$. 
In other words, 
\begin{equation}\label{eq:CBhamdecomp}
    H = \sum_{i\in B} ~ \biggl( \sum_{j\in A} h_{ij}\biggr)  
    ~~~\text{and}~~~
    \mu_{\max}( H ) = \sum_{i\in B} ~ \mu_{\max} \left( \sum_{j\in A} h_{ij} \right)
\end{equation}
holds simultaneously. 
To prove the second equation, it suffices to construct a SoS decomposition with energy $E$ and also verify that there indeed exists such state with energy $E$. 
In this case, the maximum eigenvalue is achieved by the uniform superposition of all possible $m$-singlet configurations that connect the $A$ and $B$ sublattices without overlaps, and all $n-m$ unpaired spins pointing in the same direction (say, up).

Since we already have \cref{eq:StarSOS}, it is rather easy to confirm that
\begin{align}\label{eq:CompBipSOS}
   \frac{2}{n+1} \sum_{i\in B}\Bigl(\frac{n+1}{2}\mathbb{I}-\sum_{j\in A}h_{ij}\Bigr)^{\bf 2} &+
   \frac{m}{n+1}\sum_{j<k\in A}h_{jk}^{\bf 2} \nonumber\\
  ~=~&\frac{m(n+1)}{2}\mathbb{I}-H,
\end{align}
gives the exact energy for complete bipartite graphs $K_{n,m}$.
Note that this exact SoS is really just $m$ copies of the SoS for the uniform star graph $K_{n,1}$. 
This was possible because, the Hamiltonian itself could be viewed as comprising $m$ copies of $n$-leaved star graphs. The observation here is that if a ``proper decomposition'' of the Hamiltonian such as \cref{eq:CBhamdecomp} exists, the task of finding the exact SoS for the total Hamiltonian is reduced to that of the smaller Hamiltonian in general.

The complete bipartite graph considered here are known as the Lieb-Mattis {\it model} in condensed matter physics \cite{lie62ord,lou19exa}, where the full energy spectrum is well-understood. The Lieb-Mattis {\it theorem} states that Heisenberg models with bipartite graphs (with sublattices $A$ and $B$) have ground states with total spin $\left( |A|-|B|\right) /2$, using the complete bipartite case as a starting point of the proof. 
The $\SoS$ we have here for complete bipartite graphs immediately tells you that the ``singlet density" $\langle h_{ij}\rangle$ among the same sublattice sites will always be 0, just like in the case we have mentioned for the star graph. 
This means that the two sublattices are forming the maximum total spin state, which is equivalent to the claim of the Lieb-Mattis theorem. 
We could say that our SoS is an alternative proof for the Lieb-Mattis theorem, restricted to the case of complete bipartite graphs with uniform weights.

\subsubsection{Crown Graphs}\label{subsubsec:crown}
Graphs with one additional edge to $K_{n,2}$ ($n\geq2$) connecting the two vertices of the B-sublattice 
(i.e. a complete tripartite graph $K_{n,1,1}$) also admits an exact $\SoS$ and thus $NPA_1(\mathcal{Proj})$ obtains the exact maximum eigenvalue as the upper bound.
These graphs, which we call the ``crown" graph (\cref{fig:SmallGraphs} (a)), have maximum eigenvalue $n+1$, the same value for the $K_{n,2}$ complete bipartite graphs. The additional edge does not change the maximum eigenvalue nor the maximum eigenvalue state itself. 

We can modify the SoS in \cref{eq:CompBipSOS} so that the Hamiltonian now includes the one additional edge on the right hand side. 
If we label the two vertices in the B-sublattice to be $a$ and $b$, then the $\SoS$ reads
\begin{align}\label{eq:CrownSOS}
   \frac{2}{n+1} \sum_{k=a,b}\Bigl(\frac{n+1}{2}\mathbb{I}-\sum_{i\in A}h_{ik} -\frac{2+n\pm n}{4}h_{ab}
   \Bigr)^2 \nonumber\\
   + \frac{2}{n+1}\sum_{i<j\in A}h_{ij}^2
   ~=~
   (n+1)\mathbb{I}-H, 
\end{align}
where there is a degree of freedom for the coefficient of $h_{ab}$, coming from two solutions of a quadratic equation. 
We can see that this SoS must also be exact, since it gives the same upper bound value $n+1$ as in the previous exact SoS for $K_{n,2}$ despite having an additional edge in $H$.

The fact that the only difference between this SoS and \cref{eq:CompBipSOS} is the $h_{ab}$ term leads us to the question if this form of SoS is general in some sense. 
Indeed, as it turns out, we can consider a crown graph with the term $h_{ab}$ being weighted with weight $x$, and the above form of the SoS is exact for the entirety of $x\leq (n+2)^2/4(n+1)$. 
The precise $\SoS$ becomes 
\begin{widetext}
\begin{equation}
\hspace{-3mm}
   (n+1)\mathbb{I}-\biggl(\sum_{k=a,b}\sum_{i\in A}h_{ik}+x h_{ab}\biggr)
   =
   \frac{2}{n+1} \sum_{k=a,b}\left(\frac{n+1}{2}\mathbb{I}-\sum_{i\in A}h_{ik} -\frac{2+n\pm \sqrt{(n+2)^2-4(n+1)x}}{4}h_{ab}
   \right)^2 +
   \frac{2}{n+1}\sum_{i<j\in A}h_{ij}^2
   , \label{eq:WeightedCrownSOS}
\end{equation}
\end{widetext}
which has a real solution only when $x\leq (n+2)^2/4(n+1)$. 

We can regard this SoS to be heuristically constructed in two steps. First, the case corresponding to $x=0$ was decomposable as in \cref{eq:CBhamdecomp}, yielding an SoS that retains the symmetry of the graph ($\mathbb{Z}_2$ between $a$ and $b$, and $\mathcal{S}_n$ for the A-sublattice sites). 
Next, when another edge is added also in a symmetry-preserving way, we can have an ansatz for the SoS that also still preserves the symmetry but now also includes the additional term. 
In this sense, the above SoS could be thought of as a ``perturbative" SoS from the complete-bipartite case, because if we gradually increase $x$ from 0, the SoS also can be changed continuously, always being exact. Since $1<(n+2)^2/4(n+1)$, the uniformly weighted crown graph is also exactly solvable, and we can say that the $\SoS$ for the complete bipartite graph and the crown graph are {\it adiabatically connected}. 
Intuitively, when $x$ is small enough, the ``physics" should not change a lot from the $x=0$ case, and in this case we can show that the ``radius of convergence" extends to $x \leq (n+2)^2/4(n+1)$, including $x=1$. 

The fact that the ansatz fails alone does not necessarily imply that no exact SoS exist, but it does suggest that even {\it if} such $\SoS$ exist, it will look very different from the SoS in the $x\leq(n+2)^2/4(n+1)$ region. 
As a matter of fact, we numerically observe that $NPA_1(\mathcal{Proj})$ starts to have nonzero error exactly from $x=(n+2)^2/4(n+1)$, implying that such a SoS proof indeed does not exist.  

Conversely, when we increase $x$ large enough, $NPA_1(\mathcal{Proj})$ starts to obtain the exact ground state energy again starting from $x \geq n$. Intuitively, in the $x\rightarrow\infty$ limit, the ground state should trivially become a state where there is simply one singlet placed for $h_{ab}$, and it seems natural for an SDP algorithm to be able to obtain such a simple state exactly. 
This intuition could be made rigorous by noticing that when $x\geq n$, the Hamiltonian regains the decomposition property, but now into triangles: 
\begin{align}\label{eq:crowndecomp}
H&=    \sum_{i\in A} ~~\biggl( h_{ia}+h_{ib} + \frac{x}{n} h_{ab} \biggr), \nonumber \\
\mu_{\max}(H) &=    \sum_{i\in A} \mu_{\max}\left( h_{ia}+h_{ib} + \frac{x}{n} h_{ab} \right) .
\end{align}
Since a triangle with weight $(1,1,\alpha)$ has the exact $\SoS$ of
    \begin{align}\label{eq:StrongTriangleSoS}
    &\left(\alpha+\frac{1}{2}\right)\mathbb{I} - (h_{12}+h_{23}+\alpha h_{13}) 
    = \left(\alpha+\frac{1}{2}\right)\nonumber\\
    \times&\biggl\{
    \mathbb{I}
    -\hspace{-3mm}
    \sum_{1\leq i<j\leq 3}\hspace{-3mm}
    \frac{8\alpha+4\pm(3(-1)^{i+j}-1)\sqrt{2\alpha (2\alpha -1)-2}}{6+12\alpha} h_{ij}
    \biggr\}^2\hspace{-1mm},
    \end{align}
for $\alpha\geq 1$, together with the decomposition, this can be turned into an exact $\SoS$ for the crown graph when $x\geq n$. 
Note that again, the SoS is not unique, and it has a degree of freedom in choosing $\pm$ to be fixed. 
When $\alpha<1$ the above form no longer gives a real coefficient.  However, the true ground state of the triangle also changes, and still allows an exact $\SoS$: 
\begin{align}\label{eq:WeakTriangleSoS}
    &\frac{3}{2}\mathbb{I} - (h_{12}+h_{23}+\alpha h_{13})
    \nonumber\\
    &= \frac{3}{2}\left(
        \mathbb{I} 
        -\frac{2}{3}h_{12}-\frac{2}{3}h_{23}
        -\frac{2\pm\sqrt{6(1-\alpha )}}{3}h_{13}\right)^2 ,
\end{align}
which again only gives valid coefficients for $\alpha\leq 1$. 
For our current objective of constructing a $\SoS$ for the crown graph, the existence of SoS for $\alpha<1\Leftrightarrow x<n$ does not help since the Hamiltonian no longer has the decomposition \cref{eq:crowndecomp}. 

Again, like the case for the small $x$ region, although this decomposition is just {\it one} possible heuristic method for finding the exact SoS, it turns out that the $NPA_1(\mathcal{Proj})$ does start to fail exactly for $x<n$. 
Furthermore, it is possible to prove this failure for the region $(n+2)/3<x<n$ which rigorously establishes the right-side boundary at $x=n$ but leaves an unproved open space for the left-side boundary at $x=(n+2)^2/4(n+1)$. We provide this proof in Appendix \ref{app:crown}. 
The situation for the whole $x\in\mathbb{R}$ is illustrated in \cref{fig:SmallGraphs} (b). It is rather intriguing that the ``phase transition" points for the SDP ($x=(n+2)^2/4(n+1)$ and $n$), and the phase transition for the true ground state ($x=1+n/2$) are well-separated\knote{edited sentence}. This means that there are broad regions of the $x$ parameter where the SDP algorithm fails despite having exactly the same ground state as other points where SDP succeeds, which interestingly seems to be caused by the lack of real solutions in a quadratic equation \cref{eq:StrongTriangleSoS}.
In \cref{subsubsec:MG} and \cref{subsubsec:SS}, we will see more nontrivial phase transitions in condensed matter physics models.

\subsubsection{Double Star Graphs}
While the crown graphs do not have the nice decomposition property that the complete bipartite graphs had, the double-star graphs have such a decomposition into two weighted star graphs. 
The double-star graphs are the ones with $n$ vertices connected to one vertex $a$, and the other set of $n$ vertices all connected to the other vertex $b$, and having an edge between $a$ and $b$ (thus $2n+2$ vertices in total). 

In this case, the decomposition works as 
\begin{align}\label{eq:stardecomp}
     H 
    &= ~\sum_{x=a,b}~~\left( \frac{1}{2}h_{ab}+\sum_{i\in\partial x}^{n} h_{xi} \right), \nonumber\\
    \mu_{\max}( H )
    &= \sum_{x=a,b}\mu_{\max}\left( \frac{1}{2}h_{ab}+\sum_{i\in\partial x}^{n} h_{xi} \right),
\end{align}
and the SoS reduces to the case of a weighted graph (with only one edge having weight $1/2$).
While the existence of exact $\SoS$ is provable for arbitrary weighted star graphs (\cref{thm:weightedstar}, \cite{per_comm}), for the particular case corresponding to the double star we can have relatively simple analytical forms:
\begin{eqnarray}
&&\sum_{\substack{(x, y) =\\ (a,b),(b,a)}}
\biggl\{
\frac{E}{2}\biggl( \mathbb{I} - {2\alpha} \sum_{i\in\partial x} h_{ix} - {\alpha} h_{ab} 
    \biggr)^2 
    + {\alpha}\hspace{-1mm}\sum_{i < j\in\partial x}\hspace{-1mm}h_{ij}^2 \nonumber\\
   &&~~~~~~
   +\frac{1-\alpha}{2n+1}\sum_{i\in\partial x}
    \biggl( h_{iy} - \left(n+1-\sqrt{n(n+2)}\right)h_{ix} \nonumber\\ 
  &&~~~~~~~~~~~~~~~~~~~~~~~~
  +\frac{1}{2E-2}h_{ab}\biggr)^2 \biggr\} 
   = E\mathbb{I} - H, \label{eq:dssos}
\end{eqnarray}
where $E$ is the maximum eigenvalue $E=(n+2+\sqrt{n(n+2)})/2$, $\alpha = 1-\sqrt{n/(n+2)}$, and $\sum_{i\in\partial x}$ indicates summation over all vertices $i$ that are the $n$ leaves adjacent to $x=a$ or $b$. 
For this case, it is easy to confirm that a superposition of ``singlet on one edge'' states will achieve this eigenvalue, implying that this SoS is exact (tight).

While the above $\SoS$ shows that $NPA_1(\mathcal{Proj})$ obtains the ground state energy exactly for the double stars, the following $SoS_2(\mathcal{Proj})$ is simpler in form : 
\begin{eqnarray}\label{eq:lv2sos}
    \hspace{-5mm}
    \sum_{\substack{(x, y) =\\ (a,b),(b,a)}}\biggl\{\biggl(\alpha \mathbb{I} - \beta \sum_{i\in\partial x} h_{ix} - \gamma h_{ab} 
    +\delta \sum_{i\in\partial x} h_{iy}
    \biggr)^2 \nonumber\\
    \hspace{-4mm}+\hspace{-2mm}
    \sum_{i< j\in\partial x}\hspace{-1mm}
    \left( \frac{\beta^2+\delta^2}{2}h_{ij}^2 + 2\beta\delta ( h_{ix}h_{jy} )^2\right) \biggr\} 
    =
    E
    \mathbb{I} - H
    ,
    \end{eqnarray}
where the specific coefficients are 
$\alpha=\sqrt{n+2+s}/2$, $\beta=(2\alpha+S)/(n+2)$, $\gamma=(2\alpha-(n+s)S/2)/(n+2)$, $\delta=((4\alpha^2-1)S-2\alpha)/(n+2)$,
with 
$s=\sqrt{n(n+2)}$, and $S=\sqrt{(1+2/n)^{1/2}+s-n-2}$.
Note that the $\SoS$ we provide here could again be viewed as an extension of the $\SoS$ for the complete bipartite case, just by adding another term to \cref{eq:CrownSOS}. 
This $SoS_2(\mathcal{Proj})$ Eq. (\ref{eq:lv2sos}) is technically a weaker result compared to the $\SoS$ in Eq. (\ref{eq:dssos}) since we can get the exact eigenvalue already at the first level. In other words, Eq. (\ref{eq:lv2sos}) only implies that the $NPA_2(\mathcal{Proj})$ succeeds while Eq. (\ref{eq:dssos}) shows that $NPA_1(\mathcal{Proj})$ already suffices. However, Eq. (\ref{eq:lv2sos}) straightforwardly shows that the ``two-singlet density" $\langle h_{ai}h_{bj}\rangle$ is always 0 in the double-star, a piece of information that was not obvious from the level-1 SoS Eq. (\ref{eq:dssos}).

Interestingly, $NPA_1(\mathcal{Proj})$ starts to {\it fail} once the ``double star" becomes imbalanced, 
i.e. having different number of leaves on the two sides. 
This implies that the decomposition of the double graph \cref{eq:stardecomp} only holds for very precise cases with balanced double graphs and does not exist in general.

\subsection{Complete graphs: Contrast between even and odd}\label{subsec:complete}
While the complete graphs $K_n$ do not admit similar decomposition as in \cref{eq:stardecomp}, 
we can still obtain the exact $\SoS$ by exploiting the high symmetry of the graph -- if the number of vertices $n$ is even: 
\begin{align}\label{eq:EvenCompSOS}
    \sum_{i=1}^n 
    \left(
    \sqrt{\frac{n+2}{8}}\mathbb{I} - \sqrt{\frac{2}{n+2}}\sum_{j\neq i} h_{ij}
    \right)^2 \nonumber\\
    = 
    \frac{n(n+2)}{8}\mathbb{I} - \sum_{1\leq i<j\leq n} h_{ij}. 
\end{align}
Here again, the SoS is essentially a summation of the SoS for star graphs, but with slightly different coefficients, which makes them different from the simple decompositions we have been seeing. 

The situation becomes quite different when the number of vertices is odd. The maximum eigenvalue is $(n+3)(n-1)/8$, but $NPA_1(\mathcal{Proj})$ gives $n(n+2)/8$ as the upper bound, which is $3/8$ bigger (observed numerically). 
We can see that for the odd case the $NPA_1(\mathcal{Proj})$ must do {\it at least as good as} $n(n+2)/8$ from the fact that the SoS we have above works perfectly fine even when $n$ is odd. 

Ideally for odd $n$, the exact SoS should give 
\begin{align}\label{eq:xyz_sos_proof}
\hspace{-5mm}
    &\frac{(n+3)(n-1)}{8}\mathbb{I} - \sum_{1\leq i<j\leq n} h_{ij} \succeq 0 \nonumber\\
\hspace{-5mm}    ~~\Leftrightarrow ~~&
    \sum_{i<j} \left( X_i X_j + Y_i Y_j + Z_i Z_j \right)+ \frac{3n-3}{2}\mathbb{I} \succeq 0. 
\end{align}
Let $\ell^*$ be the smallest integer such that $NPA_{\ell^*}(\mathcal{Pauli}(H))=\mu_{max}(H)$.  Since $NPA_\ell(\mathcal{Pauli})$ converges at $\ell=n$ we know $\ell^* \leq n$.
By exploiting the $SU(2)$ symmetry of the LHS, we can see that obtaining a degree-$\ell$ SoS proof for
\begin{equation}\label{eq:z_sos_proof}
    \sum_{i<j} Z_i Z_j + \frac{n-1}{2}\mathbb{I} \succeq 0 
    ~~\Leftrightarrow ~~
    \left(\sum_{i=1}^n Z_i \right)^2
    \succeq \mathbb{I} 
\end{equation}
would be a {\it sufficient} condition for showing
that $NPA_\ell(\mathcal{Pauli})=\mu_{max}\left( \sum_{i <j } h_{ij}\right)$.  
Since the Pauli operators $Z_i$ all commute,  
the problem essentially becomes classical and could be regarded as a {\scshape MaxCut} instance for the same odd complete graph. 
The problem then is equivalent to proving the following statement with SoS: 
\begin{center}
    {\it When you have odd numbers of $\pm 1$ values, \\
    their sum can never become 0.}
\end{center}
This trivial statement about parity becomes surprisingly hard to prove with SoS and is known to require $\lceil n/2 \rceil$-degree SoS \cite{gri01lin,lau03low,kun22spec}, so $\ell^* \leq \lceil n/2 \rceil$. 
While we believe that the same is most likely to be true for our case ($\ell^* =\Omega(n)$)\footnote{The tight SoS proof for {\scshape QMaxCut} on odd complete graphs can be reasonably named as the quantum version of the parity problem mentioned in the references.}, we were only able to prove the impossibility with $\SoS$. 

\begin{theorem}
    $NPA_1(\mathcal{Proj}(H))=n(n+2)/8$ for complete graphs with $n$ vertices, which gives the exact maximum eigenvalue when $n$ is even and is exactly $3/8$ larger than the exact maximum eigenvalue $(n+3)(n-1)/8$ when $n$ is odd.
\end{theorem}

\begin{proof}
    We show that the following constructed $M_1$ is a feasible solution for $NPA_1(\mathcal{Proj})$ that achieves the value $n(n+2)/8$. Together with the $\SoS$ in \cref{eq:EvenCompSOS}, this proves that the $NPA_1(\mathcal{Proj})$ gets the optimal value $n(n+2)/8$. 

    Now, consider the following moment matrix
\begin{equation}
M=
\begin{bmatrix}
1      & a     & a     & \ldots & a      & \ldots \\
a      & a     & a/4   & \ldots & a/4    & \ldots \\
a      & a/4   & a     & \ldots & b      & \ldots \\
\vdots &\vdots &\vdots & \ddots &\vdots  &        \\
a      & a/4   & b     & \ldots & a      & \ldots \\
\vdots &\vdots &\vdots &        &\vdots  & \ddots 
\end{bmatrix},
\end{equation}
with 
\begin{equation}
    a=\frac{n+2}{4(n-1)}, ~~~ b=\frac{n(n+2)}{16(n-1)(n-3)},
\end{equation}
where the shown rows and columns are indexed by operators $\mathbb{I}, h_{12}, h_{13}, \ldots, h_{24}, \ldots$. In other words, 
\begin{equation}\label{eq:oddcompmmnew}
    M(h_{ij},h_{kl})=
    \begin{cases}
        a,   & (ij)=(kl), \\
        a/4, & (ij) ~\mathrm{and}~ (kl) \mathrm{~have~exactly~one~overlap},\\
        b,   & (ij) ~\mathrm{and}~ (kl) \mathrm{~have~no~overlaps}.
    \end{cases}
\end{equation}
It is easy to verify that this moment matrix has size $(1 + {n \choose 2})\times (1 + {n \choose 2})$, achieves energy $a\times {n \choose 2} = n(n+2)/8$, and satisfies the anti-commutation relation constraint: $((a+a-a)/2)/2=a/4$.

All we need to do now is to show $M\succeq 0$, and we do this by constructing Gram vectors of $M$ \footnote{Alternatively, one can list all the eigenvalues of $M$ to show positive semidefiniteness, which has been the more traditional way to prove analogous results for the classical case \cite{gri01lin}. For completeness, we provide this in Appendix \ref{app:compeigen}.}. 
Specifically, we construct $1 + {n \choose 2}$ column vectors $\ket{\mathbb{I}}$ and $\{\ket h_{ij}\}$ for all $i,j\in [n]$ with $i<j$. 
Each column vector's elements are also indexed with the operators $\mathbb{I}$ and $h_{ij}$ as well, which we will denote as the subscript below.
We can then express the Gram vectors in the following way:
\begin{eqnarray}
    \ket{\mathbb{I}}_{\hat O}&=& 
    \begin{cases}
        1, & \text{ if } \hat{O}=\mathbb{I}\\
        0, & \text{otherwise},
    \end{cases}\label{eq:oddcompvec1}\\
    \ket{h_{ij}}_{\hat{O}}&=& 
    \begin{cases}
        a, & \text{ if } \hat{O}=\mathbb{I}\\
        \alpha, & \text{ if } \hat{O}=h_{ij}\\  
        \beta, & \text{ if } \hat{O}=h_{kl} \text{ (no overlap with $ij$) }\\  
        \gamma, & \text{ if } \hat{O}=h_{jk} \text{ (exactly one overlap with $ij$) },
    \end{cases}\label{eq:oddcompvec2}
\end{eqnarray}
with 
\begin{eqnarray}
    \alpha&=&\frac{\sqrt{3(n-3)(n^2-4)}}{4\sqrt{{(n-1)^3}}} , ~~
    \beta=\frac{\sqrt{3(n+2)}}{2\sqrt{(n-1)^3(n-2)(n-3)}} , ~~\nonumber\\
    \gamma&=&-\frac{\alpha}{(n-2)}.
\end{eqnarray}
It is straightforward to confirm that these vectors Eq.~(\ref{eq:oddcompvec1}) and Eq.~(\ref{eq:oddcompvec2}) are indeed Gram vectors for the moment matrix  (\cref{eq:oddcompmmnew}) by a counting argument, 
thus 
concluding that $M$ is the optimal $M_1$ of $NPA_1(\mathcal{Proj})$ achieving the value $n(n+2)/8$. 
\end{proof}

We can observe that the moment matrix that SDP creates is essentially ``blind to the fact that $n$ is an integer" \cite{gam22dis}
and is the reason for obtaining the wrong value $n(n+2)/8$. 
This is the energy you would get when you naively plug in an odd number to the formula for even complete graphs. 
Motivated by this fact and realizing that most of the higher order terms in the higher level moment matrix would reduce to lower degree moments (just like $a/4$ in the example above), 
we conjecture that the only independent moment matrix elements in higher levels would be 
\begin{equation}
    \left\langle 
    \overbrace{h_{ij}\cdot h_{st}\cdot \ldots}^{k \text{ independent op.s}}
    \right\rangle 
    = \prod_{l=0}^{k-1}\left( \frac{n+2-2l}{4(n-2l-1)}\right), 
\end{equation}
which is the formula for an even complete graph, but simply formally replacing $n$ with an odd number, resembling the classical case \cite{gri01lin,lau03low,kun22spec}. 
All other matrix elements would be calculable from the projector algebra constraints.

\section{Numerical results}
While the SoS proofs in the previous section only cover a very small fraction of possible uniformly weighted graphs, the SDP algorithm actually solves surprisingly many graphs exactly, in the sense that the obtained upper bound value matches the exact maximum eigenvalue. This is true for both the $\mathcal{Pauli}$ and $\mathcal{Proj}$ SDP relaxations, and in this section we will go through the numerical results showing this. 
We further observe that the SDP algorithm can be used for calculating expectation values of operators that are of physical interest. 
This is demonstrated in section \ref{subsec:HeisenbergChain}, where the $NPA_1(\mathcal{Proj})$ is applied to the Heisenberg Chain up to size $L=60$, and the critical correlation functions show the correct criticality up to error bars. 

In the following of this section, the term ``solve exactly" means that the upper bound value obtained by SDP theoretically matches exactly with the maximum eigenvalue.

\subsection{Exhaustive numerical results on small graphs}

\begin{figure*}[t]
    \centering
    \includegraphics[width=17cm]{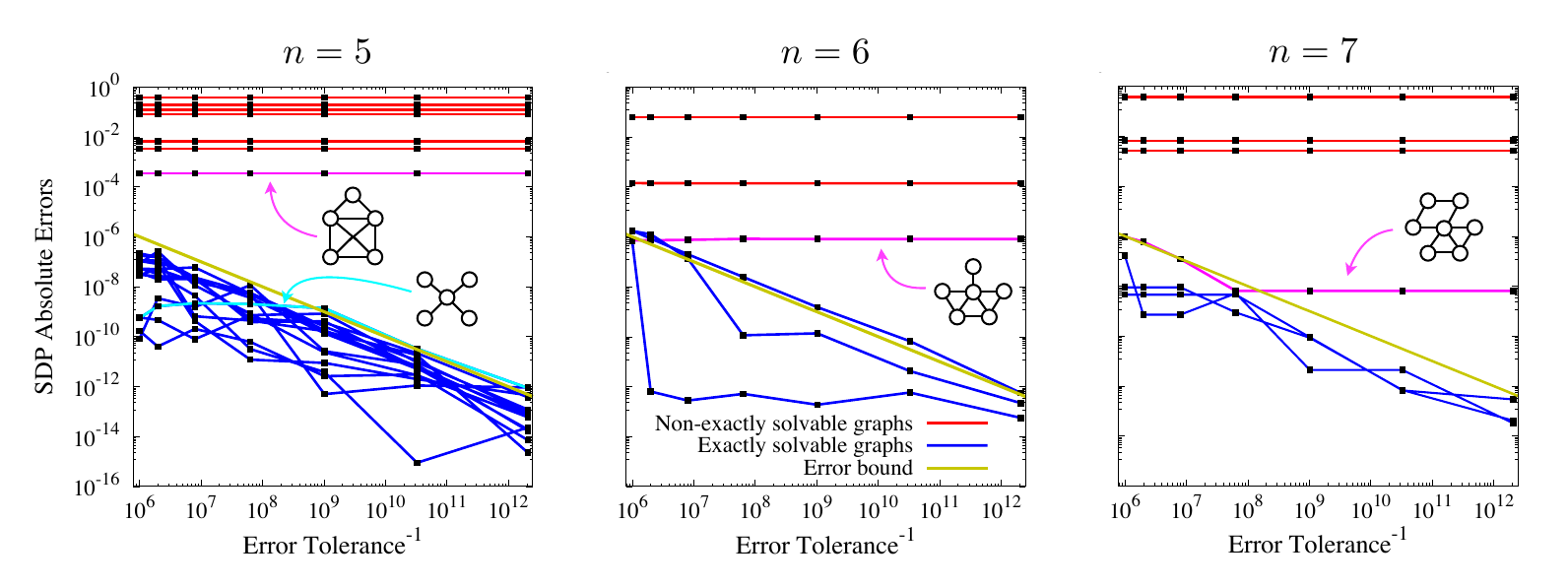}
    \caption{The absolute error of the energy with $NPA_2(\mathcal{Pauli})$ as a function of the tolerance parameter. For graphs with $n=5$ vertices, we show the curves for all possible 21 (connected) 5-vertex graphs, and for $n=6$ and 7, we only show a few randomly chosen graphs, since the number of the graphs becomes enormous (112 and 853 each). 
    Blue and red lines correspond to cases exactly solvable and not respectively, and the magenta lines show the smallest errors that are nonzero.}
    \label{fig:Precision}
\end{figure*}

Here, we show the results of $NPA_2(\mathcal{Pauli})$ applied to all possible uniform graphs up to $n=8$ vertices. 
The main observation is that $NPA_2(\mathcal{Pauli})$ is exact for many graphs with $n\leq 6$ vertices. While the percentage of such graphs seems to shrink as we go to larger system sizes, it suggests that there are many cases where an exact SoS exists that are not covered in the previous section. 

\subsubsection{Probing exact solvability numerically}

Before presenting the main numerical results, here we address the subtle issue arising from numerical precision of the SDP algorithm. 
It is fundamentally impossible to determine whether the SDP algorithm obtains the correct energy value for a particular Hamiltonian solely from numerical results. 
This is because the SDP algorithm always requires a precision parameter which is usually referred to as ``error-tolerance" $\epsilon$, and the algorithm only optimizes up to that $\epsilon$. 
Even if the algorithm seems to give very close values to the true energy we cannot {\it a priori} conclude if that is actually obtaining the exact solution, or if the error of the algorithm is merely small yet non-zero.

To address this issue systematically, we analyzed the optimal value (upper bound of the maximum eigenvalue) obtained by the SDP algorithm as a function of the error tolerance. 
More precisely, as shown in Fig. \ref{fig:Precision}, we plot the discrepancy of the SDP-obtained optimal value and the exact maximum eigenvalue $\Delta E :=|E_{\mathrm{SDP}} - E_{\mathrm{GS}}|$ as a function of $\epsilon^{-1}$. 
This plot, especially for $n=5$, shows a very clear dichotomy of $n=5$ connected graphs. While 7 graphs (red curves) have an almost constant $\Delta E >0$, the rest of the 14 graphs (blue curves) show a decay in $\Delta E$, roughly proportionally to $\epsilon$. 
This could be regarded as strong numerical evidence that the 14 graphs are exactly-solvable instances by the SDP algorithm while the 7 graphs are not. It is quite surprising that a simple five-vertex graph can naturally yield a very small error value around $0.00034$ (the graph shown in Fig. \ref{fig:Precision} with arrows in magenta). 

However, we must note that this method is not entirely decisive. 
As depicted in the center and right panels of Fig. \ref{fig:Precision}, the dichotomy becomes less clear as we go to larger sizes $n=6,7$, and is even worse for $n=8$ (not shown). 
This is because as we proceed to larger system size, an unweighted graph can potentially have extremely small error values $\Delta E$, such as $\sim 10^{-8}$ and even smaller. 
At some point, it practically becomes impossible, since smaller error tolerance $\epsilon$ requires longer iterations in the SDP optimization. 

We can also see that the theoretical error bound of $\Delta E < \epsilon$ (drawn in yellow lines in the figure) for any exactly-solvable graph, is not necessarily satisfied always. For example, although we rigorously prove that the star graph is exactly solvable by the SDP algorithm (see \S\ref{subsec:compbip}), the error of the star graph in Fig. \ref{fig:Precision}, $n=5$ (in cyan) is slightly above the error tolerance $\epsilon$. 
This arises from subtleties in how the error tolerances are handled inside the SDP package, and is difficult to control in general. 

Despite these subtleties, the behavior of the absolute energy error $\Delta E$ as a function of the error tolerance $\epsilon$ serves as a good rule of thumb for distinguishing exactly-solvable graphs from instances with merely small errors. 
For instance, we can be fairly confident that the graphs with magenta arrows indeed do have extremely small but non-zero errors such as $\sim 10^{-6}$.

\subsubsection{Exactly solvable small graphs and their statistics}\label{subsubsec:smallstat}

\begin{figure*}[t]
    \centering
    \includegraphics[width=15cm]{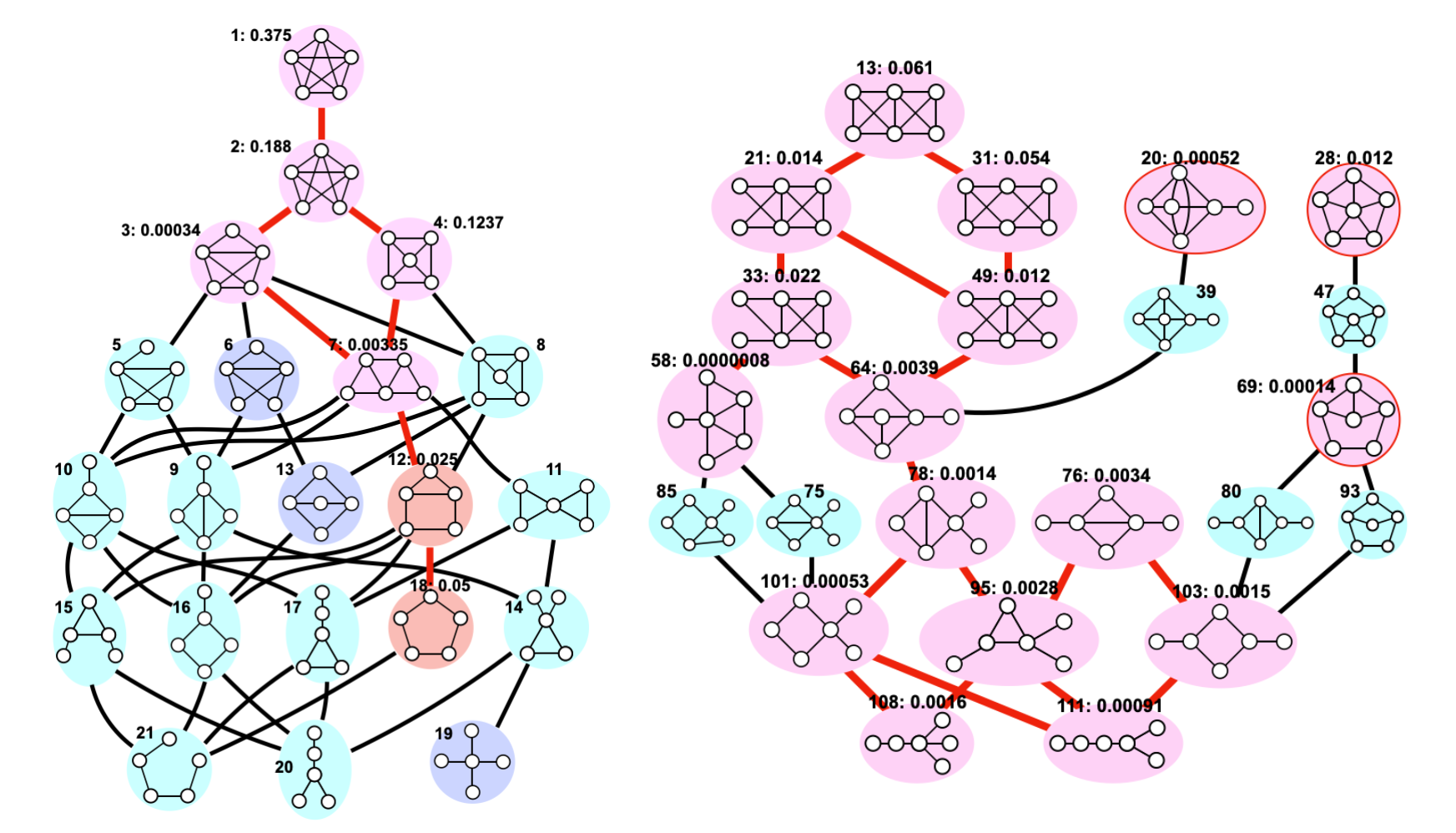}
    \caption{All graphs with $n=5,6$ which the level-2 Pauli basis SDP algorithm fails to obtain the true ground state energy, colored in red/magenta. The labeling of the graphs follows \cite{cve84ata}, and the values following the labels represent the error value $\Delta E$ from Lv. 2 Pauli SDP. For the $n=5$ diagram, we further distinguish graphs where we know an explicit construction of SoS (dark blue) and graphs where Lv. 2 Pauli SDP still fails even when adding a constraint on the total spin $S^2\geq3/4$ (red). }
    \label{fig:MetaGraphs}
\end{figure*}

Once we can confidently determine whether or not the SDP algorithm obtains the true ground state energy, we can start to ask questions such as ``When and how often does the SDP algorithm give us the exact solution?".
To address this question, we present an exhaustive study for all connected graphs with $n=5,6,7$ and $8$ vertices.

Figure \ref{fig:MetaGraphs} shows all of the 7 (out of 21) $n=5$ connected graphs and the 17 (out of 112) $n=6$ connected graphs that the $NPA_2(\mathcal{Pauli})$ SDP algorithm fails to obtain the exact ground state energy (colored in red/magenta). The numbers are labeling of the graphs according to a convention introduced in \cite{cve84ata}. 
It is rather surprising that the algorithm obtains the exact ground state energy for the vast majority of the graphs (colored in blue/cyan) up to this system size, noting that for most of the graphs the SoS is unknown and most likely very complicated (graphs in cyan). 

The figure also shows the topological relations of the graphs, by connecting them with a thick bond whenever two graphs only differ by one edge. 
In this way, we can see that for $n=5$ the red/magenta graphs (SDP fail) form one cluster. In other words, any two $n=5$ connected graphs that Lv. 2 Pauli SDP fails, can be transformed into one from the other by adding and subtracting one edge at a time, always maintaining the SDP algorithm to be failing. 
This is not the case for $n=6$, where the magenta graphs seem to form one big cluster and also three disconnected ``islands" (namely, graphs 20, 28, and 69, as indicated by the red circles). However, as we will see in the following, the ``single-clusteredness" of the hard graphs recovers once we focus on the errors from the $NPA_1(\mathcal{Proj})$. 

The ``failing cluster" includes the complete graph for $n=5$ but not for $n=6$. This is exactly as expected as we explained in \cref{subsec:complete}. This raises the question whether we can actually further constrain the SDP algorithm, not with a higher level, but simply by adding a constraint corresponding to the minimum total spin of the ground state. More specifically, the constraint would be 
\begin{equation}
    \sum_{1\leq i<j\leq n} M(\mathbb{I}, h_{ij})\leq \frac{(n+3)(n-1)}{8}, 
\end{equation}
from \cref{eq:xyz_sos_proof} for odd $n$. 
When we add this constraint, $NPA_1(\mathcal{Proj})$ not only was able to solve the complete graph $K_5$ exactly, but other graphs in the vicinity. 
This information is indicated in \cref{fig:MetaGraphs}, by showing graph 12 and 18 in red, being the only two graphs that $NPA_1(\mathcal{Proj})$ with this additional constraint still failed. 
Note that we cannot do the same thing when we have even number of qubits, because $NPA_1(\mathcal{Proj})$ already succeeds for the complete graphs, i.e., already {\it knows} about this constraint on total spin.


\begin{figure}
    \centering
    \subfigure{\includegraphics[width=8cm]{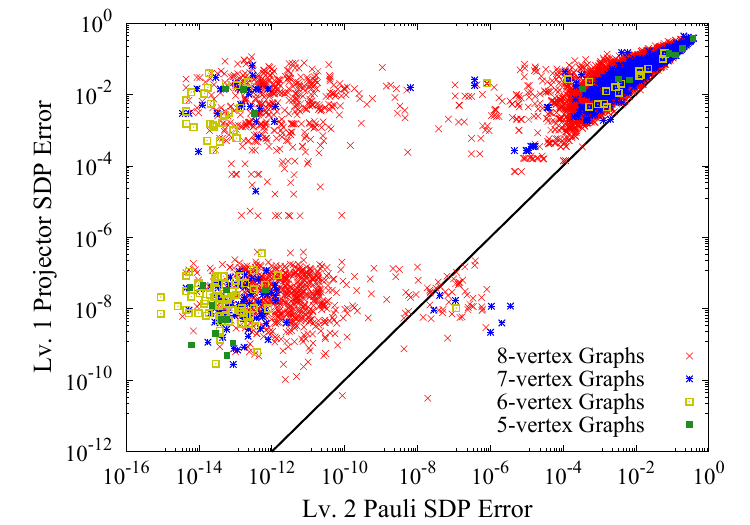}}
    \subfigure{\includegraphics[width=7cm]{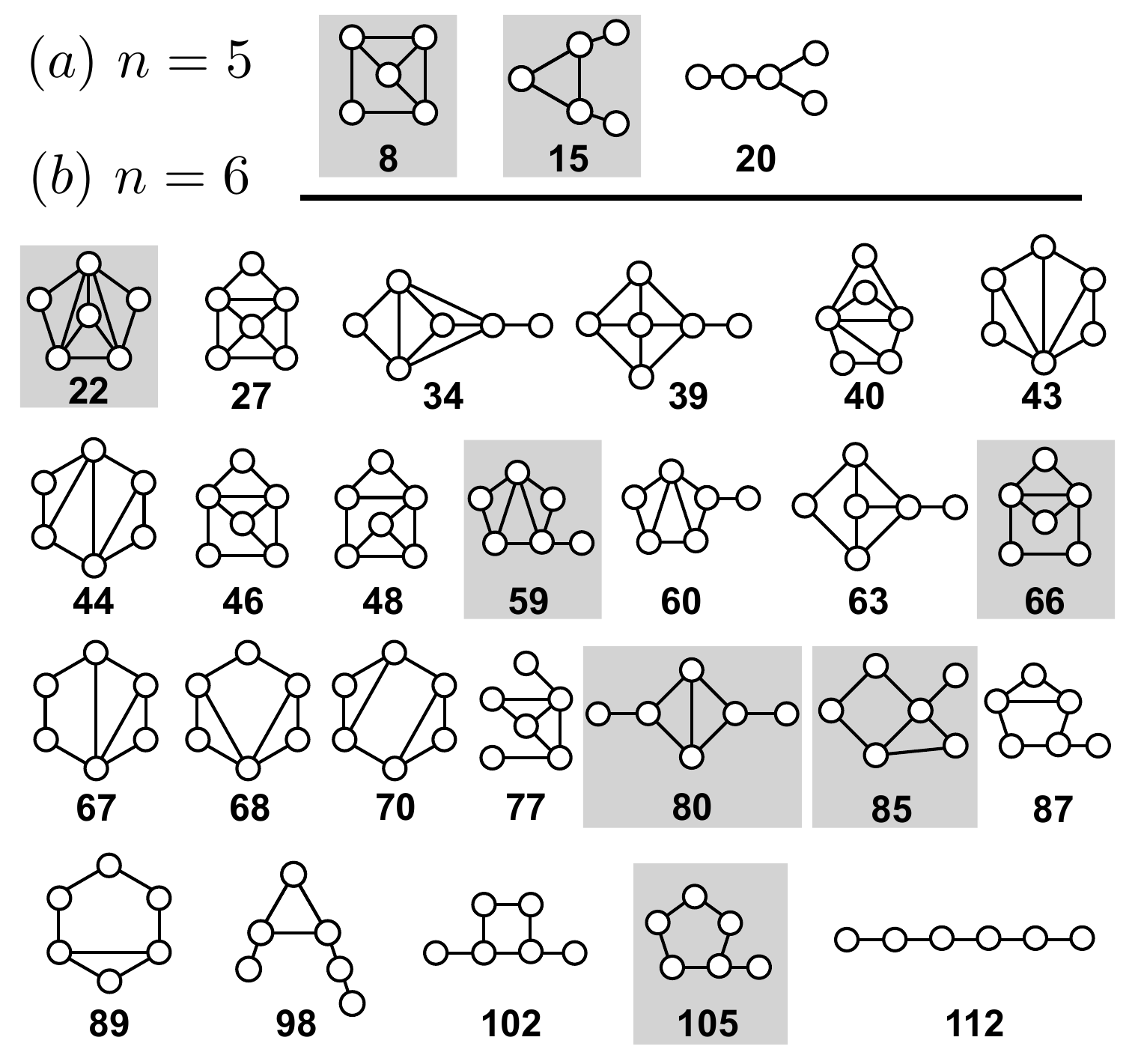}}
    \caption{ The energy errors from the SDP algorithm in two different bases for all graphs up to $n=8$ (top) and list of all (a) $n=5$ and (b) $n=6$ graphs that $NPA_2(\mathcal{Pauli})$ is exact but $NPA_1(\mathcal{Proj})$ is not (bottom). The shaded graphs are the ones which $NPA_2^{\mathbb{R}}(\mathcal{Pauli})$ fails as well (it succeeds for all other graphs listed here). The labeling of the graphs follows \cite{cve84ata}.}
    \label{fig:PauliProjCompare}
\end{figure}

We also compare the performance of the different SDP algorithms ($NPA_2(\mathcal{Pauli})$, $NPA_2^{\mathbb{R}}(\mathcal{Pauli})$, and $NPA_1(\mathcal{Proj})$) for all of these graphs up to $n=8$ in Fig. \ref{fig:PauliProjCompare}. 
The scatter plot shows the energy errors for $NPA_2(\mathcal{Pauli})$ and $NPA_1(\mathcal{Proj})$. The fact that the scattered points roughly forms four different clusters could be understood in the following way. 

Firstly, the cluster on the top right corresponds to graphs that the SDP algorithms with either bases fail to obtain the exact ground state. 
If we believe in typical hardness of the random {\scshape QMaxCut} instances, the ratio of graphs in this cluster in the scatter plot should reach 1 in the large problem size limit. 
The fact that all of the points in this cluster are on the left of the black line indicating $x=y$ reflects the fact that the $NPA_2(\mathcal{Pauli})$ SDP can never perform worse than the $NPA_1(\mathcal{Proj})$ SDP. This could be easily seen from the fact that you can always convert an SoS proof using degree-1 polynomials of projectors into SoS that uses degree-2 Pauli polynomials, but not necessarily the other way around.

Whether the aforementioned inequality $NPA_2(\mathcal{Pauli}(H))  \leq  NPA_1(\mathcal{Proj}(H))$ is actually an equality or not for {\scshape QMaxCut} instances is  not obvious until we actually see examples. 
The second cluster on the top left of \cref{fig:PauliProjCompare} reflects exactly that there are indeed graphs where $NPA_2(\mathcal{Pauli})$ SDP is exact but $NPA_1(\mathcal{Proj})$ SDP fails, i.e., that the inequality is {\it strict} in general. 
We list up all the $n=5$ and 6 graphs that fall under this second cluster on the right side of \cref{fig:PauliProjCompare}. 
Furthermore, we also checked how $NPA_2^{\mathbb{R}}(\mathcal{Pauli})$ SDP performs on these graphs to find that the inequality $NPA_1(\mathcal{Proj}(H)) > NPA_2^\mathbb{R}(\mathcal{Pauli}(H)) > NPA_{2} (\mathcal{Pauli}(H))$ is also strict in general\footnote{The nonstrict inequality could be quickly understood in the same manner as the argument in the previous paragraph}. 
Specifically, we find that $NPA_2^\mathbb{R}(\mathcal{Pauli})$ fails for all of  the graphs shaded in \cref{fig:PauliProjCompare}, while it succeeds for all of the other graphs with $n=5$ and 6.
This means that the exact Pauli SoS for unshaded graphs are ``breaking the SU(2) symmetry" in the individual squares possibly by having one-body Pauli terms in them. Those effects must cancel out as a whole when all the SoS terms are added since the final Hamiltonian has SU(2) symmetry and has no one-body terms. 
For the shaded graphs, this ``symmetry breaking" trick is not enough to obtain the exact SoS, and complex SoS are required to do so. 
As a concrete example, the graph labeled 8 in \cref{fig:PauliProjCompare} has errors $1.4\cdot{} 10^{-2}$, $5.7\cdot{}10^{-4}$ and $8.06\cdot{}10^{-12}$ for $NPA_1(\mathcal{Proj}(H))$, $NPA_{2}^\mathbb{R}(\mathcal{Pauli}(H))$ and $NPA_{2} (\mathcal{Pauli}(H))$ respectively, which we interpret as the complex Pauli hierarchy being exact on this instance, but the real Pauli and complex projector hierarchy have nonzero errors.
Also, when we define the ``failing graphs'' based on the inexactness of $NPA_1(\mathcal{Proj}(H))$, instead of $NPA_2(\mathcal{Pauli}(H))$ as we did for Fig. \ref{fig:MetaGraphs}, the list of such failures is strictly larger because of this strict inequality. Namely, we need to add those listed in Fig. \ref{fig:PauliProjCompare} (a,b). In the language of Fig. \ref{fig:MetaGraphs}, there are more graphs that should be colored in magenta, which connects all of the ``islands'' of topologically disconnected failure graphs from the rest of the failing cluster for the $n=6$ case. 

The third cluster on the bottom left corresponds to graphs where the SDP algorithm succeeds with either of the bases. The ratio of the graphs in this third category seems to decrease as we get to larger sizes of graphs, which we will discuss further later. 
Noticing that the separation between $NPA_2(\mathcal{Pauli})$, $NPA_2^{\mathbb{R}}(\mathcal{Pauli})$, and $NPA_1(\mathcal{Proj})$ are strict in general from the previous paragraph, it seems more natural to regard this cluster as instances where $NPA_1(\mathcal{Proj}(H))=0$ forces the other two SDPs to have 0 error as well. 
From this perspective, it is more intriguing when $NPA_1(\mathcal{Proj}(H))=NPA_2(\mathcal{Pauli}(H))>0$, i.e., exactly on top of the $x=y$ line in \cref{fig:PauliProjCompare}, but in the top right cluster. 
Up to $n=8$ connected graphs we have computed, the only cases when that happens are all graphs related to complete graphs (simplest cases discussed in \cref{subsec:complete}).

There is a rather small fourth cluster on the right bottom, that extends beyond to the right side of the $x=y$ line. 
Since $NPA_2(\mathcal{Pauli})$ must always perform no worse than $NPA_1(\mathcal{Proj})$, this suggests a numerical error of some sort. We have observed that the SDP packages for these graphs do not converge as quickly as other graphs, and tends to give results that have larger duality gaps than specified. This practically does not become a problem since the errors are very small (around $10^{-6}$), and all graphs which we explicitly exemplify as ``NPA failing'' in this work are not from this group \footnote{This may occur strange to the physicist readers that a convex optimization which theoretically does not have a local minimum, still seems to ``get stuck'' in practice. This is actually not uncommon in the field of convex optimization, since e.g. a very narrow feasible region can cause practically slow convergences like this.}. 


Notably, instances falling on the right side of the $x=y$ line only occur at very small errors (bottom right), while none are observed in the top right cluster. This is 
encouraging, since we can be confident that these practically pathological cases only arise when we demand high numerical precision. This allows us to consider all of the graphs in the fourth cluster (bottom right) to be theoretically easy for both bases of SDP, i.e., actually belonging to the third cluster (bottom left).


\begin{figure*}[t]
    \centering
    \includegraphics[width=17cm]{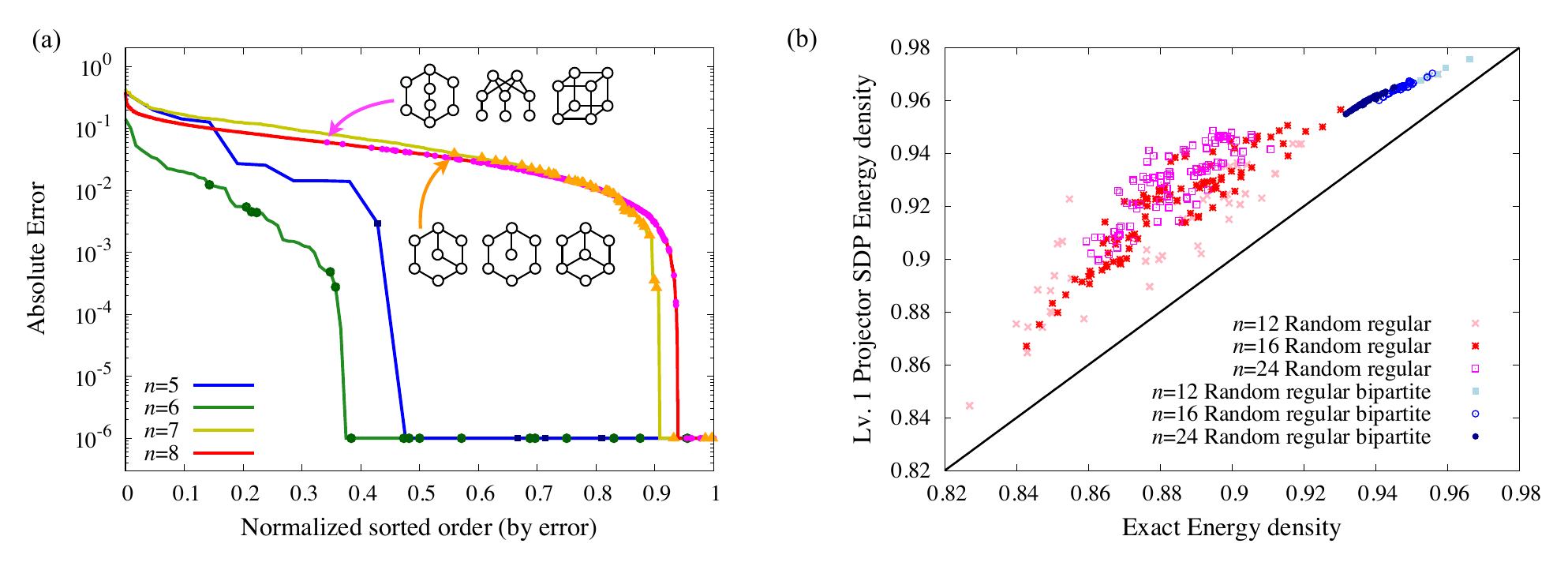}
    \caption{ (a) Absolute error values $E_{\mathrm{SDP}}-E_{\mathrm{GS}}$ for all connected graphs with $n=5,6,7$ and $8$ vertices, shown in descending order for each $n$. All error values smaller than $10^{-6}$ can be regarded as 0, and is displayed as $10^{-6}$ for visual simplicity. We show points corresponding to bipartite graphs with points on top of the curves, showing the general tendency of bipartite graphs having relatively smaller errors. We also illustrate the top three bipartite $n=8, 7$ graphs with largest errors. (b) The scatter plot of random regular graphs and random regular bipartite graphs with degree 3 and size $n=12, 16$ and 24 qubits, 100 samples each.}
    \label{fig:errorstats}
\end{figure*}

In order to see the statistics of the errors more closely, in \cref{fig:errorstats} (a), we show the values of the error for the $NPA_1(\mathcal{Proj})$ SDP in descending order for each size of graphs $n=5,6,7$ and $8$. 
The $x$-axis is rescaled so that the data of 21, 112, 853, and 11,117 graphs all fit into $[0,1]$. Thus, the figure is the inverse of the cumulative distribution function of errors. 

For example, all four curves display an acute decline at some point corresponding to the separation between graphs that have nonzero errors and (essentially) zero error. The graph shows that the ratio of such non-exactly solvable graphs are roughly $46\%, 37\%, 91\%$ and $94\%$ among all connected $n=5,6,7$ and $8$-vertex graphs respectively. This means that the ratio of exactly solvable graphs tend to decrease as the number of vertices increases, possibly converging to 0 in the $n\rightarrow \infty$ limit. 
Yet still, the actual {\it number} of connected graphs that are exactly solvable seems to grow with $n$ at least for this size regime: 11, 67, 77, and 670, for $n=5,6,7$ and 8. 

Another piece of information in the graph, represented as the points in the figure, is how the {\it bipartite} graphs are distributed among this descending-error ordering. The {\scshape QMaxCut} problem on bipartite graphs is oftentimes described as having ``no geometric frustration" in condensed matter physics, since the singlet projector $h_{ij}$ could be seen as a constraint that favors the two qubits to be pointing in the opposite direction. \footnote{Not to be confused with ``frustration-free'' explained in section \ref{subsubsec:MG}.} From this point of view, we would consider an odd-length loop as geometrically frustrated because the interaction would not be (even relatively) satisfied with a simple approach of having the qubits point the opposite directions alternately. 
This difference has practical applications, such as bipartite cases allowing the quantum Monte Carlo method to efficiently\footnote{Only known empirically, in terms of precise complexity theory statements. While the time complexity scaling is known to scale as $\mathcal{O}(\epsilon^{-2})$ with respect to the error tolerance $\epsilon$, the scaling with number of qubits $n$ is hard to bound rigorously for Markov-chain Monte Carlo methods in general, albeit cases of quantum Monte Carlo methods being applied to hundreds or thousands of qubits is common in computational physics \cite{san10com}.} obtain the ground state classically. 
Therefore, it is not so surprising that the bipartite graphs in Fig. \ref{fig:errorstats} (a) are distributed relatively on the right side of each curves, implying (exponentially) smaller errors. In some sense, the surprise is in the other direction, that SDP fails to obtain the exact ground states of such ``easily classically simulable" instances most of the time. 
It is unclear if the tendency of bipartite graphs having relatively small errors will remain for larger $n$, since it is already apparent that the position of the largest-error bipartite graph shifts to the left in Fig. \ref{fig:errorstats} (a) from $n=7$ to $n=8$.


In order to test the difference between bipartite graphs and non-bipartite graphs in a more systematic way, we also ran the SDP algorithm for random regular graphs with degree-3. When such graphs are generated uniformly randomly, for sufficiently large $n$, the graph is almost certainly non-bipartite. 
We generate 100 of such samples, and compare the performance of $NPA_1(\mathcal{Proj})$ 
against exact diagonalization for $n=12, 16$ and 24. 
It is also possible to generate uniformly random graphs that are bipartite and regular, and both results are displayed in Fig. \ref{fig:errorstats} (b). 
It is immediately apparent that the non-bipartite random regular graphs have a broader distribution in the two-dimensional scatter plot, compared to the bipartite cases. The cluster is also located farther away from the $x=y$ line in black, showing a larger relative error compared to bipartite random graphs. The bipartite random graph data also seem 
 to form a ``line" in the scatter plot, 
indicating that the optimal SDP objective can give a fairly narrow estimate of the true energy value by a properly fitted linear function. In contrast, the non-bipartite random graph data extends in a two-dimensional manner forming a oval-like shape, resulting in broader estimates of the true energy given the SDP energy.

\begin{figure*}[t]
    \centering
    \includegraphics[width=17cm]{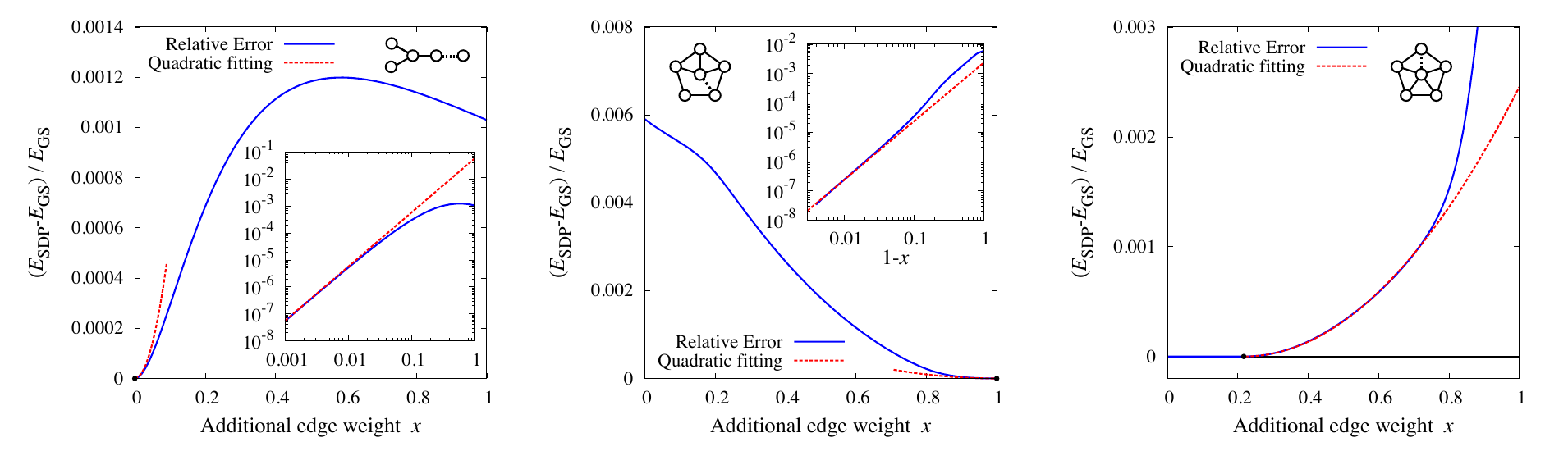}
    \caption{ Three different cases of adding an additional edge to a graph with weight $x$, resulting in the change of solvability with $NPA_1(\mathcal{Proj})$. The insets show the log-log plot to clarify the matching of the quadratic fitting near the transition.}
    \label{fig:SmallTrans}
\end{figure*}

\subsubsection{Transition points in the solvability of small graphs}

The clusteredness of hard and easy graphs shown in Fig. \ref{fig:MetaGraphs} leads to the question of what happens at the boundary between them. 
If there is a pair of graphs which one is exactly solvable while the other is not, with only one edge difference as graphs, then we can add that one different edge with weight $x\in[0,1]$. 
This procedure continuously connects the graphs and demonstrates where exactly SDP starts to fail. 

In \cref{fig:SmallTrans}, we show three different cases of such a procedure. 
On each panel, we show the graph we use for demonstration, with the dotted edge being the weighted one. 
The left most panel shows the case for interpolating between the $n=4$ star graph and the Y-shaped $n=5$ graph (graph \# 20 in \cref{fig:PauliProjCompare} (a)), which is the easiest case of such. In this case, we can see that the moment we add $\epsilon>0$ amount of the new edge, SDP starts to fail. 
This could be argued that the solvability of the star graph in this situation is rather {\it fragile}, and immediately fails when perturbed away. 

The same thing could be argued for the case shown in the middle panel connecting graph \#69 and \#47 of \cref{fig:MetaGraphs} right. 
Again in this case, the moment the graph diverges away from the exactly solvable \#47, the SDP algorithm starts to fail. 
However, there exist cases where the ``transition" happens not at the edges but at a nontrivial value, as shown in the right panel. The error becomes as small as the duality gap set for the SDP solver for $x<0.22$. In this case, we can say that the solvability of graph \#47 is somewhat robust, and survives the perturbation in the direction considered here (towards graph \#28). 

Curiously, for all cases we have checked for interpolations between solvable and unsolvable graphs with $NPA_1(\mathcal{Proj})$, we always observe a quadratic initial increase of the error, as shown with the red dotted lines in \cref{fig:SmallTrans}. The quadratic fit is extremely good at the vicinity of the ``transition points" where the error starts to become nonzero, as shown in the insets of the figures. 
This resembles universal critical behavior seen in physics, where phase transition {\it points} vary largely depending on the details of the statistical physics model, but an indicator of the phase transition (called the order parameter) behaves as $\propto |T-T_{c}|^{\beta}$ with a universal exponent denoted by $\beta$. 
Although our observed exponent $\beta=2$ is clearly present numerically, we were unable to provide a general explanation, and leave it for future studies.

\subsection{Numerical results for some condensed matter physics models}\label{subsec:cmp}

Here, we demonstrate the power of the SDP algorithm when applied to a number of condensed matter physics models. The message is two-fold: first, the SDP algorithm could be used to probe exact-solvability of models in some settings, giving rise to the possibility of numerical exploration for exactly-(analytically) solvable systems. Second, the method could be seen as the first-order approximation of the ground state, it actually gives very accurate numbers in practice, with errors only up to $\sim 4\%, 7\%, 2\%$ for the models we study. 

Both the Majumdar-Ghosh model and the Shastry-Sutherland model are known to be ``frustration-free" in the quantum spin system literature \cite{deb10sol,sat16whe,wou21int,ans22ana}. 
This means that the Hamiltonian could be rewritten as sum of terms that could all be satisfied simultaneously in the ground state. 
The standard way to show this is to rewrite the physical Hamiltonian (thus with the opposite sign from our {\scshape QMaxCut} convention as in Eq. (\ref{eq:QMCHamDef})) as a sum of projectors with additive and multiplicative constants. 
If there exists a state that the all the projectors evaluate to 0, that must be the ground state (note that the physics convention here is a minimization of the eigenvalue).

\begin{figure*}[t]
    \centering
    \includegraphics[width=15cm]{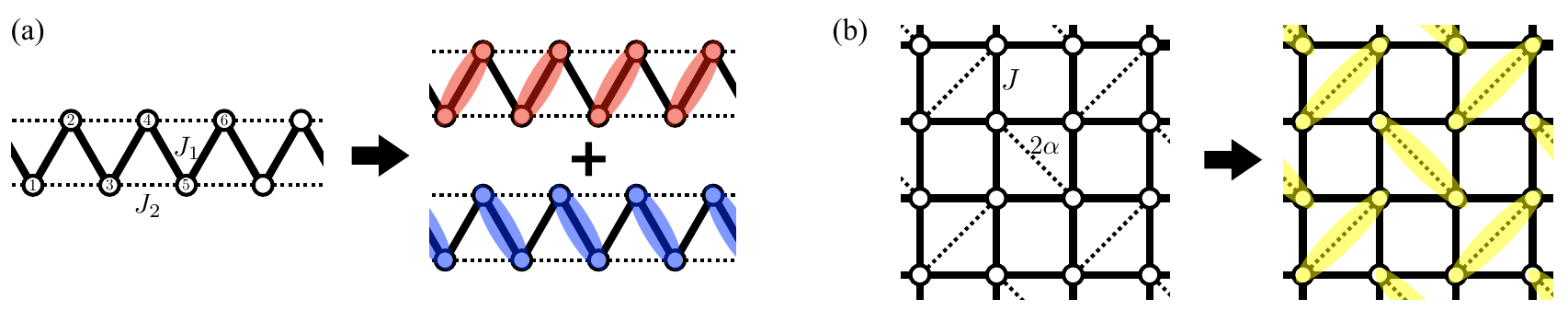}
    \caption{ The graph structures of the (a) $J_1$-$J_2$ model and (b) the Shastry-Sutherland model. The exactly solvable ground states for both models are also illustrated where the colored ovals represent singlet pairs. }
    \label{fig:lattices}
\end{figure*}

Because all projectors are square of themselves, we can immediately obtain an SoS of some degree by flipping the entire sign, and redefining the Hamiltonian with the {\scshape QMaxCut} convention. 
In most cases in physics, the definition of ``frustration-free" requires the rewritten terms of the Hamiltonian to be spatially local. Thus, the SoS hierarchy could be seen as a generalization of the frustration-free notion, where we do not necessarily require spatial locality, but restrict the degree of the terms as polynomials.
The fact that the degree-restriction could be arbitrarily relaxed by the level of the hierarchy, and that SDP algorithms can solve the optimization problem efficiently for $\mathcal{O}(1)$ level, provides us a systematic approach to explore frustration-free Hamiltonians in a computational way. 

\begin{figure*}[t]
    \hspace{-1.5cm}
    \begin{minipage}[b]{0.32\linewidth} 
    \includegraphics[width=7cm]{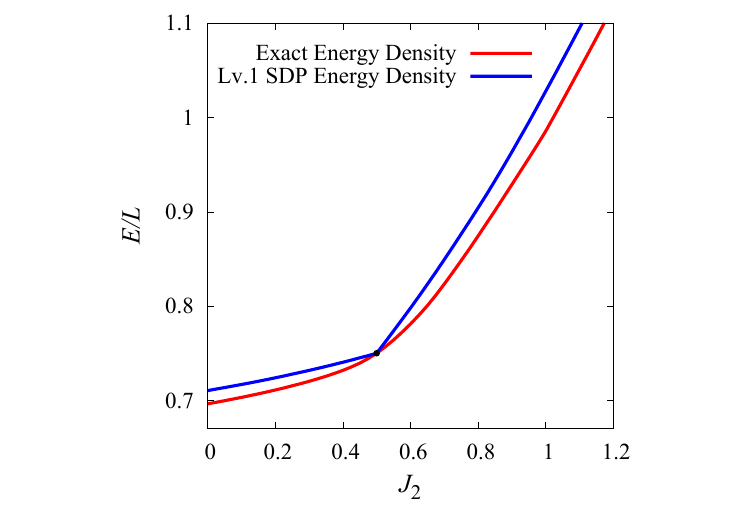}
    \end{minipage}
    \begin{minipage}[b]{0.32\linewidth} 
    \includegraphics[width=7cm]{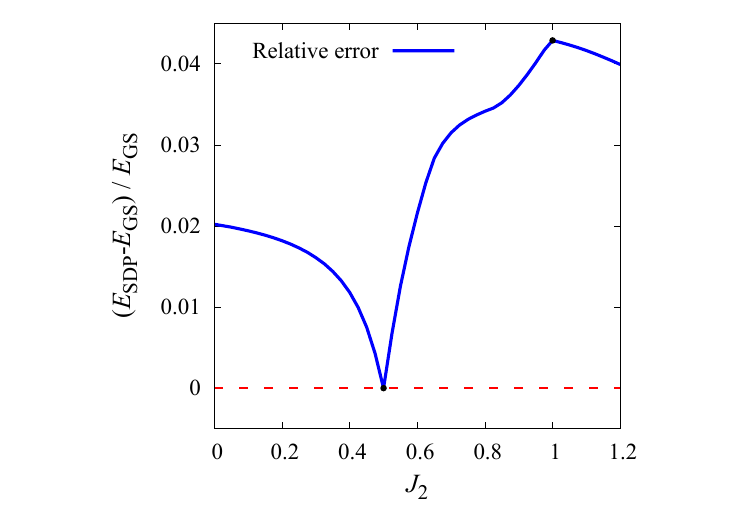}
    \end{minipage}
    \begin{minipage}[b]{0.32\linewidth} 
    \includegraphics[width=7cm]{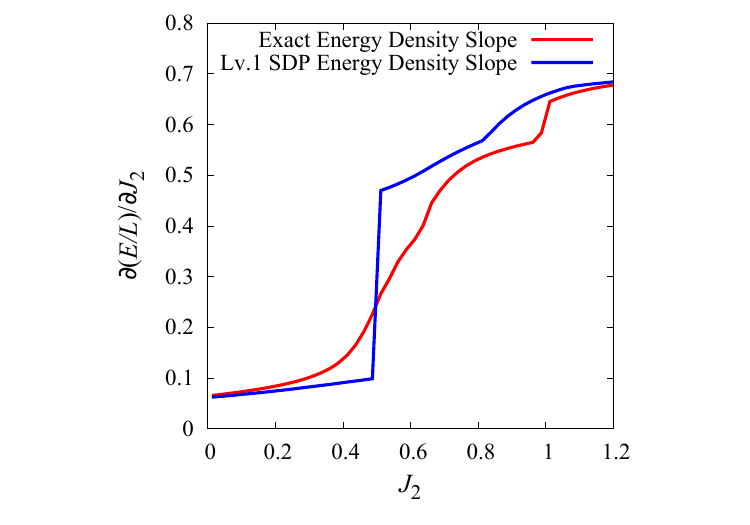}
    \end{minipage}
\caption{\label{fig:MG}Performance of the SDP algorithm for the the size $L=16$ $J_1$-$J_2$ model with periodic boundary condition and $J_1=1$ fixed. Comparison is made between the exact ground state energy $E_{\mathrm{GS}}$ and the values obtained by the Lv.1 singlet projector SDP algorithm $E_{\mathrm{SDP}}$. The energy densities (left) are very close, and we can see that the relative error (center) becomes exactly 0 at the MG point. When we take the derivative of the energy densities (right), the SDP algorithm actually detects the MG point at $J_2/J_1 =1/2$ in the model way more sharply compared to the exact solution, at this system size.}
\end{figure*}

\subsubsection{The Majumdar-Ghosh model}\label{subsubsec:MG}

The Majumdar-Ghosh (MG) model \cite{maj69nex} was one of the earliest proposed quantum spin models where the ground state could be obtained exactly. 
The Hamiltonian being considered is 
\begin{equation}\label{eq:MGHam}
    H = \sum_{i=1}^L \Bigl(J_1 h_{i,i+1} + J_2  h_{i,i+2}\Bigr), 
\end{equation}
where we use the {\scshape QMaxCut} convention (i.e. the ``ground state" we are searching for now is the maximum eigen state of this operator). The lattice structure is shown in Fig. \ref{fig:lattices}(a). 
The Hamiltonian above would typically be referred to as the $J_1$-$J_2$ Heisenberg chain in the context of condensed matter physics, and the MG model corresponds to the case where $J_2/J_1=1/2$, also known as the MG {\it point}. 
At the MG point, the ground state is two-fold degenerate with two different singlet-product states where closest neighbors are forming singlets periodically, as depicted in Fig. \ref{fig:lattices}(a): 
\begin{equation} \label{eq:MGGS}
    |\mathrm{GS}\rangle = 
    \prod_{i\in \mathrm{even}}^{\otimes}|s_{i,i+1}\rangle
    +\prod_{i\in \mathrm{odd}}^{\otimes}|s_{i,i+1}\rangle ,
\end{equation}
where by $|s_{i,j}\rangle$ we denote a singlet state between the spins on sites $i$ and $j$.

The exact ground state could be understood from the fact that the Hamiltonian at the MG point decomposes in the same sense of Eq. (\ref{eq:CBhamdecomp}). 
Specifically,
\begin{eqnarray}
    H=\frac{J_1}{2}\sum_{i=1}^L \Bigl(h_{i,i+1} + h_{i+1,i+2} + h_{i,i+2}\Bigr), \\
    \lVert H\rVert=\frac{J_1}{2}\sum_{i=1}^L \Big\lVert h_{i,i+1} + h_{i+1,i+2} + h_{i,i+2}\Big\rVert 
\end{eqnarray}
holds. The matrix norm $\lVert\cdot\rVert$ plays the same role as $\mu_{\max}$ here.
Since the individual terms after this decomposition reduces to a triangle with equal weights, we can reuse the exact $\SoS$ from Eq. (\ref{eq:WeakTriangleSoS}), obtaining 
\begin{equation}
    \frac{3 J_1 L}{4}\mathbb{I} - H 
    = 
    \sum_{k=1}^{L}
    \frac{3J_1}{4}
    \biggl\{
        \mathbb{I}-\frac{2}{3}
        \bigl(
            h_{k-1,k} + h_{k,k+1} + h_{k-1,k+1}
        \bigr)
        \vspace{-2mm}
    \biggr\}^2
\end{equation}
for the periodic boundary condition ($L+i\equiv i$). 
This corresponds to the standard projector expression for frustration-free models as we mentioned earlier, and implies that $NPA_1(\mathcal{Proj})$ is able to obtain the exact ground state energy at the MG point. 
As was the case for proving exactness of graphs considered in \cref{subsec:compbip}, we can see that this SoS is indeed exact (tight) from the fact that the ground state \cref{eq:MGGS} achieves energy $3L/4$. The latter calculation can be carried out rather easily by noticing that the singlet configuration fully satisfies the bonds ($h_{ij}$) they are on ($J_1$) while bonds without the singlets ($J_2$) still gain energy $1/4$ of their weights coming from the identity in \cref{eq:QMCHamDef}.

This is demonstrated in Fig. \ref{fig:MG}, where we compare the exact energy $E_{\mathrm{GS}}$ and the SDP energy $E_{\mathrm{SDP}}$ for the $L=16$ case with periodic boundary condition with various values of $J_2$ fixing $J_1=1$. 
The fact that the SDP algorithm obtains the exact ground state energy is reflected as the two energy density values coinciding at the MG point in the left panel, and correspondingly in the center panel the relative error becomes 0. 
Interestingly, the SDP algorithm seems to be ``more sensitive" to the MG point in this $J_1$-$J_2$ model compared to the actual ground state energy, when we try to detect it by looking into the derivatives of the energy (right panel).

The fact that in Fig. \ref{fig:MG} we see the SDP algorithm only obtaining the energy exactly at the MG point establishes that there are no other exactly-solvable points in the $J_1$-$J_2$ model, even if we allow non-local terms as long as they are limited to degree-2 in polynomials of singlet projectors.

\subsubsection{The Shastry-Sutherland model}\label{subsubsec:SS}

The Shastry-Sutherland (SS) model \cite{sha81exa} is a two-dimensional Heisenberg model that also admits an exact ground state representation for a certain parameter region. 
The Hamiltonian in the {\scshape QMaxCut} convention would read 
\begin{equation}
    H=J\sum_{\langle ij\rangle} h_{ij} + 2\alpha \sum_{\langle\hspace{-0.5mm}\langle ij\rangle\hspace{-0.5mm}\rangle} h_{ij}
\end{equation}
where $\langle ij\rangle$ represents bonds of the $L\times L$ square lattice (with periodic boundary condition) and $\langle\hspace{-1mm}\langle ij\rangle\hspace{-1mm}\rangle$ represents diagonal bonds of the Shastry-Sutherland lattice as illustrated in Fig. \ref{fig:lattices}(b). For simplicity, we fix $J=1$. 

This model has an obvious unique ground state when $\alpha$ is large enough, since the diagonal bonds with weight $2\alpha$ gives a perfect matching of the sites. In that parameter region, the unique ground state could be written as 
\begin{equation}\label{eq:SSGS}
    |\mathrm{GS}\rangle = 
    \prod_{\langle\hspace{-0.5mm}\langle ij\rangle\hspace{-0.5mm}\rangle}^{\otimes}|s_{i,j}\rangle , 
\end{equation}
again illustrated in Fig. \ref{fig:lattices}(b). 

For the SS model with $\alpha>1$, we are again able to decompose the Hamiltonian into triangles with weights 1, 1, and $\alpha$, as previously discussed in section \ref{subsubsec:crown}. 
It is easy to check that the Shastry Sutherland lattice (Fig. \ref{fig:lattices} (b)) can be decomposed into such triangles geometrically, with all triangles having two edges from the square lattice and one from the diagonal edges. 
Now, we can reuse Eq. (\ref{eq:StrongTriangleSoS}) to obtain 
\begin{eqnarray}\label{eq:SS-SOS}
    &&\left(\alpha+\frac{J}{2}\right)N\mathbb{I} - H =
    \sum_{\triangle}
    \left(\alpha+\frac{J}{2}\right)\times\\
    &&\hspace{-3mm}\biggl\{
    \mathbb{I}-
    \sum_{\mathrm{edges}\in\triangle} \hspace{-3mm}   \frac{4\alpha+2J\pm(-2)^{j}\sqrt{2\alpha(2\alpha-J)-2J^2}}{3J+6\alpha}h_{\mathrm{edge}}
    \biggr\}^2 , \nonumber
\end{eqnarray}
which gives the exact value only when $\alpha\geq J$. 
The summation $\sum_{\triangle}$ is taking the summation for all right triangles in the SS lattice as in the decomposition, and the summation inside of the square is for the three different edges for each such triangle. The $(-2)^j$ factor only appears for the edges with weight $J$ belonging to the square lattice where we set $j=1$, and not for the diagonal edges with weight $\alpha$ which we set $j=0$. Just as in Eq. (\ref{eq:StrongTriangleSoS}), the SoS has a degree of freedom in choosing $\pm$ for the square root term. 
As was the case for the Maujumdar Ghosh model, here again the SoS \cref{eq:SS-SOS} is upper bounding the ground state energy with $(\alpha+J/2)N$ where $N=L^2$ is the number of qubits in the lattice. Confirming that \cref{eq:SSGS} indeed achieves the energy $(\alpha+J/2)N$ completes the proof that the SoS is exact.

In Fig. \ref{fig:SS}, we demonstrate the performance of the $NPA_1(\mathcal{Proj})$ SDP algorithm applied to the SS model with system size $n=L^2=16$ and $J=1$ fixed. We can see that the algorithm obtains the exact ground state energy for the entirety of the $\alpha\geq 1$ region, which exactly coincides where Eq.~(\ref{eq:SS-SOS}) gives a proper SoS (otherwise it has no real coefficients), and also the decomposition exists. 
The true ground state actually becomes the dimer singlet state \cref{eq:SSGS} from $\alpha\geq3/4$ for this system size, although the SDP algorithm fails to obtain that. 
This means that while $\alpha\geq 1$ was the condition used to show frustration-freeness in \cite{sha81exa}, relaxing the notion to allow non-local terms (but still only having degree-1 terms in the SoS) does not enlarge the region of exact-solvability. It would be interesting to see how the exactly solvable region changes as a function of the level of the NPA hierarchy.

\begin{figure*}[t]
    \hspace{-1.75cm}
    \begin{minipage}[b]{0.32\linewidth} 
    \includegraphics[width=7cm]{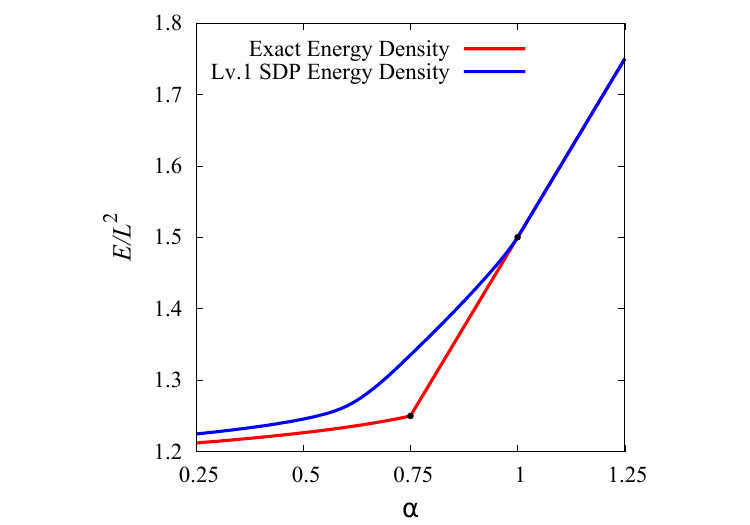}
    \end{minipage}
    \begin{minipage}[b]{0.32\linewidth} 
    \includegraphics[width=7cm]{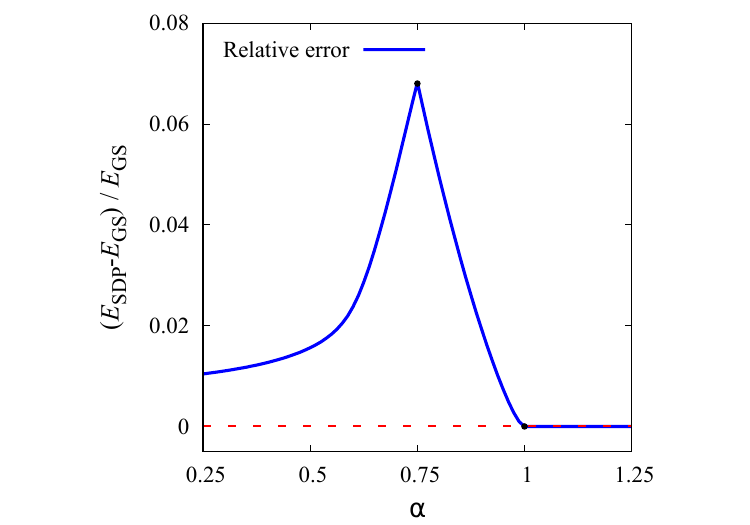}
    \end{minipage}
    \begin{minipage}[b]{0.3\linewidth} 
    \includegraphics[width=7cm]{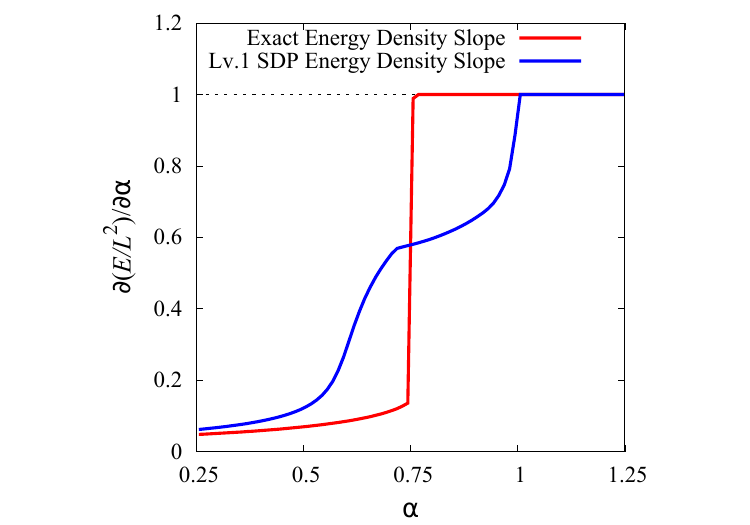}
    \end{minipage}
    \caption{Performance of the SDP algorithm for system size $n=L^2=16$ Shastry-Sutherland model with periodic boundary condition. Comparison is made between the exact ground state energy $E_{\mathrm{GS}}$ and the values obtained by the Lv.1 singlet projector SDP algorithm $E_{\mathrm{SDP}}$. The energy densities (left) are very close, and we can see that the relative error (center) becomes exactly 0 for $\alpha\geq 1$. When we take the derivative of the energy densities (right), the SDP algorithm seems to be detecting the two phase transitions in the model more sharply compared to the exact solution, at this system size.}
    \label{fig:SS}
\end{figure*}

Furthermore, by looking at the first derivative of the SDP energy as a function of $\alpha$, we can clearly see that there are two points where $\partial E/\partial \alpha$ has a singularity (\cref{fig:SS} right), namely $\alpha \simeq 0.73$ and $\alpha =1$. 
The existence and the nature of different phases in the SS model is actively discussed in the condensed-matter physics context, where there is expected to be at least two phase transition points, i.e., singular points \cite{kog00qua,lee19sig,yan22qua}. 
The fact that the SDP energy derivative exhibit two singular points from relatively small system sizes suggest the possibility of this approach be used to detect phase transitions in similar models, without relying on comparison with exactly obtained ground states. The SDP algorithm also allows us to calculate observables other than energy such as the squared N{\'e}el order parameter from the moments, with a guarantee that they converge to the true value when the NPA hierarchy converges to the true ground state energy value. 
This lets us to interpret such physical observables obtained this way to be regarded as a first-order approximation. In the next section we see that even such approximated quantities can show essential characteristics in physical systems. 

Another thing to note is that both Hamiltonians for the SS model and the MG model allowed decomposition of the Hamiltonian as in \cref{eq:CBhamdecomp} and \cref{eq:stardecomp}. In the cases of SS and MG models, the sub-Hamiltonians were the triangles with $\alpha,J,J$ bonds and $J_1/2, J_1/2, J_2=J_1/2$ bonds respectively. 
However, it should be noted that the existence of such decomposition is not a necessary condition for obtaining exact SoSs. For example, for the crown graph which we present an exact SoS in section \ref{subsubsec:crown}, it appears that there are no such decomposition, while still having an exact SoS. 
The same could be said for complete graphs with even number of vertices. 
This fact gives us hope on discovering new exactly-solvable Hamiltonians, since oftentimes the search for frustration-free Hamiltonians relies on the existence of such decomposition \cite{gho23exa,kum02qua}. 

\subsubsection{The Heisenberg chain}\label{subsec:HeisenbergChain}
The nearest neighbor antiferromagnetic Heisenberg chain is one of the simplest yet nontrivial quantum spin system that also happens to be a {\scshape QMaxCut} instance. 
The Hamiltonian we consider here is simply the chain 
\begin{equation}
H=\sum_{i=1}^L h_{i,i+1}, 
\end{equation}
with a periodic boundary condition $L+k \equiv k$. 
This corresponds to setting $J_2=0$ for the $J_1$-$J_2$ model in section \ref{subsubsec:MG}. 
Although the Heisenberg chain has an exact solution thanks to the Bethe ansatz \cite{bet31zur}, the exact solvability of the model is quite different from the previous two models: it does not involve frustration-freeness, and our SDP algorithm is therefore not expected to solve it exactly. 

\begin{figure}[t]
    \subfigure{\includegraphics[width=8.3cm]{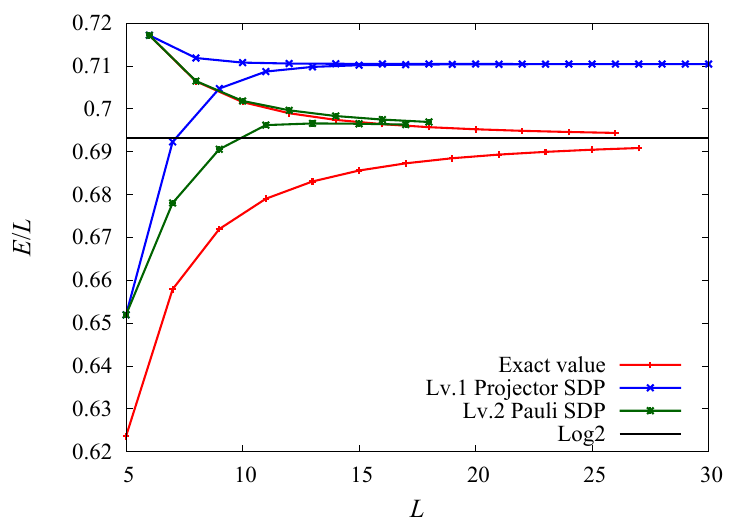}}
    \subfigure{\includegraphics[width=8.3cm]{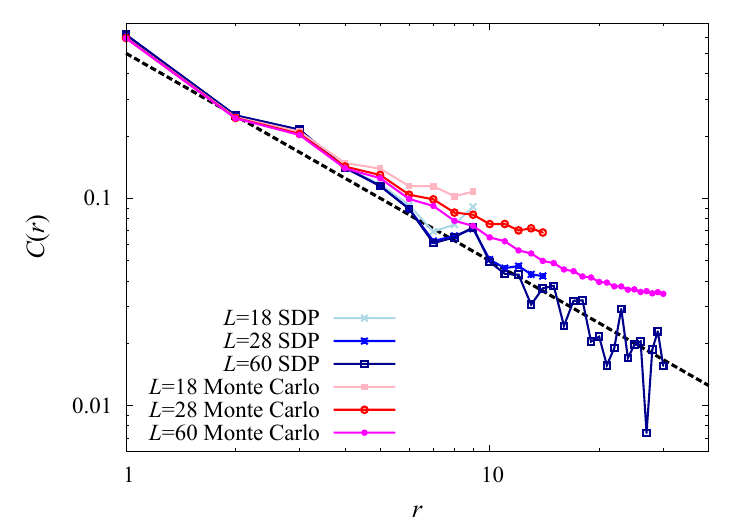}}
\caption{\label{fig:Cycle}Performance of the SDP algorithm on the Heisenberg chain with periodic boundary conditions(cycle graphs). The left panel shows the exact and SDP obtained energy densities for system sizes $L=5$-$30$. In the thermodyanmic limit $L\rightarrow\infty$, the true energy density converges to $\log 2$. The right panel shows the correlation function $C(r)$ as a function of distance, both obtained by Lv. 1 projector SDP and quantum Monte Carlo (essentially exact). The dotted line shows the theoretically known asymptotic decay exponent $C(r)\propto r^{-1}$.}
\end{figure}

In the left panel of \cref{fig:Cycle}, we show our numerical results on how the SDP algorithm performs on the Heisenberg chain, by comparing the $NPA_1(\mathcal{Proj})$, $NPA_2(\mathcal{Pauli})$, and the exact value for various system sizes. 
We plot the energy {\it density} $E/L$ here, so the fact that all three cases converge to different values indicate that the absolute error of the total energy increases linearly with the system size $L$ for large enough $L$. Still, the qualitative behavior of approaching the limiting value from above and below for even and odd $L$ is reproduced in both of the SDP methods. 

In the right panel, we show the correlation function 
\begin{equation}\label{eq:corr}
    C(r) := \langle \mathrm{GS} | (-1)^r Z_i Z_{i+r} |\mathrm{GS}\rangle = (-1)^r\frac{1-4 ~M(\mathbb{I}, h_{i,i+r})}{3}, 
\end{equation}
obtained by $NPA_1(\mathcal{Proj})$ and Monte Carlo (virtually exact value) for system sizes $L=18, 28,$ and $60$. 
Note that the translation symmetry of the cycle ensures the well-definedness of $C(r)$ regarding the choice of $i$ in the definition. 
The second equality in \cref{eq:corr} is valid only in the case where the SDP algorithm obtains the exact ground state. However, we can still measure the RHS quantity even in cases where the algorithm fails and consider the obtained result as an approximation.
Remarkably, when utilizing the SDP-obtained correlation function in this manner, it aligns very well with the true correlation function for small values of $r$ as shown in the figure. This can be attributed to the fact that the energy density exhibits a relative error of only $\sim2\%$.

One characteristic feature of the Heisenberg chain is that the ground state displays a power-law decaying correlation with critical exponent $\eta =1$, i.e., $C(r)\propto r^{-1}$, which is closely linked to the long-range entanglement it has. 
The fact that the SDP-obtained correlation function displays the same type of power-law decay with essentially the correct exponent (albeit the jagged feature) is quite interesting especially when it is compared to the exact correlation function of finite systems, since it appears to have even smaller finite-size corrections. 
This raises the intriguing possibility that SDP-derived quantities capture the underlying ``physics" of the ground state, even when there is no physical quantum state corresponding to the optimal moment matrix. 

Finally, an alert reader may notice from the figure that both $NPA_1(\mathcal{Proj})$ and $NPA_2(\mathcal{Pauli})$ are exact for the $L=6$ hexagon case. 
Although we were unable to obtain an analytic $\SoS$, we were able to study the structure of the ground state from the SDP perspective, which we provide in Appendix \ref{app:hexagon}.

\begin{acknowledgments}
We thank Claudio Procesi for bringing to our attention several important references.  J.T. thanks Hosho Katsura, Cristopher Moore, Sho Sugiura, and Seiji Takahashi for valuable discussions. C.Z. thanks Micha{\l} Adamaszek from MOSEK for helpful discussions regarding improving the efficiency of solving SDPs using MOSEK. 
K.T. and O.P. acknowledge discussions with Anirban Chowdhury.  J.T., C.R., and C.Z. thank Elizabeth Crosson for initially igniting this fruitful project. 
J.T. and C.Z. acknowledge support from the U.S. National Science Foundation under Grant No. 2116246, the U.S. Department of Energy, Office of Science, National Quantum Information Science Research Centers, and Quantum Systems Accelerator. 
O.P. and K.T. are supported by Sandia National Laboratories. Sandia National Laboratories is a multimission laboratory managed and operated by National Technology and Engineering Solutions of Sandia, LLC., a wholly owned subsidiary of Honeywell International, Inc., for the U.S. Department of Energy’s National Nuclear Security Administration under contract DE-NA-0003525. This work was supported by the U.S. Department of Energy, Office of Science, Office of Advanced Scientific Computing Research, Accelerated Research in Quantum Computing.
\end{acknowledgments}

\section*{Data Availability}
The data presented in this article are openly available at \cite{data}.

\appendix
\section{Proof that $\mathcal{Perm}(V, w)$ computes $\SWAP(V, w)$}

The proof is relatively simple given the well-known form of the irreducible representations of the Symmetric group.  For completeness we review the necessary background before proving the theorem at the end of the section.  

\subsection{Representation Theory}\label{sec:rep_theory}
Given a finite group $G$ a representation is a pair $(\rho, V)$ where $V$ is a vector space and $\rho: G\rightarrow GL(V)$ which is a homomorphism of groups.  As such, $\rho(g) \in GL(V)$ satisfies $\rho(g_1 g_2)=\rho(g_1) \rho(g_2)$ and $\rho(g^{-1})=\rho(g)^{-1}$.  We say two representations $(\rho_1, V)$ and $(\rho_2, W)$ are isomorphic if there is an isomorphism of vector spaces $T: V\rightarrow W$ which satisfies $T\rho_1(g) T^{-1} =\rho_2(g) $ for all $g\in G$.  If $(\rho, V)$ is a representation and $V'\subseteq V$ is a subspace satisfying $\rho(g) \ket{v'} \in V'$ for all $\ket{v'} \in V'$ and $g\in G$ then $V'$ is called a $G$-invariant subspace.  A representation $(\rho, V)$ is called irreducible if the only $G$-invariant subspaces are $V$ and $\{0\}$.  We will refer to an irreducible representation simply as an irrep. 
 We will often drop the function $\rho$ from a representation and talk about elements of the group as acting on vectors from $V$, i.e. $g \ket{v}:= \rho(g) \ket{v} $.  

The group algebra, denoted $\mathbb{C}[G]$, is a particular representation which fully captures the representation theory of a given group.  This is the vector space of formal complex linear combinations of the group elements $\mathbb{C}[G]=\{\sum_{g\in G} c_g g :\,\, c_g \in \mathbb{C} \,\, \forall g\}$.  $G$ acts on this space according to the multiplcation rule of the group: $g\sum_{g'\in G} c_{g'} g'=\sum_{g'\in G} c_{g'} g g' $.  The vector space is also an algebra with multiplication extended by linearity to the formal sums: $\left( \sum_{g\in G} b_g g \right) \left(\sum_{g'\in G} c_{g'} g' \right)=\sum_{g, g'\in G} b_g c_{g'} g g'$.

It is natural in the setting of representations to consider ``algebraic constraints''.  If $p=\sum_g c_g g \in \mathbb{C}[S_n]$, we say $(\rho, V)$ satisfies constraint $p$ if
$$
\sum_{g \in G} c_g \rho(g)=0.
$$
It is clear that if $(\rho_1, V)$ and $(\rho_2, W)$ are isomorphic representations then $(\rho_1, V)$ satisfies constraint $p$ if and only if $(\rho_2, W)$ satisfies constraint $p$:
\begin{align*}
\sum_g c_g \rho_1(g)=0 \Leftrightarrow T\left( \sum_g c_g \rho_1(g)\right) T^{-1}=0 \\
\Leftrightarrow \sum_g c_g \rho_2(g)=0.
\end{align*}
So, it well-defined to say a particular representation satisfies an algebraic constraint without reference to an explicit basis (any isomorphic representation also satisfies the constraint).  Similarly it is well-defined to talk about eigenvalues of a representation in abstraction since the characteristic polynomial is invariant under isomorphism:
\begin{align*}
&\det\bigg(\lambda\mathbb{I}-\sum_g c_g \rho_1 (g)\bigg)=0 \\
\Leftrightarrow& \det(T) \det(T^{-1})\det\bigg(\lambda\mathbb{I}-\sum_g c_g \rho_1 (g)\bigg)=0 \\
\Leftrightarrow& \det\bigg(T \left( \lambda\mathbb{I}-\sum_g c_g \rho_1 (g)\right)T^{-1}\bigg)=0 \\
\Leftrightarrow& \det\bigg(\lambda\mathbb{I}-\sum_g c_g \rho_2 (g)\bigg)=0
\end{align*}

The central theorem of representation theory for finite groups is that an arbitrary representation $(\rho', V')$ is isomorphic to one which decomposes into irreps.  For every group $G$ there is some finite list of irreps $\{V_1, V_2, ... V_q\}$ such that an arbitrary finite dimensional representation $(\rho', V')$ is isomorphic to a representation $(\rho, V)$ where for a set of non-negative integers $\{m_i\}_{i=1}^q$ $V$, decomposes as $V=\bigoplus_{i=1}^q m_i V_i$ where $m_iV_i= \bigoplus_{j=0}^{m_i} V_i$ ($m_i$ copies of $V_i$).  Further it holds that $\rho(g)$ is block diagonal in this decomposition for all $g\in G$.  In this decomposition the block of $\rho(g)$ corresponding to a particular $V_i$ is simply $\rho_i(g)$ where $\rho_i$ is the representation corresponding to $V_i$.  

Now we may reinterpret the previous observations about the spectrum and constraints in the context of this decomposition.  Since $\rho(g)=\bigoplus_i m_i \rho_i(g)$ is block diagonal, a particular constraint is satisfied by $(\rho, V) \cong \bigg(\bigoplus_i m_i V_i, \bigoplus_i m_i \rho_i(g) \bigg)$ if and only if it is satisfied by all the irreps $V_i$ in the decomposition with $m_i \geq 1$.  We say an irrep $V_i$ is involved in $(\rho, V)$ if $m_i\geq 1 $.  Further we can find the smallest/largest eigenvalue of some $p \in \mathbb{C}[G]$ in $(\rho, V)$ by finding the smallest/largest eigenvalue of each irrep and taking the minimum/maximum of the set of eigenvalues. 

For any set $X$, $G=\langle X \rangle_F$ is the free group generated by strings of elements from $X$.  The product of two strings is defined as their concatenation and $x^{-1}$ is formally defined so that the set of strings composed of $x$ and $x^{-1}$ forms a group: $\langle X \rangle_F= \{x_1... x_p:  x_i \in X \text{ or } x_i^{-1} \in X \,\, \forall i\}$.  If $R\subseteq \langle X\rangle$ then $G=\langle X | R \rangle_F := \langle X \rangle_F/N $ where $N$ is the smallest normal subgroup of $\langle X \rangle_F$ containing $R$.  If $H$ is some group generated by $X$, then we say group $H$ has a finite presentation given by generators $X$ and relations $R$ if $H\cong \langle X|R\rangle_F$ (isomorphism of groups).  

\subsection{Specht Modules}\label{sec:specht}

Before presenting the proof of \Cref{thm:perm_is_opt} we will need a little background on the irreps of the Symmetric group.  This material is all standard, please see \cite{jam06} or \cite{ful13} for a reference.  Irreps of $S_n$ are parameterized by partitions of $n$.  Formally, $Part_n=\{(\lambda_1, \lambda_2, ..., \lambda_d): \,\, d\leq n, \,\, \lambda_i\in \mathbb{Z}_{ > 0}, \text{ and } \lambda_i \geq \lambda_{i+1} \forall i\}$.  We will designate a partition $(\lambda_1, ..., \lambda_d)$ simply as $\lambda$.  Partitions correspond uniquely to Young diagrams.  A Young diagram consists of rows of adjacent squares where the number of squares in row $i$ is $\lambda_i$:
\begin{center}
$(4, 2, 1) \leftrightarrow $\,\,\,\ydiagram{4, 2, 1}.
\end{center}
A Young Tableaux is a numbering of the boxes of a diagram using $[n]$, i.e.
\begin{center}
\begin{ytableau}
2 & 1 & 7 & 6 \\
3 & 5 \\
4
\end{ytableau}.
\end{center}
We say that a Young tableaux is of shape $\lambda\in Part_n$ if it is obtained by numbering the diagram corresponding to partition $\lambda$.  For $d\leq n$ we define a{\it Young fragment} as a labeling of a diagram from $Part_d$ using a subset of the letters, $A\subseteq [n]$.  This is simply a Young Tableaux for $S_d$ but with some of the letters replaced by elements of $[n]$, i.e.
\begin{center}
\begin{equation}\label{eq:young_frag_ex}
\begin{ytableau}
2 & 1 & 7 & 6 \\
3 
\end{ytableau}.
\end{equation}
\end{center}

A group action of $S_n$ on Young Tableaux can be naturally defined where $\sigma \in S_n$ acts on a tableaux by permuting the letters in each box according to $\sigma$.  Formally if $[a_1, ..., a_p]$ is a row of $t$ then the corresponding row of $\sigma \,t$ is $[\sigma(a_1), ..., \sigma(a_p)]$:
\begin{center}
\begin{equation}\label{eq:action_on_tab}
\begin{array}{|c|c|c|}
\hline
1 & 2 & 3\\
\hline
 4 & 5 & 6\\
\hline
 7&\multicolumn{2}{c}{}\\
\cline{1-1}
\end{array}
\xrightarrow[]{\sigma}
\begin{array}{|c|c|c|}
\hline
\sigma(1) & \sigma(2) & \sigma(3)\\
\hline
 \sigma(4) & \sigma(5) & \sigma(6)\\
\hline
 \sigma(7)&\multicolumn{2}{c}{}\\
\cline{1-1}
\end{array}.
\end{equation}
\end{center}
Given a tableaux $t$ we define $R_t$ and $C_t$ as the row and column stabilizer respectively.  Formally, $\sigma \in R_t$ if every row of $\sigma \,t$ has the same elements as the corresponding row of $t$ (possibly with a permuted order).  In the example in \Cref{eq:action_on_tab}, $R_t$ is the set of all permutations $\sigma\in S_7$ such that $\sigma(\{1, 2, 3\})=\{1, 2, 3\}$, $\sigma(\{4, 5, 6\})=\{4, 5, 6\}$ and $\sigma(7)=7$.  $C_t$ is defined analogously.  Given a Young fragment $f$ the row stabilizer $R_f$ is defined as the set of all permutations which stabilize the rows of $f$ as well as the letters not included in $f$.  For the example in \Cref{eq:young_frag_ex} $n=7$ so $\sigma \in R_f$ if $\sigma(\{2, 1, 7, 6\})=\{2, 1, 7, 6\}$, $\sigma(3)=3$, $\sigma(4)=4$ and $\sigma(5)=5$.  

Given a tableaux or fragment $f$ (note that tableau are also fragments) we can now define the{\it Young symmeterizer} as:
\begin{equation}
Y_f=\sum_{\sigma \in R_f} \sum_{\gamma \in C_f} sign(\gamma) \sigma \gamma \in \mathbb{C}[S_n].
\end{equation}
In general we will denote $a_f =\sum_{\sigma \in R_f} \sigma $ and $b_f=\sum_{\gamma \in C_f } sign(\gamma) \gamma$ so that $Y_f=a_f b_f$.

Young symmeterizers can be used to construct all the representations of $S_n$ as subspaces of $\mathbb{C}[S_n]$.  For $t$ some tableaux of shape $\lambda$ define the subspace
\begin{equation}
    V_t=\left\{\left(\sum_{g\in G} c_g g\right) Y_t: c_g \in \mathbb{C} \,\, \forall g\right\}\subseteq \mathbb{C}[S_n].
\end{equation}
We interpret $V_t$ as a representation of $S_n$ by acting from the left with elements of $S_n$ (a left $S_n-$module) and interpreting the resulting expression using multiplication in the group algebra.  If $t$ and $t'$ are tableau of the same shape $\lambda$ then $V_t \cong V_{t'}$, so we will often denote $V_t$ simply as $V_\lambda$.  The set of subspaces $\{V_\lambda\}$ form a complete set of representatives for irreps of the symmetric group:
\begin{theorem}[\cite{ful13} Theorem 4.3]
For each $\lambda\in Part_n$ let $t(\lambda)$ be some tableaux of shape $\lambda$.  $\{V_{t(\lambda)}\}_{\lambda\in Part_n}$ is a complete set of irreps for $S_n$.
\end{theorem}
The Young symmeterizers act in many ways as projectors onto the corresponding irreps.  Indeed, a crucial property (see \cite{ful13} Lemma 4.26) of $Y_t$ is that $V_t Y_t =V_t$ or equivalently that:
\begin{equation}\label{eq:square_young}
    Y_t^2= n_t Y_t
\end{equation}
with $n_t > 0$.

Let $t$ be a Young tableaux of shape $\lambda$ and let $f$ be a fragment.  We say that $f${\it fits inside} $\lambda$ if $f$ is contained in $t$ when we place the two diagrams on top of each other with the top left corners coincident (see \Cref{fig:fit_no_fit}).  In order to formally define this let $t$ have $k$ rows $r_1, ..., r_k$ with $s_i$ elements in $r_i$.  Let $f$ have $\ell$ rows $q_1, ..., q_\ell$ with row $q_i$ having $x_i$ elements.  We say that $f$ fits inside $\lambda$ or $t$ if $\ell \leq t$  and $x_i \leq s_i$ for all $i \in [\ell]$.  The main fact that we will need in this context is that inside a subspace $V_t$, $V_f=0$ for all fragments $f$ of a given shape (as a constraint) if and only if $f$ does not fit in $t$.  We will only use a special case of this known result \cite{lit34, pro76} (Pieri's rule)\knote{edited here}, namely when $f$ is of shape $\lambda=(1, 1, 1)$, so we provide a simple proof of this fact for the readers convenience.  

\begin{theorem}[\cite{ful13} Exercise 4.44]\label{thm:fragment}
Let $t$ be some tableaux.  If $t$ is of shape $(n-d, d)$ then all fragments $f$ of shape $(1, 1, 1)$ evaluate to $0$ on $V_t$: $Y_f p Y_t=0$ for all $p\in \mathbb{C}[S_n]$.  If $t$ has more than two rows ($t$ is of shape $(\lambda_1, ... \lambda_d)$ for $d\geq 3$) then there exists a fragment $f$ of shape $(1, 1, 1)$ such that $Y_f\neq 0$ on $V_t$.   
\end{theorem}
\begin{proof}
Let $f$ be the fragment corresponding to a single column with the letters $f_1, f_2$ and $f_3$.  For the first part of the theorem by linearity it is sufficient to show $Y_f g Y_t=0 $ for all $g\in S_n$.  It is simple to verify that $gY_t g^{-1}=Y_{gt}$ so $Y_f g Y_t=Y_f Y_{t'} g$ for $t'$ the same shape as $t$.  Since $t$ has two rows, two of $\{f_1, f_2, f_3\}$ must be in the same row of $t'$.  WLOG assume $f_1$ and $f_2$ are in the same row of $t'$.  Since row and column stabilizers are subgroups of the $S_n$ and since $sign(\gamma (f_1, f_2))=-sign(\gamma)$ for all $\gamma \in S_n$, it holds that 
\begin{align}
a_{t'} =\frac{1+(f_1, f_2)}{2} a_{t'}  \,\,\,\,\,\, \text{and} \,\,\,\,\,\,  b_{f}=b_{f} \frac{1-(f_1, f_2)}{2},
\end{align}
where $(f_1, f_2)$ is the transposition of $f_1$ and $f_2$.
It follows that 
$$
Y_f g Y_t= b_f a_{t'} b_{t'}= b_f \frac{1-(f_1, f_2)}{2} \frac{1+(f_1, f_2)}{2} a_{t'} b_{t'} =0
$$

For the second part of the theorem given $t$ we want to give a fragment $f$ and $p\in \mathbb{C}[S_n]$ such that $Y_f p Y_t \neq 0$.  Let $f_1$, $f_2$ and $f_3$ be the first three numbers in the first column of $t$ and let $p=b_t a_t$.  Once again since $R_t$ is a subgroup of $S_n$ it holds that $a_t^2 = |R_t|\, a_t$.  By the definition of $f$ and the fact that $C_t$ is a subgroup it holds that: for any $\sigma \in C_f$, $sign(\sigma) \sigma b_t=b_t$.  Hence,
$$
Y_fp Y_t=b_f b_t a_t a_t b_t =|C_f| \,  b_t a_t a_t b_t= |C_f| \, |R_t|\, b_t a_t b_t. 
$$
Since $a_t (b_t a_t b_t) =Y_t^2 \neq 0$ by \Cref{eq:square_young}, $Y_fp Y_t= b_t a_t b_t\neq 0$.

\end{proof}

The quantum permutation unitaries have a well known decomposition into irreps of $S_n$, we will need the decomposition here.  If $\sigma\in S_n$ is a permutation we can define the permutation unitary 
\begin{equation}
p(\sigma) \ket{x_1} \ket{x_2} ... \ket{x_n}=\ket{x_{\sigma^{-1}(1) }} \ket{x_{\sigma^{-1}(1) }}... \ket{x_{\sigma^{-1} (n)}}.
\end{equation}
$(p, (\mathbb{C}^2)^{\otimes n})$ is a representation of the symmetric group with known decomposition into irreps:
\begin{equation}\label{eq:quantum_irreps}
(\mathbb{C}^2)^{\otimes n} \cong \bigoplus_{\substack{\lambda\in Part_n:\\ \lambda=(n-d, d)}} (n+1-2 d) V_\lambda
\end{equation}

\subsection{Proof of \Cref{thm:perm_is_opt}}\label{sec:perm_conv}

Let $Q_{ij}$ be some feasible solution for $\mathcal{Perm}$.  The set of matrices generated by $\{Q_{ij}\}$ is naturally a group since $Q_{ij}$ is self-inverse.  Let this group be denoted $G'$.  We will demonstrate a homomorphism $\psi: S_n \rightarrow G'$ which implies the operators $\{Q_{ij}\}$ correspond to a representation of the Symmetric group.  Then we note that the final set of constraints, \Cref{eq:anti_comm}, is only valid if the representation can be decomposed into certain irreps of the symmetric group (exactly those in \Cref{eq:quantum_irreps}).  This then implies that the permutation programs gets the optimal quantum eigenvalue by the discussion in \Cref{sec:rep_theory}.

Let $E=\{ij: i, j \in [n] \text{ and } i<j\}$ be the set of all possible undirected edges.  Define $X=\{q_{ij}\}_{ij\in E}$.  Recall from \Cref{sec:rep_theory} that $\langle X \rangle_F$ is the group set of strings with elements from $X$ where multiplication is defined through concatenation.  Here we are treating the operator variables as formal symbols in the context of the free group.  Let $R_1, R_2, R_3 \subseteq \langle X \rangle_F$ be defined as 
\begin{gather*}
R_1=\{q_{ij}^2:\,\,ij \in E\},\\
R_2=\{(q_{ij}q_{kl})^2: \,\, ij, kl \in E \text{ and $i$, $j$, $k$, $l$ all distinct}\},\\
\text{ and }R_3=\{(q_{ij} q_{jk})^3: ij, jk \in E \text{ and $i$, $j$ $k$ all distinct}\}.
\end{gather*}
Using the self-inverse property of the operators $Q_{ij}$ it is clear that $(Q_{ij}Q_{kl})^2=\mathbb{I}$ and $(Q_{ij} Q_{jk})^3=\mathbb{I}$ are equivalent to sets of constraints \Cref{eq:sym_const_2} and \Cref{eq:sym_const_3} respectively.  Let $R=R_1 \cup R_2 \cup R_3$ and let $N$ be the smallest normal subgroup containing $R$. It is known \cite{cox13gen} that $S_n$ has a finite presentation with generators $X$ and relations $R$ so $S_n \cong \langle X\rangle_F/N$.  In particular it is known that the map $\rho: S_n \rightarrow  \langle X\rangle_F/N$ defined by $ \rho((i, j)) = q_{ij} N$ (which is extended to $S_n$ by writing a permutation as a product of transpositions) is well-defined and an isomorphism of groups.  Let $\theta: \langle X \rangle_F \rightarrow G'$ be the map defined as 
\begin{gather*}
\theta(q_{i_1j_1}q_{i_2j_2} ...q_{i_qj_q}) = Q_{i_1j_1}\,\,Q_{i_2j_2} \,\,... \,\, Q_{i_qj_q},  \\
\theta(1) =\mathbb{I}.
\end{gather*}
Since $Q_{ij}$ is self-inverse it is clear that $\theta$ is a homomorphism of groups.  Since $N=\langle g r g^{-1} : \,\, r\in R, g\in \langle X \rangle_F\,\, \rangle $, 
\begin{align*}
&\theta((g_1 r_1 g_1^{-1})(g_2 r_2 g_2^{-1})... (g_q r_q g_q^{-1}))\\=
&\theta(g_1) \theta(r_1) \theta(g_1)^{-1}\theta(g_2) \theta(r_2) \theta(g_2)^{-1} ... \theta(g_q) \theta(r_q) \theta(g_q)^{-1}.
\end{align*}
The operators satisfy constraints from $R$ so $\theta(r_i)=\mathbb{I}$ for all $i$ and $\theta(n)=\mathbb{I}$ for all $n\in N$.  Let $b: \langle X \rangle_F\rightarrow \langle X \rangle_F/N $ be the homomorphism defined as $b(g)=g N$.  The ``Fundamental theorem on homomorphisms'' (Theorem 26 from \cite{sau99alg}) implies the existence of a homomorphism $\phi: \langle X \rangle_F/N \rightarrow G'$ such that $\theta= \phi \circ b$.  Note that $\phi$ must map $q_{ij}N \rightarrow Q_{ij}$.  It follows that the map $\psi:=\phi\circ\rho: S_n \rightarrow G'$ is a homomorphism of groups and that the operators $Q_{ij}$ must correspond to a valid representation of $S_n$. Let this representation be denoted $(\psi, W)$ where $W$ is the Hilbert space on which the $Q_{ij}$ act.

By \Cref{sec:rep_theory} $(\psi, W)$ must be isomorphic to a decomposition of the form $\left(\bigoplus_\lambda m_\lambda \rho_\lambda, \bigoplus_\lambda m_\lambda V_\lambda\right) $ where $V_\lambda$ are the spaces defined in \Cref{sec:specht}.  We will demonstrate that $m_{\lambda}=0$ unless $\lambda=(n-d, d)$.  First note that the final set of constraints  \Cref{eq:anti_comm} all correspond to young fragments (see \Cref{sec:specht})
\begin{align}
    &\mathbb{I}-Q_{ij}-Q_{jk} -Q_{ik}+Q_{ij}Q_{jk}+Q_{jk}Q_{ij}=0 \\ 
    &\Leftrightarrow ~~ \psi(1)-\psi((i, j))-\psi((j, k))-\psi((i, k)) \nonumber\\
    &~~~~~~ +\psi((i, j)(j, k))+\psi((j, k)(i, j))=0, \nonumber
\end{align}
which is then equivalent to $Y_f=0$ for $f$ the fragment
\begin{center}
\begin{equation} 
f=\begin{ytableau}
i \\
j \\
k
\end{ytableau}.
\end{equation}
\end{center}
By \Cref{thm:fragment}, $Y_f=0$ for all such $f$ on a subspace $ V_\lambda$ if and only if $\lambda=(n-d, d)$.  It follows that $m_\lambda>0$ only for $\lambda=(n-d, d)$ so $\mathcal{Perm}(G, w) \geq \SWAP(G, w)$.  Taking the quantum swap operators and the optimal eigenstate as a feasible solution ($p_{ij}=P_{ij}$ from \Cref{eq:swap_def}) we can see that $\mathcal{Perm}(G, w) \leq \SWAP(G, w)$ hence $\mathcal{Perm}(G, w)=\SWAP(G, w)$.

\begin{figure}[h]
\caption{Overall idea of proof.  All Young fragments $f$ of a given shape have value $0$ on subspaces $V_\lambda$ exactly when they don't ``fit inside" $\lambda$ (right side).  By enforcing that all the fragments with a single column of three entries evaluate to zero, we can exactly select for $V_\lambda$ with two rows (exactly those that appear in the decomposition \Cref{eq:quantum_irreps}).}\label{fig:fit_no_fit}
\centering
\includegraphics[width=0.5\textwidth]{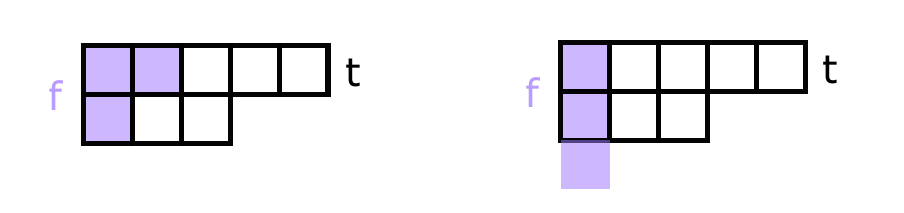}
\end{figure}

\section{Derivations of relations in $\mathcal{Proj}$ from the minimal constraints} \label{app:relationderivation}

In \cref{subsec:qmcop} we argued from \cref{thm:perm_is_opt} that all relations among singlet projectors are derivable just from the minimal relations, namely \cref{eq:anticommproj,eq:singprojnormalize,eq:singproj1,eq:singproj2,eq:singprojcomm}. 
Here, we give concrete examples of such derivation. 
Let us first look into the prominent fact that all Heisenberg Hamiltonians have total spin as a good quantum number.

\begin{proposition}\label{prop:totalspincons}
Any {\scshape QMaxCut} Hamiltonian $H$ preserves the total spin, i.e., commutes with $\sum_{i<j}h_{ij}$.
\end{proposition}

\begin{proof}
    Since any Hamiltonian is a summation of individual projector terms, we can consider a particular $h_{12}$ WLOG, and then if we can show $[h_{12},\sum_{i<j}h_{ij}]=0$, the proof is complete.
    From \cref{eq:singprojcomm}, $[h_{12},h_{ij}]=0$ if $i,j\neq 1,2$, so what remains in the sum $\sum_{i<j}h_{ij}$ that do not trivially commute are $\sum_{j\neq 1,2}(h_{1j}+h_{2j})$. 
    Now we show $[h_{12},h_{1j}+h_{2j}]=0$ for any $j$, proving the proposition. 
    The anticommutation relation \cref{eq:anticommproj} allows expanding $h_{12}$ in terms of $h_{1j}$ and $h_{2j}$, giving us
    \begin{eqnarray}
        &[h_{12},h_{1j}+h_{2j}]  \\
        &=[h_{1j}+h_{2j}-2(h_{1j}h_{2j}+h_{2j}h_{1j}),h_{1j}+h_{2j}]\nonumber\\
        &= -2~\bigl\{
        h_{1j}h_{2j}h_{1j}-h_{1j}^2h_{2j}+h_{1j}h_{2j}^2-h_{2j}h_{1j}h_{2j}\nonumber\nonumber\\
        &
        +h_{2j}h_{1j}^2-h_{1j}h_{2j}h_{1j}+h_{2j}h_{1j}h_{2j}-h_{2j}^2h_{1j}
        \bigr\}=0,\nonumber\label{eq:totalspinlastline}
    \end{eqnarray}
    where we used \cref{eq:singprojnormalize} for the last line. 
    Taking the sum over $j$ gives us the wanted expression. 
\end{proof}

Note again that the relation being derived here $[h_{12},h_{1j}+h_{2j}]=0$, 
like any other relation for singlet projectors, 
is 1. trivially verifiable by explicitly calculating matrices $H_{ij}$s, 
but 2. also derivable when $h_{ij}$s are regarded as abstract algebraic objects
, which was what we have shown here. 
In general, such derivation can become quite complex, since there are cases which require the order of the polynomial to become larger than the degree of the relation itself during the proof. 
In the above proof, the commutation $[h_{12},h_{1j}+h_{2j}]$ is a degree-2 expression, but the derivation required intermediate steps with a degree-3 polynomial as in Eq. (\ref{eq:totalspinlastline}).
To demonstrate how nontrivial the derivation can become, let us consider the following relation that reduces a degree-3 monomial into a degree-2 polynomial. This is one of the simplest cases with 4 qubits. 

\begin{proposition}\label{prop:line3reduce}
    For any distinct $i,j,k,$ and $l$, 
    \begin{align}\label{eq:line3reduce}
        h_{ij}h_{jk}h_{kl}=&\frac{1}{4}\bigl(h_{ij}h_{jk}+h_{jk}h_{kl}+h_{ij}h_{kl}+h_{ik}h_{jl}\\&~~-h_{ij}h_{jl}-h_{ik}h_{kl}-h_{il}h_{jk}\bigr). \nonumber
    \end{align}
\end{proposition}

\begin{proof}
    Similarly to the previous proof, we start by expanding $h_{jk}$ and $h_{kl}$ using the anticommutation relation \cref{eq:anticommproj}, with $l$ and $j$ as the ``pivot" respectively, obtaining
    \begin{equation}\label{eq:derivation1}
h_{ij}h_{jk}h_{kl}=\frac{1}{2}\left(\frac{1}{2}\left( h_{ij}h_{jk} + h_{ij}h_{kl}\right)-h_{ij}h_{jk}h_{jl} - h_{ij}h_{jl}h_{kl}\right)
    \end{equation}
    after organizing with \cref{eq:singprojnormalize}. 
    Since we can obtain
    \begin{equation}
        h_{jl}h_{jk}h_{ij}=\frac{1}{2}h_{jl}h_{ij}+h_{ij}h_{jl}h_{kl}-\frac{1}{2}(h_{ij}+h_{jl}-h_{il})h_{kl}, 
    \end{equation}
    by expanding $h_{jk}$ with $l$ and applying \cref{eq:anticommproj}, we can substitute the last term of \cref{eq:derivation1} to get 
    \begin{align}\label{eq:derivation2}
    h_{ij}h_{jk}h_{kl}=\frac{1}{2}\bigl(&\frac{1}{2} h_{ij}h_{jk} -h_{ij}h_{jk}h_{jl} - h_{jl}h_{jk}h_{ij} \nonumber\\
    &+\frac{1}{2}\left(h_{jl}h_{ij}-h_{jl}h_{kl}
    +h_{il}h_{kl}\right)\bigr),
    \end{align}
    after cancellation. 
    Thanks to the fact that the remaining degree-3 terms are exactly the same up to ordering, we can evaluate the sum of them by reordering $h_{jl}h_{jk}h_{ij}$ using the anticommutation relation \cref{eq:anticommproj} to make $h_{ij}h_{jk}h_{jl}$. If the reordering is done from right to left, we obtain 
    \begin{align}\label{eq:treesum}
        h_{ij}h_{jk}h_{jl}+h_{jl}h_{jk}h_{ij}= &\frac{1}{2}\left(h_{il}h_{jk}-h_{ij}h_{kl}-h_{jl}h_{ik}\right)\nonumber\\
        &~~+\frac{1}{4}\left(h_{ij}+h_{jl}-h_{il}\right). 
    \end{align}
    Plugging this into \cref{eq:derivation2} gives us 
    \begin{align}\label{eq:derivation3}
        h_{ij}h_{jk}h_{kl}=&\frac{1}{4}\bigl(h_{ij}\left(h_{jk}-h_{jl}+h_{kl}\right)-h_{jl}h_{kl}\nonumber\\
        &~~~~~~+h_{il}h_{kl}+h_{jl}h_{ik}-h_{il}h_{jk}\bigr),
    \end{align}
    after cancellation and using \cref{eq:anticommproj} to clean up the degree-1 terms. 
    Finally, we can use the relation 
    \begin{equation}
        (h_{ij}+h_{jk})h_{ik}=(h_{il}+h_{kl})h_{ik} \label{eq:SingProjBasicNonHermitian}
    \end{equation}
    obtainable by starting from \cref{eq:quarterformula} to have $h_{ik}(h_{ij}+h_{jk})h_{ik}=h_{ik}(h_{il}+h_{kl})h_{ik}$ and then reordering both sides again using \cref{eq:anticommproj}. 
    Applying this to \cref{eq:derivation3} results in \cref{eq:line3reduce}. 
    
\end{proof}

Of course from this proof alone we cannot completely conclude that {\it any} derivation of \cref{prop:line3reduce} must be at least this long.
However, it does show that a carefully designed sequence of formula application is needed to have the right cancellations to occur and finally enable us to use something like \cref{eq:treesum}, a somewhat easy degree-3 term to reduce. 

We also show several derivations of other relations that are needed to show convergence.

\begin{proposition}\label{prop:star3reduce}
    For any distinct $i,j,k,a, b$ and $c$ the following equations hold.
    \begin{align}\label{eq:2star_to_2star}
        h_{ij}h_{ik}-h_{ik} h_{jk}=(h_{ij}-h_{jk})/2.
    \end{align}
    \begin{align}\label{eq:triangle_reduce}
        h_{ij}h_{jk}h_{ik}=h_{jk}h_{ik}+(h_{ij}h_{ik}-h_{ik})/2.
    \end{align}
    \begin{align}\label{eq:star3reduce}
        h_{ij}h_{ik}h_{ia}=&\frac{1}{4}\bigl(h_{ij}h_{ik}+h_{ij}h_{ia}-h_{ij}h_{ka}+h_{ia}h_{ka}\\&~~+h_{ik}h_{ja}-h_{ja}h_{ka}-h_{ia}h_{jk}\bigr). \nonumber
    \end{align}   
    \begin{align}\label{eq:double2starreduce}
    4h_{ij}h_{jk}h_{ab}h_{bc}=h_{ia}h_{jb}h_{kc} + h_{ib}h_{jc}h_{ka} + h_{ic}h_{ja}h_{kb} \\
    \nonumber- h_{ia}h_{jc}h_{kb} - h_{ib}h_{ja}h_{kc} - h_{ic}h_{jb}h_{ka}\\
    \nonumber + h_{ij}h_{jk}(h_{ab} + h_{bc} - h_{ac}) + h_{ij}h_{ab}h_{bc}\\
    \nonumber - (h_{ij}h_{ja} - h_{ja}h_{ac} + h_{ij}h_{jc} + h_{ia}h_{ac})h_{kb}\\
   \nonumber  + (h_{ij}h_{ja} + h_{ja}h_{ab} - h_{ij}h_{jb} - h_{ia}h_{ab})h_{kc} \\
   \nonumber + (h_{ij}h_{jb} + h_{jb}h_{bc} - h_{ij}h_{jc} - h_{ib}h_{bc})h_{ka}\\
\nonumber + (h_{ia}h_{ac} - h_{ia}h_{ab} - h_{ib}h_{bc})h_{jk} \\
\nonumber + h_{ik}(h_{ja}h_{ab} + h_{jb}h_{bc} - h_{ja}h_{ac})
+ h_{ja}h_{ab}h_{bc} \\
\nonumber - h_{ij}h_{jb}h_{bc} - h_{ia}h_{ab}h_{bc} + h_{jb}h_{ab}h_{bc} - h_{ib}h_{ab}h_{bc} \\
\nonumber - h_{jb}h_{kb}h_{bc}
+ h_{ib}h_{kb}h_{bc} - h_{ij}h_{ja}h_{ab} - h_{ja}h_{ka}h_{ab} \\
\nonumber + h_{ia}h_{ka}h_{ab} + h_{ij}h_{ja}h_{ac} + h_{ja}h_{ka}h_{ac} - h_{ia}h_{ka}h_{ac}.
    \end{align}
\end{proposition}

\begin{proof}
In order to prove \Cref{eq:2star_to_2star} we can observe 
\begin{align}\label{eq:4}
     \frac{(h_{ij} - h_{jk})}{2} - h_{ij}  h_{ik} +  h_{ik}  h_{jk} + \\
    \nonumber \frac{1}{2} (-\frac{h_{ij}}{2} - \frac{h_{ik}}{2} + \frac{h_{jk}}{2} + h_{ij}  h_{ik} + h_{ik}  h_{ij}) \\
    \nonumber -  h_{ik}  (\frac{h_{ij}}{2} - \frac{h_{ik}}{2} - \frac{h_{jk}}{2} + h_{ik}  h_{jk} + h_{jk}  h_{ik}) \\
    \nonumber + (-h_{ik} +     h_{ik}^2)  h_{jk} -  h_{jk}  (-h_{ik} + h_{ik}^j)  \\
    \nonumber +\frac{1}{2} (-\frac{h_{ij}}{2} +\frac{h_{ik}}{2}+ \frac{h_{jk}}{2} - h_{ik}  h_{jk} - h_{jk}  h_{ik}) \\
    \nonumber + (\frac{h_{ij}}{2} - \frac{h_{ik}}{2} - \frac{h_{jk}}{2} + h_{ik}  h_{jk} +     h_{jk}  h_{ik})  h_{ik}=0
\end{align}
when we expand it.  Since the first three terms correspond to the identity we want to prove and the rest of the terms correspond to constraints, we can conclude \Cref{eq:2star_to_2star}.

In order to prove \Cref{eq:triangle_reduce} we can relabel \Cref{eq:2star_to_2star} to obtain the substitution $h_{ij}h_{jk} \rightarrow h_{jk}h_{ik}+(h_{ij}-h_{ik})/2$.
\begin{align}\label{eq:5}
    h_{ij}h_{jk} h_{ik}=(h_{jk}h_{ik}+(h_{ij}-h_{ik})/2)h_{ik}\\
    \nonumber = h_{jk}h_{ik}+(h_{ij}h_{ik}-h_{ik})/2.
\end{align}

    In order to prove \Cref{eq:star3reduce} we first start by expanding the final term $h_{ia}$ with $k$ as the pivot using \cref{eq:anticommproj}. This will yield 
    \begin{align}
        h_{ij}h_{ik}h_{ia}=& h_{ij}h_{ik} (h_{ka}+h_{ik}-2h_{ik}h_{ka}-2h_{ka}h_{ik})\\
        =&h_{ij}h_{ik}h_{ka} + h_{ij}h_{ik}^2\nonumber\\
          &~~~~~~~-2h_{ij}h_{ik}^2h_{ka} -2h_{ij}h_{ik}h_{ka}h_{ik}\label{eq:StarReductionProofIntermediate}\\
          =&\frac{1}{2}h_{ij}h_{ik}-h_{ij}h_{ik}h_{ka}\label{eq:StarReductionProofIntermediate2}
    \end{align}
    after applying \cref{eq:quarterformula} to the last term in \cref{eq:StarReductionProofIntermediate}. Now the second term in \cref{eq:StarReductionProofIntermediate2} is the same as the degree-3 term reduced in \cref{eq:line3reduce} up to relabeling of qubits. We can repeat the derivation in the proof of \cref{prop:line3reduce} to apply \cref{eq:line3reduce} to the second term, and obtain 
    \begin{align}
        h_{ij}h_{ik}h_{ia}=& 
        \frac{1}{4}(h_{ij}h_{ik}-h_{ik}h_{ka}-h_{ij}h_{ka}+h_{ij}h_{ia}\nonumber\\
        &~~~~~~~ +h_{jk}h_{ka}
        -h_{ia}h_{jk}+h_{ik}h_{ja}).\label{eq:StarReductionProofIntermediate3}
    \end{align}
    Finally, we can apply the relation $(h_{ia}-h_{ja})h_{ka}=(h_{jk}-h_{ik})h_{ka}$ to replace the two terms in \cref{eq:StarReductionProofIntermediate3} to obtain \cref{eq:star3reduce}. The relation is a  relabeling of \cref{eq:SingProjBasicNonHermitian}.

  For \Cref{eq:double2starreduce} we begin with $h_{ij} h_{jk} h_{ab} h_{bc}$ and expand $h_{jk} \rightarrow h_{ka}+h_{ja}-2 h_{ka}h_{ja}-2 h_{ja} h_{ka}$ to obtain
\begin{align}\label{eq:6}
h_{ij}h_{jk} h_{ab} h_{bc}=h_{ij}h_{ka}h_{ab}h_{bc}+h_{ij}h_{ja}h_{ab}h_{bc}\\
\nonumber-2(h_{ij} h_{ka}h_{ja}h_{ab}h_{bc}+h_{ij}h_{ja}h_{ka}h_{ab}h_{bc})
\end{align}
Relabelling \Cref{eq:star3reduce} and \Cref{eq:line3reduce} yields
\begin{align}\label{eq:relabelled_star}
    h_{ak}h_{aj}h_{ab}=\frac{1}{4}(h_{ak}h_{aj}+h_{ak}h_{ab}-h_{ak}h_{jb}+h_{ab}h_{jb}+\\
    \nonumber h_{aj}h_{kb}-h_{kb}h_{jb}-h_{ab}h_{jk})
\end{align}
and 
\begin{align}\label{eq:relabelledline}
    h_{ka}h_{ab}h_{bc}=\frac{1}{4}(h_{ka}h_{ab}+h_{ab}h_{bc}+h_{ka}h_{bc}+h_{kb}h_{ac}\\
    \nonumber -h_{ka}h_{ac}-h_{kb}h_{bc}-h_{kc}h_{ab}).
\end{align}
From \Cref{eq:relabelled_star} and \Cref{eq:relabelledline} we obtain
\begin{align}\label{eq:1}
    h_{ij} h_{ka}h_{ja}h_{ab}h_{bc}+h_{ij}h_{ja}h_{ka}h_{ab}h_{bc}\\
    \nonumber=\frac{1}{4}(h_{ij}h_{ak}h_{aj}h_{bc}
    +h_{ij}h_{ak}h_{ab}h_{bc}-h_{ij}h_{ak}h_{jb}h_{bc}\\
    \nonumber +h_{ij}h_{ab}h_{jb}h_{bc}+h_{ij}h_{aj}h_{kb}h_{bc}-h_{ij}h_{kb}h_{jb}h_{bc}\\
\nonumber-h_{ij}h_{ab}h_{jk}h_{bc}+ h_{ij}h_{ja}h_{ka}h_{ab} +h_{ij}h_{ja}h_{ab}h_{bc}\\
\nonumber+h_{ij}h_{ja}h_{ka}h_{bc}+h_{ij}h_{ja}h_{kb}h_{ac}-h_{ij}h_{ja}h_{ka}h_{ac}\\
\nonumber-h_{ij}h_{ja}h_{kb}h_{bc}-h_{ij}h_{ja}h_{kc}h_{ab})
\end{align}
The 5th term and 13th term cancel in \Cref{eq:1}.  Applying a commutation relation to the 7th term yields
\begin{align}\label{eq:2}
    h_{ij} h_{ka}h_{ja}h_{ab}h_{bc}+h_{ij}h_{ja}h_{ka}h_{ab}h_{bc}\\
    \nonumber=-\frac{1}{4}h_{ij}h_{jk}h_{ab}h_{bc}+\frac{1}{4}(h_{ij}h_{ak}h_{aj}h_{bc}
    +h_{ij}h_{ak}h_{ab}h_{bc}\\
    \nonumber -h_{ij}h_{ak}h_{jb}h_{bc}+h_{ij}h_{ab}h_{jb}h_{bc}-h_{ij}h_{kb}h_{jb}h_{bc}\\
    \nonumber+ h_{ij}h_{ja}h_{ka}h_{ab} +h_{ij}h_{ja}h_{ab}h_{bc}+h_{ij}h_{ja}h_{ka}h_{bc}\\
    \nonumber+h_{ij}h_{ja}h_{kb}h_{ac}-h_{ij}h_{ja}h_{ka}h_{ac}-h_{ij}h_{ja}h_{kc}h_{ab}).
\end{align}
Substituting \Cref{eq:2} into \Cref{eq:6} we obtain
\begin{align}
    h_{ij}h_{jk} h_{ab} h_{bc}=h_{ij}h_{ka}h_{ab}h_{bc}+h_{ij}h_{ja}h_{ab}h_{bc} \\
    \nonumber +\frac{1}{2} h_{ij}h_{jk}h_{ab}h_{bc} -\frac{1}{2}(h_{ij}h_{ak}h_{aj}h_{bc}
    +h_{ij}h_{ak}h_{ab}h_{bc}-\\
    \nonumber h_{ij}h_{ak}h_{jb}h_{bc}+h_{ij}h_{ab}h_{jb}h_{bc}-h_{ij}h_{kb}h_{jb}h_{bc}\\
    \nonumber+ h_{ij}h_{ja}h_{ka}h_{ab} +h_{ij}h_{ja}h_{ab}h_{bc}+h_{ij}h_{ja}h_{ka}h_{bc}\\
    \nonumber+h_{ij}h_{ja}h_{kb}h_{ac}-h_{ij}h_{ja}h_{ka}h_{ac}-h_{ij}h_{ja}h_{kc}h_{ab}),
\end{align}
or
\begin{align}\label{eq:3}
    \frac{1}{2}h_{ij}h_{jk} h_{ab} h_{bc}=h_{ij}h_{ka}h_{ab}h_{bc}+h_{ij}h_{ja}h_{ab}h_{bc} \\
    \nonumber -\frac{1}{2}(h_{ij}h_{ak}h_{aj}h_{bc}
    +h_{ij}h_{ak}h_{ab}h_{bc}-h_{ij}h_{ak}h_{jb}h_{bc}\\
    \nonumber +h_{ij}h_{ab}h_{jb}h_{bc}-h_{ij}h_{kb}h_{jb}h_{bc}
    + h_{ij}h_{ja}h_{ka}h_{ab} \\ \nonumber+h_{ij}h_{ja}h_{ab}h_{bc}+h_{ij}h_{ja}h_{ka}h_{bc}
    +h_{ij}h_{ja}h_{kb}h_{ac}\\
    \nonumber-h_{ij}h_{ja}h_{ka}h_{ac}-h_{ij}h_{ja}h_{kc}h_{ab}).
\end{align}
Finally we note all terms on the R.H.S. of \Cref{eq:3} may be reduced to degree-3 using either the 3-star constraint (\Cref{eq:star3reduce}) or the 3-line constraint (\Cref{eq:line3reduce}).  The final expression is given in \Cref{eq:double2starreduce}.
\end{proof}

The proofs above constitute implicit proofs that the identities above are elements of $U^\ell(\{\eta_j\})$ for small $\ell$ where $U^\ell$ is defined in \Cref{eq:U_ell_def} and $\{\eta_j\}$ corresponds to the constraints described in \Cref{eq:singproj1}-\Cref{eq:singproj2}.  Indeed algebraic substitutions can be interpreted as adding elements from $U^\ell(\{\eta_j\})$: $h_{ij}h_{jk} h_{ab}\rightarrow (-h_{jk}h_{ij}+(h_{ij}+h_{jk}-h_{ik}))h_{ab}$ can be interpreted as starting with $h_{ij}h_{jk} h_{ab}$ and adding $-(h_{ij}h_{jk} +h_{jk}h_{ij}-(h_{ij}+h_{jk}-h_{ik})/2)h_{ab}$.  The entire proof of a particular identity can be interpreted as starting with some identity which we wish to prove, i.e.  $h_{ij}h_{ik}-h_{ik} h_{jk}-(h_{ij}-h_{jk})/2$ and adding elements of $U^\ell(\{\eta_j\})$ until we reach zero at which point we have proven the identity plus some element of $U^\ell(\{\eta_j\})$ is zero.  This implies the identity itself is inside $U^\ell(\{\eta_j\})$.  Note that we wrote the proof of \Cref{eq:2star_to_2star} exactly in this way.  As is common in quantifying the degree of Sum of Squares proofs \cite{bar14}, in order to determine the minimal $\ell$ for which a particular constraint is in $U^\ell$ we must examine the proof and keep track of the max degree of any polynomial used.  This is crucial in examining i.e. the proof of \Cref{eq:star3reduce}, which is actually a degree-$5$ proof since it invokes \Cref{eq:quarterformula} rather than degree-$4$ as one might expect from a cursory glance.  We can restate the polynomial identities we have proven as follows.  We will use the notation ``$a=b$ $\in U^\ell(\{\eta_j\})$'' to mean that $a-b \in U^\ell(\{\eta_j\})$
\begin{corollary}\label{cor:alg1}
    Let $\{\eta_j\}$ be the constraints described in \Cref{eq:singproj1}-\Cref{eq:singproj2}.  For any distinct $i,j,k,a, b$ and $c$ the following statements hold.
    \begin{enumerate}
        \item  \label{item:alg3}\Cref{eq:triangle_reduce} $\in U^3(\{\eta_j\})$.
        \item\label{item:alg1} \Cref{eq:line3reduce} $\in U^4(\{\eta_j\}$.
        \item  \label{item:alg2} \Cref{eq:2star_to_2star} $\in U^3(\{\eta_j\})$
        \item  \label{item:alg4}\Cref{eq:star3reduce} $\in U^5(\{\eta_j\})$.
        \item  \label{item:alg5}\Cref{eq:double2starreduce} $\in U^7(\{\eta_j\})$
    \end{enumerate}
\end{corollary}

Using the NC Algebra Mathematica Package \cite{NCalg} we were able to prove a slightly stronger result for \Cref{item:alg4} which propagates and proves a slightly stronger result for \Cref{item:alg5}.  Convergence of the projector hierarchy at level $\ell^*=\lceil n/2 \rceil+3$ requires this stronger result while we can prove convergence at level $\ell^*=\lceil n/2 \rceil+4$ with \Cref{cor:alg1}.  Writing out the full proof (a very long equation of the form of \Cref{eq:2star_to_2star}) would be tedious so we only give the statement here.  
\begin{corollary}\label{cor:alg2}
     Let $\{\eta_j\}$ be the constraints described in \Cref{eq:singproj1}-\Cref{eq:singproj2}.  For any distinct $i,j,k,a$ \Cref{eq:star3reduce} $\in U^4(\{\eta_j\})$ and \Cref{eq:double2starreduce} $\in U^6(\{\eta_j\})$.
\end{corollary}

The previous examples can be seen as a degree-reduction formula of singlet projectors from degree-$i$ to degree-$j$ for $i\geq j$. These formulae let us reduce the degree of even-higher degree terms such as $h_{ij}h_{jk}h_{kl}h_{lm}$ because the derivation can be analogously done even when an additional term like $h_{lm}$ is present on the leftmost or rightmost of the monomial. 

\section{Nonexactness proofs of $NPA_1(\mathcal{Proj})$ with eigenvalue enumeration}
In this section, we provide proofs that $NPA_1(\mathcal{Proj})$ cannot obtain the exact extremal eigenvalue for two different types of graphs. The first is a crown graph \cref{fig:SmallGraphs} with a specific weight, and the second is uniform complete graphs with odd number of vertices. 
For both proofs, the essence is that we can explicitly construct a PSD $NPA_1(\mathcal{Proj})$ moment matrix that has an exceedingly large cost function value compared to the true extremal eigenvalue. The most nontrivial part is always showing that the matrix is PSD, which here we show by enumerating all eigenvalues explicitly. 
While for the odd complete graph case, we can show PSDness by constructing all the Gram vectors of the moment matrix analytically (which we do in \cref{subsec:complete}), that becomes too complicated for the crown graph case. Thus we provide proofs for both cases with eigenvalue enumeration here for completeness. 

\subsection{Nonexactness of $NPA_1(\mathcal{Proj})$ for certain weighted crown graphs}\label{app:crown}

Here, we prove that $NPA_1(\mathcal{Proj})$ fails to obtain the exact extremal eigenvalue for the crown graph 
\begin{equation}
    H=x h_{ab} + \sum_{j=1}^n \left( h_{aj}+h_{bj}\right) ,
\end{equation}
with certain range of the weight $x$ for the ``center edge" $h_{ab}$. Note that the $n+2$ vertices in total are labeled as $a, b, 1, 2, 3, \ldots, n$ and we assume $n\geq 3$. 
The true extremal eigenvalue of this Hamiltonian is 
\begin{equation}\label{eq:crownenergy}
    E_{\mathrm{GS}}=
    \begin{cases}
        n+1, & x\leq 1+\frac{n}{2},\\
        x+\frac{n}{2}, & x\geq 1+\frac{n}{2}.
    \end{cases},
\end{equation}
with two types of ground states for each $x$ range, depicted in Fig. \ref{fig:SmallGraphs} (b). 
As we proved in Sec. \ref{subsubsec:crown}, $\SoS$ obtains the exact ground state for ranges $x\leq (n+2)^2/4(n+1)$ and $x\geq n$. What we prove here is that $\SoS$ fails for the region $(n+2)/3<x<n$, i.e., that it gives a strictly larger value than the true extremal eigenvalue. We do this by explicitly constructing a moment matrix $M_1^{\mathbb{R}}$ that is a feasible solution for $NPA_1^\mathbb{R}(\mathcal{Proj})$ achieving the value $(3n^2 + 3(n-2)x)/4(n-1)$, which is strictly larger than the true value \cref{eq:crownenergy} in the aforementioned region. 
Note that the region we prove inexactness here matches the boundary where $\SoS$ fails/succeeds at $x=n$, but not for the $x=(n+2)^2/4(n+1)$ boundary. The intermediate region $(n+2)^2/4(n+1)<x<(n+2)/3$ is left as an open problem, although numerics strongly suggest that SDP indeed fails in that region.

\begin{proof}
Consider the following symmetric moment matrix $M$ which has columns and rows labeled by the identity $\mathbb{I}$ and $h_{ai}, h_{bi}, h_{ab}, h_{ij}$ for $i, j\in [n]$ and $i<j$. 
We set the matrix elements as 
\begin{align}\label{eq:oddcompmm}
    &M(\mathbb{I},h_{ab})=M(h_{ab},h_{ab})=\frac{3(n-2)}{4(n-1)},\\    &M(\mathbb{I},h_{ai})=M(h_{ai},h_{ai})
    =\frac{3n}{8(n-1)},\\
    &M(h_{ai},h_{aj})=\frac{3n}{16(n-1)},\\
    &M(h_{ai},h_{bi})=\frac{3}{8(n-1)},\\
    &M(h_{ai},h_{bj})=\frac{3(n+2)}{16(n-1)},\\
    &M(h_{ai},h_{ab})=M(h_{bi},h_{ab})=\frac{3(n-2)}{16(n-1)}\\
    &M(\hat O,h_{ij})=0,
\end{align}
where $\hat O$ is any operator, i.e., the rows and columns for $h_{ij}$ are all 0. 
By construction, $M$ satisfies the requirements \cref{eq:singprojnormalize,eq:singprojcomm,eq:anticommproj,eq:singproj1,eq:singproj2}. 
The only nontrivial constraint that needs to be checked is $M\succeq 0$, which we show by listing all the eigenvalues and associated eigenvectors of $M$: 
\begin{enumerate}
    \item[(i)] $n(n-1)/2$ eigenvectors with eigenvalue trivially 0.
    \item[(ii)] $(n-1)$ eigenvectors with eigenvalue $3n/8(n-1) >0$. 
    \item[(iii)] $(n+1)$ eigenvectors with eigenvalue nontrivially 0. 
    \item[(iv)] Two eigenvectors with positive eigenvalues of the form $(\alpha,\beta,1,\ldots,1,0,\ldots,0)$ . 
\end{enumerate}
Since these $n(n-1)/2+(n-1)+(n+1)+2=1+{n+2 \choose 2}$ eigenvalues exhaust all of the eigenvalues of $M$ (size $1+{n+2 \choose 2}$), confirming the above list concludes the proof. 
In the following we confirm the eigenvectors belonging to the eigenvalues $(\mathrm{i})$ - $(\mathrm{iv})$ above. The vector elements are labeled by the operators in the same way as the moment matrix. We use subscripts and superscripts for labeling the eigenvectors and use brackets for specifying the element.

(i) The eigenvectors of the form 
\begin{equation}
    V^{\mathrm{(i)}}_{ij}(\hat O) = 
    \begin{cases}
    1,  & \hat{O}=h_{ij},\\
    0,  & \mathrm{otherwise},
    \end{cases}
\end{equation}
for all $i,j\in [n], i<j$. Such vectors have eigenvalue trivially 0 and all $n(n-1)/2$ of them are linearly independent. 

(ii) The eigenvectors of the form 
\begin{equation}
    V^{\mathrm{(ii)}}_{ij}(\hat O) = 
    \begin{cases}
    +1,  & \hat{O}=h_{ai} \mathrm{~or~} \hat{O}=h_{bj},\\
    -1,  & \hat{O}=h_{bi} \mathrm{~or~} \hat{O}=h_{aj},\\
    0,  & \mathrm{otherwise}. 
    \end{cases}
\end{equation}
It is straightforward to confirm that these vectors indeed have eigenvalue $3n/8(n-1)$, and there are $n-1$ linearly independent such vectors. One example of such a linearly independent set would be $\{ V^{\mathrm{(ii)}}_{12},V^{\mathrm{(ii)}}_{13},\ldots,V^{\mathrm{(ii)}}_{1n} \}$. Specifically, $V^{\mathrm{(ii)}}_{jk}$ is a linear combination of $V^{\mathrm{(ii)}}_{ij}$ and $V^{\mathrm{(ii)}}_{ik}$.

(iii) The eigenvectors of the form 
\begin{equation}
    V^{\mathrm{(iii)}}_{j}(\hat O) = 
    \begin{cases}
    +1,  & \hat{O}=h_{aj} \mathrm{~or~} \hat{O}=h_{bj} \mathrm{~or~} \hat{O}=h_{ab},\\
    -3/2, & \hat{O}=\mathbb{I},\\
    0,  & \mathrm{otherwise},
    \end{cases}
\end{equation}
for $j=1,2,\ldots,n$ and one other:
\begin{equation}
    V^{\mathrm{(iii)}}_{a}(\hat O) = 
    \begin{cases}
    +1,  & \hat{O}=h_{ai},\\
    -3n/4, & \hat{O}=\mathbb{I},\\
    n/2, & \hat{O}=h_{ab},\\
    0,  & \mathrm{otherwise}.
    \end{cases}
\end{equation}
It is again straightforward to confirm that these vectors indeed have eigenvalue 0, and there are $n+1$ linearly independent such vectors, namely $n$ of the form $V^{\mathrm{(iii)}}_{j}$ and a single $V^{\mathrm{(iii)}}_{a}$.

(iv) The eigenvectors of the form 
\begin{equation}
    V^{\mathrm{(iv)}}_{\pm}(\hat O) = 
    \begin{cases}
    \alpha_{\pm}, & \hat{O}=\mathbb{I},\\
    \beta_{\pm}, & \hat{O}=h_{ab},\\
    1,  & \mathrm{otherwise}.
    \end{cases}
\end{equation}
The equation for this vector to be an eigenvector yields the following two solutions: 
\begin{eqnarray}
    \alpha_{\pm} &=& \frac{6n^2-34n+64\pm 2\sqrt{9n^4+24n^3+121n^2-368n+592}}{3(6-7n)}\nonumber\\    
    \beta_{\pm} &=& \frac{3n^2-3n+20\pm\sqrt{9n^4+24n^3+121n^2-368n+592}}{6-7n}\nonumber
\end{eqnarray}
where the $\pm$ sign in $\alpha$ and $\beta$ must be set the same. This gives two solutions $V^{\mathrm{(iv)}}_{+}$ and $V^{\mathrm{(iv)}}_{-}$, where the eigenvalues are
\begin{equation}\label{eq:crowneigen}
    \lambda_{\pm}=\frac{20-17n-3n^2\pm\sqrt{9n^4+24n^3+121n^2-368n+592}}{16(1-n)},
\end{equation}
again corresponding to the two solutions $V^{\mathrm{(iv)}}_{\pm}$. These two eigenvalues are always positive when $n\geq3$, which concludes the proof. 

\end{proof}

\subsection{List of all eigenvectors of odd complete graphs}\label{app:compeigen}

Here, we list up all the eigenvectors and eigenvalues of the $NPA_1(\mathcal{Proj})$ moment matrix constructed in Sec. \ref{subsec:complete} for the odd complete graphs. 
This would provide an alternative proof for the inexactness of $NPA_1(\mathcal{Proj})$ for odd complete graphs, with the same approach as Appendix \ref{app:crown}, but more importantly follows how the analgous {\it classical} case was proved more generally \cite{gri01lin,lau03low,kun22spec}. 
Using the same notation as in Sec. \ref{subsec:complete}, the eigenvalues are: 
\begin{enumerate}
    \item[(i)] $(n-1)$-degenerate eigenvalues of $an/4 - (n-3)b = 0$.
    \item[(ii)] $n(n-3)/2$-degenerate eigenvalues of $a/2 - b >0$. 
    \item[(iii)] Two eigenvectors of the form $(x,1,1,\ldots,1)$ with eigenvalues $0$ ($x<0$) and positive ($x>0$),  
\end{enumerate}
which adds up to $n-1+n(n-3)/2+2=1 + {n \choose 2}$ in total, matching the size of $M$. The explicit forms of each eigenvectors and their degeneracy (dimension of eigenspace) are shown in the following. 

(i) The eigenvectors of the form 
\begin{equation}
    V^{(\mathrm{i})}_{\alpha\beta}(\hat O) = 
    \begin{cases}
    0,  & \hat{O}=\mathbb{I} \mathrm{~or~} h_{\alpha\beta}\mathrm{~or~} h_{ij} \mathrm{~with~}i,j\neq \alpha,\beta,\\
    +1, & \hat{O}=h_{ij}, ~i=\alpha \mathrm{~or~} j=\alpha,\\
    -1, & \hat{O}=h_{ij}, ~i=\beta \mathrm{~or~} j=\beta,\\
    \end{cases}
\end{equation}
where $\alpha$ and $\beta$ are two distinct vertices chosen beforehand and the vector elements are labeled with the operators corresponding to rows and columns of the moment matrix. 
While there are superficially $n(n-1)$ different possible eigenvectors of the form $V_{\alpha\beta}$, 
many of them are not linearly independent, 
the most obvious ones being $V_{\alpha\beta}=-V_{\beta\alpha}$. 
It turns out that there are exactly $n-1$ of these eigenvectors that are linearly independent. 
Confirming linear independence of the set $\{ V_{1\beta}\}_{\beta=2,3,\ldots,n}$ is straightforward, as well as confirming that these are indeed eigenvectors and have eigenvalue 0. 

(ii) The eigenvectors of the form 
\begin{equation}
    V^{(\mathrm{ii})}_{C}(\hat O) = 
    \begin{cases}
    0,  & \hat{O}=\mathbb{I}\mathrm{~or~} h_{ij} \mathrm{~with~}(ij)\notin C,\\
    \pm 1, & \hat{O}=h_{ij}, ~(ij) \in C\mathrm{~with~parity~}\pm 1,\\
    \end{cases}
\end{equation}
where $C$ is a simple cycle of any even length, with alternating signs associated to each edge it passes.  
Again, although there are many distinct cycles $C$, 
there are only $n(n-3)/2$ linearly independent $V_{C}$'s. 
Confirming the linear independence of the set $\{V_{C}\}_{C\in \mathcal{C}}$ is also easy when we set 
\begin{align}
    \mathcal{C} = \bigl\{ (i,i+1,i+2,\ldots,i+k,i) &~|~ \nonumber\\ i=1,2,\ldots n ;
    ~k=&4,6,\ldots ,n-1\bigr\}, 
\end{align}
(notice that each even ``chord" of the complete graph $(i+k,i)$ appears exactly once)
as well as confirming that these are indeed eigenvectors that have eigenvalues $a/2-b$. 

(iii) Finally, it is straight forward to see that 
\begin{equation}
    V^{(\mathrm{iii})}_{x}(\hat O) = 
    \begin{cases}
    x,  & \hat{O}=\mathbb{I},\\
    1, & \hat{O}=h_{ij},\\
    \end{cases}
\end{equation}
is an eigenvector when either 
\begin{equation}
\begin{cases}
    x= -a {n \choose 2} &, \mathrm{~Eigenvalue~} 0,\\
    x= a^{-1} &, \mathrm{~Eigenvalue~} 1+\frac{n(n-2)^2}{32(n-1)}>0.\\
\end{cases}
\end{equation}

\section{Description of $NPA_1(\mathcal{Proj})$ for the hexagon}\label{app:hexagon}
Numerically, we observe clear evidence that both $NPA_1(\mathcal{Proj})$ and $NPA_2(\mathcal{Pauli})$ obtain the exact ground state for the cycle graph with $n=6$ qubits. 
While unfortunately we have not been able to obtain an analytic $\SoS$ that verifies the numerically observed solvability, we here present the structure of the Gram vectors corresponding to the optimal (exact) moment matrix. 

\begin{figure*}[t]
    \begin{minipage}[b]{0.66\linewidth} 
    \includegraphics[width=10cm]{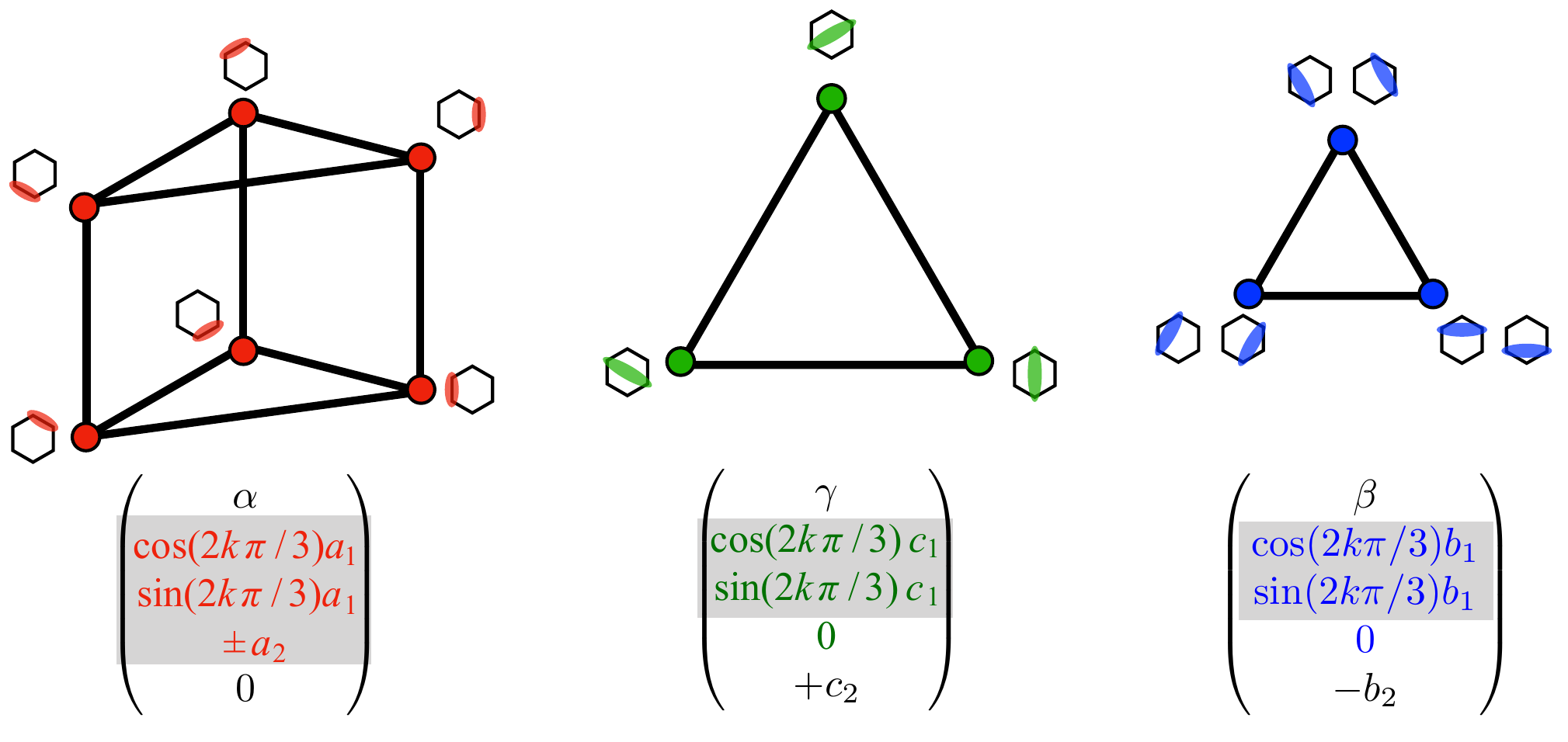}
    \end{minipage}
    \begin{minipage}[b]{0.32\linewidth} 
    \hspace{-10mm}
    \includegraphics[width=6cm]{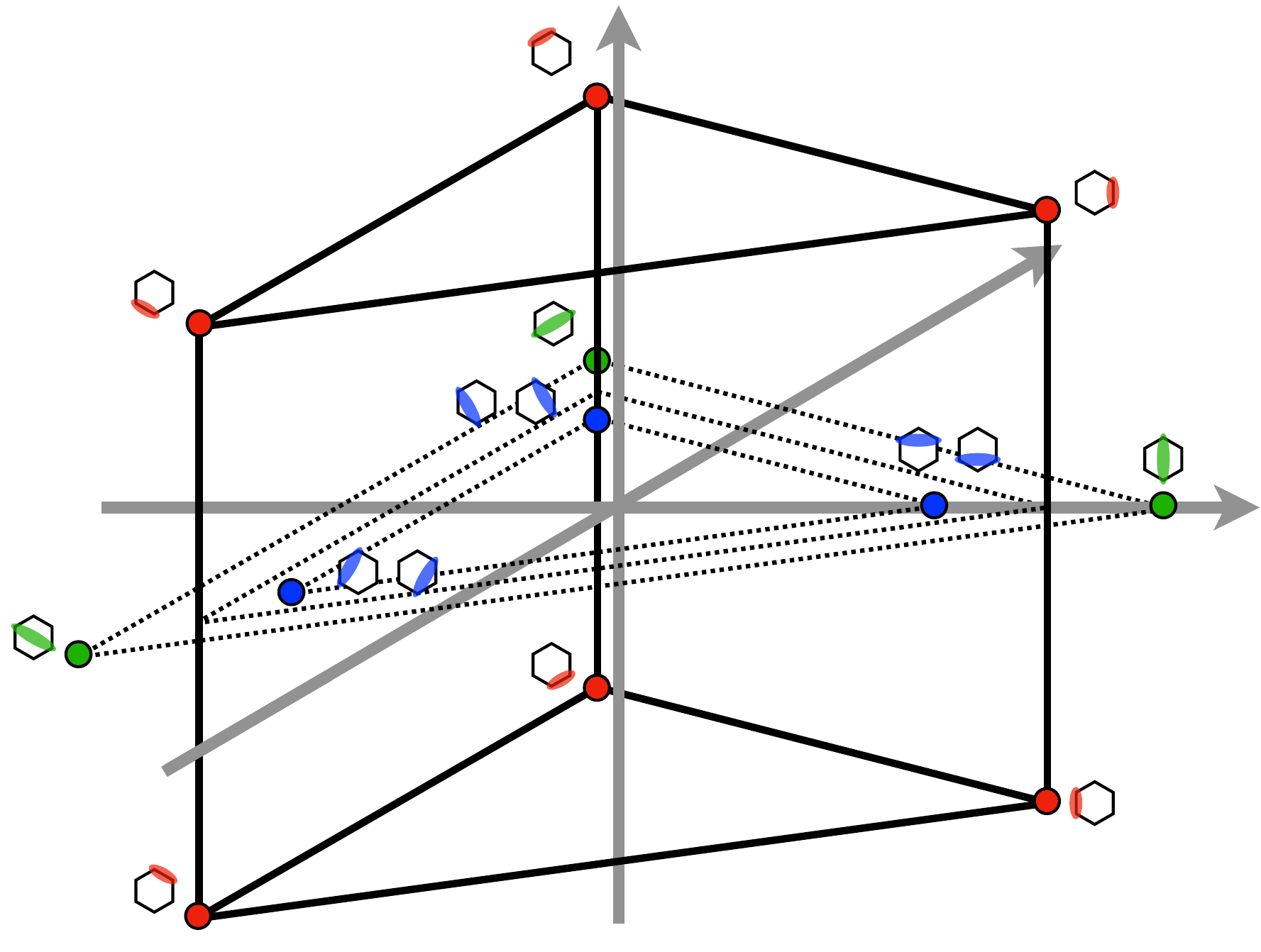}
    \end{minipage}
    \caption{The relation of all of the Gram vectors of $NPA_1(\mathcal{Proj})$ for the hexagon. The three different types of singlet projections (nearest neighbor, next nearest, and opposite position) each form either a triangle or a prism shape (left) which could be depicted as the right side figure when all are combined. 
    Although they form a 5-dimensional relation, we project them into lower dimensions for visualization. $k$ in the vector elements label the vertices with $k=0, 1,$ and $2$.}
    \label{fig:5dHex}
\end{figure*}

The Gram vectors corresponding to all $1+{6 \choose 2}=16$ rows/columns of the moment matrix forms a 5-dimensional relation, which is shown in Fig. \ref{fig:5dHex}. 
When we break it down, the Gram vectors for $h_{i,i+1}, h_{i,i+2}, \text{ and } h_{i,i+3}$ form 3, 2, and 2-dimensional relation respectively as shown on the left side of the figure.
Note that the Gram vectors for $h_{i,i+2}$ and $h_{i+3,i+4}$ are exactly the same. 
All three kinds of Gram vectors combine to form a 5-dimensional relation, but we only show the three-dimensional subspace in the figure (right) for visualization. 

The actual values in the coordinates of the Gram vectors are 
\begin{eqnarray*}
    \alpha = \frac{5+\sqrt{13}}{12}\approx0.717,  &
    \beta = \frac{1}{4}(1-3\phi)\approx 0.042, & \\
    \gamma = \frac{2}{3} (1-2\phi)\approx 0.482,~  &
    a_1 = \frac{1}{6}\sqrt{\frac{1}{2}(5+\phi)}\approx0.271 &\\
    b_1 = \sqrt{\frac{1}{8}(1-3\phi)}\approx 0.145,  &
    c_1 = \frac{1}{3}\sqrt{1+2\phi}\approx 0.416,& \\
    a_2 = \sqrt{\frac{1}{12}(1+2\phi)}\approx0.360, &
    b_2 = \phi/2\approx 0.139,  &\\
    c_2 = \phi \approx 0.278, \nonumber
\end{eqnarray*}
where $\phi=1/\sqrt{13}$. 

The basis we chose here is the most simplest, which intuitively could be understood in the following way. We choose the Gram vector for the ground state (identity in the moment matrix) to be $(1,0,0,0,0)$ without loss of generality and for simplicity. 
Then the first coordinate for all of the remaining Gram vectors ($\alpha, \beta$ and $\gamma$) simply correspond to the the expectation value of each kind of projectors. The next three coordinates correspond to the nontrivial prism-shape the $\{h_{i,i+1}\}$s form, and the other Gram vectors happen to follow the same triangular pattern, but without the height dimension. 
The last coordinate $b_2$ and $c_2$ could be seen as an additional constant term in order to ensure the normalization $h_{ij}^2=h_{ij}$, i.e., each vector must square to become $\alpha, \beta$, and $\gamma$. 

Interestingly, the hexagon Hamiltonian admits a decomposition that was used heavily in this study: 
\begin{equation}
    H=\left(\frac{1}{2}h_{12}+h_{23}+h_{34}+\frac{1}{2}h_{45}\right)+
    \left(\frac{1}{2}h_{45}+h_{56}+h_{61}+\frac{1}{2}h_{12}\right),
\end{equation}and
\begin{equation}
    \lVert H\rVert=\Big\lVert\frac{1}{2}h_{12}+h_{23}+h_{34}+\frac{1}{2}h_{45}\Big\rVert+
    \Big\lVert\frac{1}{2}h_{45}+h_{56}+h_{61}+\frac{1}{2}h_{12}\Big\rVert 
\end{equation}
simultaneously, {\it however}, $NPA_1(\mathcal{Proj})$ fails for the decomposed sub-hamiltonians, which is only observed for this particular case. 
It is quite likely that the $\SoS$ for the hexagon becomes severely more complicated than any other $\SoS$ we provided in this work.

\bibliography{refs}
\end{document}